\documentclass{CSML}

\def\dOi{11(4:12)2015}
\lmcsheading%
{\dOi}
{1--76}
{}
{}
{Jan.~16, 2015}
{Dec.~16, 2015}
{}

\ACMCCS{[{\bf Theory of computation}]: Semantics and
  reasoning---Program constructs---Type structures\,/\,Object oriented
  constructs; Models of computation---Concurrency---Parallel computing
  models; Semantics and reasoning--- Program semantics--- Operational
  semantics}


\usepackage[pdftex, setpagesize=false,
pdfborder=0 0 0, pdfdisplaydoctitle, pdflang={en},
pdftitle={Modular Session Types for Objects},
pdfauthor={Simon J. Gay, Nils Gesbert, Antonio Ravara, Vasco T. Vasconcelos},
pdfsubject={Languages, Theory, Verification},
pdfkeywords={Non-uniform method availability, object-oriented calculus,
  session types, typestate}]
{hyperref}
\usepackage[latin1]{inputenc}
\usepackage[T1]{fontenc}
\usepackage{amsmath}
\usepackage{amssymb}
\usepackage{stmaryrd}
\usepackage{bussproofs}
\usepackage{subfigure}
\usepackage{listings}
\usepackage{xspace}
\usepackage{tikz} 
\usepackage{mathbbol}

\newlength\digitwidth
\footnotesize
\let\LSTINLINE\lstinline
\def\lstinline{\LSTINLINE[basicstyle=\sffamily\small\upshape]}

\renewcommand*\thelstnumber{\ifnum\value{lstnumber}<10\relax\kern\digitwidth\fi\arabic{lstnumber}}

\makeatletter
\renewenvironment{figure}
  {\@float{figure}\footnotesize}
  {\end@float}
\makeatother

\let\subset\subseteq
\let\phi\varphi

\let\variables V

\def\fcmp{\ \textbf,\ }

\newcommand*{\children}[1]{\mkterm{children}_{#1}}
\newcommand*{\desc}[1]{\mkterm{desc}_{#1}}
\newcommand\roots{\mkterm{roots}}
\newcommand\chans{\mkterm{chans}}
\newcommand\objs{\mkterm{objs}}

\newcommand{\ea}{et al.\xspace}

\newcommand{\ie}{i.e.\xspace}

\newcommand{\Int}{\mkterm{Int}}

\newcommand\Bool{\mkterm{Bool}}

\newcommand*{\hadd}[2]{#1,\{#2\}}
\newcommand{\bnf}{\;::=\;}
\newcommand{\alt}{\;\;|\;\;}
\newcommand\classterm{\mkterm{class}}
\newcommand\reqterm{\mkterm{req}}
\newcommand\ensterm{\mkterm{ens}}
\newcommand\forterm{\mkterm{for}}
\newcommand*{\class}[4]{\mkterm{class}~{#1}~\{#2;#3;#4\}}

\newcommand*{\method}[4]{{#2}(#3)~\{#4\}}
\newcommand*{\annotmethod}[6]{\reqterm~{#1}~\ensterm~{#2}~\forterm~{#3}~{#4}(#5)~\{#6\}}

\usepackage{color}

\newcommand*{\mkterm}[1]{\mathsf{#1}}
\newcommand\nullterm{\mkterm{null}}
\newcommand\this{\mkterm{this}}
\mathchardef\mhyphen="2D
\newcommand\linkthis{{\mkterm{variant\mhyphen{}tag}}}
\newcommand\sessterm{\mkterm{session}}
\newcommand{\returnterm}{\mkterm{return}}

\newcommand\fieldsterm{\mkterm{fields}}
\newcommand{\nulltype}{\mkterm{Null}}

\newcommand{\End}{\mkterm{end}}

\newcommand*{\rectype}[2]{\mu{#1}.{#2}}

\newcommand*{\objecttype}[2]{{#1}[#2]}
\newcommand*{\methcal}[3]{{#1}.{#2}({#3})}
\newcommand*{\methsign}[3]{{#3}\;{#1}({#2})}
\newcommand*{\selfcal}[2]{{#1}({#2})}
\newcommand*{\seq}[2]{{#1};{#2}}

\newcommand*{\swap}[2]{{#1}\leftrightarrow{#2}}
\newcommand*{\new}[1]{\mkterm{new}~{#1}()}
\newcommand{\reduces}{\longrightarrow}
\newcommand{\vreduces}{\reduces}
\newcommand*{\subs}[2]{\{^{#1}\!{/}_{#2}\}}
\newcommand*{\state}[3]{({#1}*{#2}\fcmp{#3})}

\newcommand*{\branch}[3]{\{#1:#2\}_{#3}}
\newcommand*\choice[3]{\langle #1:#2\rangle_{#3}}
\newcommand*{\vfield}[3]{\langle #1:#2\rangle_{#3}}
\newcommand*\iset[2]{\{#1\}_{#2}}

\newcommand{\switchterm}{\mkterm{switch}}
\newcommand*{\switch}[4]{\switchterm~({#1})~\{#2:#3\}_{#4}}

\newcommand{\whileterm}{\mkterm{while}}

\newcommand*{\while}[2]{\whileterm~({#1})~\{#2\}}

\newcommand\true{\mkterm{True}}
\newcommand\false{\mkterm{False}}

\newcommand*{\hentry}[2]{{#1}[#2]}

\newcommand\linkterm{\mkterm{link}}
\newcommand*{\linktype}[2]{\linkterm~{#2}}

\newcommand*\return[2]{\mkterm{return}~#1}

\newcommand*{\typedheap}[2]{\ensuremath{\vdash #2 : #1}}
\newcommand*{\typedheapchan}[3]{\ensuremath{#1\vdash #3 : #2}}
\newcommand*{\typedexp}[4]{{#1}\vartriangleright{#2} : {#3}\vartriangleleft{#4}}
\newcommand*{\typedexpc}[4]{{#1}\vartriangleright_{C}{#2} : {#3}\vartriangleleft{#4}}
\def\judgment#1>#2:#3<#4/{\ensuremath{\typedexp{#1}{#2}{#3}{#4}}}
\let\judgement\judgment
\def\subjudgment#1>#2:#3<#4/{\ensuremath{{#1}\triangleright{#2}\subt{#3}\triangleleft{#4}}}
\def\judgmentc#1>#2:#3<#4/{\ensuremath{\typedexpc{#1}{#2}{#3}{#4}}}

\newcommand*\typedsess[3]{\ensuremath{#1\vdash #2 : #3}}
\newcommand{\typedmeth}[2]{\vdash_{#1}#2}
\newcommand{\typingRuleSkip}{2ex}

\newcommand*\changetype[3]{{#1}\{#2\mapsto #3\}}
\newcommand*\changeval[3]{{#1}\{#2\mapsto #3\}}

\newcommand{\subt}{\mathrel{\mathsf{<:}}}
\newcommand\supt{\mathrel{\mathsf{:>}}}

\newcommand*\rulename[1]{\LeftLabel{\scshape (#1)}}
\newcommand*\rrulename[1]{\RightLabel{\scshape (#1)}}
\newcommand*\axiomname[1]{{\scshape (#1)~}}

\newcommand*\wideexclam{{\mskip 2mu !\mskip 1.5mu}}
\newcommand*\widequestion{{?}}
\newcommand*\send[1]{\mathop{\wideexclam}\left[#1\right]}
\newcommand*\rcv[1]{\mathop{\widequestion}\left[#1\right]}
\newcommand*\chanbranch[1]{\mathop{\&}\left\{#1\right\}}
\newcommand*\chanchoice[1]{\mathop{\oplus}\left\{#1\right\}}
\newcommand*\chanseq{\mathbin{.}}

\newcommand*{\dl}{\mathit{dual}}

\newcommand*{\access}[1]{\langle{#1}\rangle}

\newcommand{\requestterm}{\mkterm{request}}
\newcommand{\acceptterm}{\mkterm{accept}}
\newcommand*\request[1]{\methcal{#1}{\requestterm}{}}
\newcommand*\accept[1]{\methcal{#1}{\acceptterm}{}}
\newcommand*\sendterm{\mkterm{send}}
\newcommand*\rcvterm{\mkterm{receive}}
\newcommand{\spawnterm}{\mkterm{spawn}}
\newcommand*{\spawn}[1]{\spawnterm~#1}

\newcommand*\thread[3]{\state{#1}{#2}{#3}}
\newcommand*\parcomp{\mathbin{||}}

\newcommand*\typedconf[2]{\ensuremath{#1\vdash#2}}

\newcommand{\sem}[1]{\llbracket#1\rrbracket}
\newcommand{\tr}{\mathit{tr}}
\newcommand{\dom}{\mathit{dom}}
\newcommand{\labtrans}[3]{#1\stackrel{#2}{\longrightarrow}#3}
\newcommand{\labtransstar}[3]{#1\stackrel{#2}{\longrightarrow^*}#3}

\newcommand*{\envref}[2]{{#1}*{#2}}

\newcommand*{\triplet}[3]{({#1},\, {#2},\, {#3})}

\newcommand*{\alga}[4]{\mathcal{A}_{#1} ({#2},\, {#3},\, {#4})}
\newcommand*{\algb}[3]{\mathcal{B}_C ({#1},\, {#2},\, {#3})}

\newcommand{\algskip}{0.5ex}

\newcommand*{\comment}[1]{}

\newcommand*{\mkRrule}[1]{\textsc{R-#1}}

\newcommand*{\mkTrule}[1]{\textsc{T-#1}}

\newcommand*{\mkErule}[1]{\textsc{E-#1}}
\newcommand{\Eassoc}{\mkErule{Assoc}}
\newcommand{\Ecomm}{\mkErule{Comm}}
\newcommand{\Escope}{\mkErule{Scope}}

\newcommand{\Rswap}{\mkRrule{Swap}}
\newcommand{\Rcall}{\mkRrule{Call}}

\newcommand{\RcomBase}{\mkRrule{ComBase}}
\newcommand{\RcomObj}{\mkRrule{ComObj}}
\newcommand{\Rcontext}{\mkRrule{Context}}
\newcommand{\Rinit}{\mkRrule{Init}}
\newcommand{\Rnew}{\mkRrule{New}}

\newcommand{\Rseq}{\mkRrule{Seq}}
\newcommand{\Rspawn}{\mkRrule{Spawn}}

\newcommand{\Rswitch}{\mkRrule{Switch}}

\newcommand{\Rreturn}{\mkRrule{Return}}

\newcommand{\Tcall}{\mkTrule{Call}}
\newcommand{\Tchan}{\mkTrule{Chan}}
\newcommand{\Tname}{\mkTrule{Name}}
\newcommand{\Tclass}{\mkTrule{Class}}

\newcommand{\Tspawn}{\mkTrule{Spawn}}

\newcommand{\TlinVar}{\mkTrule{LinVar}}
\newcommand{\Tnew}{\mkTrule{New}}
\newcommand{\TnewChan}{\mkTrule{NewChan}}
\newcommand{\Tpar}{\mkTrule{Par}}

\newcommand{\Treturn}{\mkTrule{Return}}
\newcommand{\TselfCall}{\mkTrule{SelfCall}}
\newcommand{\Tseq}{\mkTrule{Seq}}
\newcommand{\Tstate}{\mkTrule{State}}
\newcommand{\Tsub}{\mkTrule{Sub}}
\newcommand{\TsubEnv}{\mkTrule{SubEnv}}
\newcommand{\Tswitch}{\mkTrule{Switch}}
\newcommand{\TswitchLink}{\mkTrule{SwitchLink}}
\newcommand{\Tthread}{\mkTrule{Thread}}
\newcommand{\Tvar}{\mkTrule{Var}}

\newcommand{\Tnull}{\mkTrule{Null}}
\newcommand{\Tlabel}{\mkTrule{Label}}
\newcommand{\Tswap}{\mkTrule{Swap}}
\newcommand{\TvarF}{\mkTrule{VarF}}
\newcommand{\TvarS}{\mkTrule{VarS}}
\newcommand{\Tref}{\mkTrule{Ref}}
\newcommand{\TannotMeth}{\mkTrule{AnnotMeth}}
\newcommand{\THempty}{\mkTrule{HEmpty}}
\newcommand{\THadd}{\mkTrule{HAdd}}
\newcommand{\THide}{\mkTrule{Hide}}

\newcommand\ignore[1]{}
\newenvironment{proofnoqed}{\proof}{\par}

\newcommand{\citeN}[1]{\cite{#1}}

\newcommand{\citeyear}[1]{\cite{#1}}

\newenvironment{longitem}{\begin{itemize}}{\end{itemize}}
\newtheorem{theorem}{Theorem}[section]

\newtheorem{corollary}[theorem]{Corollary}
\newtheorem{proposition}[theorem]{Proposition}
\newtheorem{lemma}[theorem]{Lemma}
\newtheorem{definition}[theorem]{Definition}
\renewcommand*\paragraph[1]
{\par\vskip 1.5em plus .5em minus 1em\noindent\textit{#1}~~\null}

\begin{document}
\title[Modular Session Types for Objects]{Modular Session Types for Objects}
\author[S.~J.~Gay]{Simon J.~Gay\rsuper a}
\address{{\lsuper a}University of Glasgow}
\email{Simon.Gay@glasgow.ac.uk}
\author[N.~Gesbert]{Nils Gesbert\rsuper b}
\address{{\lsuper b}Grenoble INP -- Ensimag}
\email{Nils.Gesbert@grenoble-inp.fr}
\author[A.~Ravara]{Ant\'{o}nio Ravara\rsuper c}
\address{{\lsuper c}Universidade Nova de Lisboa}
\email{aravara@fct.unl.pt}
\author[V.~T.~Vasconcelos]{Vasco T.~Vasconcelos\rsuper d}
\address{{\lsuper d}Universidade de Lisboa}
\email{vmvasconcelos@ciencias.ulisboa.pt}


\begin{abstract}
  Session types allow communication protocols to be specified
  type-theoretically so that protocol implementations can be verified
  by static type checking. We extend previous work on session types
  for distributed object-oriented languages in three ways. (1) We
  attach a session type to a class definition, to specify the possible
  sequences of method calls. (2) We allow a session type (protocol)
  implementation to be \emph{modularized}, i.e.\ partitioned into
  separately-callable methods. (3) We treat session-typed
  communication channels as objects, integrating their session types
  with the session types of classes. The result is an elegant
  unification of communication channels and their session types,
  distributed object-oriented programming, and a form of typestate
  supporting non-uniform objects, i.e.\ objects that dynamically
  change the set of available methods. We define syntax, operational
  semantics, a sound type system, and a sound and complete type
  checking algorithm for a small distributed class-based
  object-oriented language with structural subtyping. Static typing
  guarantees that both sequences of messages on channels, and
  sequences of method calls on objects, conform to type-theoretic
  specifications, thus ensuring type-safety. The language includes
  expected features of session types, such as delegation, and expected
  features of object-oriented programming, such as encapsulation of
  local state.  
\end{abstract}


\maketitle


\section{Introduction}
\label{sec:intro}
\comment{La derni\`{e}re chose qu'on trouve en faisant un
  ouvrage, est de savoir celle qu'il faut mettre la
  premi\`{e}re. Blaise Pascal} 
\comment{SG: updated 11.5.2011}
Computing infrastructure has become inherently concurrent and
distributed, from the internals of machines, with the generalisation of
many- and multi-core architectures, to data storage and sharing
solutions, with ``the cloud''. Both hardware and software systems are
now not only distributed but also collaborative and
communication-centred. Therefore, the precise specification of the
protocols governing the interactions, as well as the rigorous
verification of their correctness, are critical factors to ensure the
reliability of such infrastructures.

Software developers need to work with technologies that may provide
correctness guarantees and are well integrated with the tools usually
used. Since Java is one of the most widely used programming languages,
the incorporation 
of support to specify and implement correct software components and
their interaction protocols would be a step towards more reliable systems.

Behavioural types represent abstractly and concisely the interactive
conduct of software components. These kind of types are simple yet
expressive languages, that characterise the permitted interaction
within a distributed system. The key idea is that some aspects of
dynamic behaviour can be verified statically, at compile-time rather
than at run-time. In particular, when working with a programming
language equipped with a behavioural type system, one can statically
ensure that an implementation of a distributed protocol, specified
with types, is not only safe in the standard sense (``will not go
wrong''), but also that the sequences of interactions foreseen by the
protocol are realizable by its distributed implementation. Thorough
descriptions of the state-of-the-art of research on the topic have
been prepared by COST Action IC1201 (Behavioural Types for Reliable
Large-Scale Software Systems, ``BETTY'') \cite{betty:wg3,betty:wg1}.

The problem we address herein is the following: can one specify the
full interactive behaviour of a given protocol as a collection of
types and check that a Java implementation of that protocol realises
safely such behaviour?

The solution we present works like this: first, specify as a session
type (a particular idiom of behavioural types) the behaviour of each
party involved in the protocol; second, using each such term to type a
class, implement the protocol as a (distributed, using channel-based
communication where channels are objects) Java program; third, simply
compile the code and if the type checker accepts it, then the code is
safe and realises the protocol. Details follow.

Session types \cite{HondaK:lanptd,HondaK:intblt} allow communication
protocols to be specified type-theore\-ti\-cally, so that protocol
implementations can be verified by static type checking. The
underlying assumption is that we have a concurrent or distributed
system with bi-directional point-to-point communication
channels. These are implemented in Java by TCP/IP socket
connections. A session type describes the protocol that should be
followed on a particular channel; that is to say, it defines the
permitted sequences, types and directions of messages. For example,
the session type $S = \send{\Int}\chanseq\rcv{\Bool}\chanseq\End$
specifies that an integer must be sent and then a boolean must be
received, and there is no further communication. More generally,
branching and repetition can also be specified. A session type can be
regarded as a finite-state automaton whose transitions are annotated
with types and directions, and whose language defines the protocol.

Session types were originally formulated for languages closely based
on process calculus. Since then, the idea has been applied to
functional languages
\cite{GaySJ:sestip,GaySJ:lintta,NeubauerM:impst,PucellaR:hassta,VasconcelosVT:typmfl},
the process-oriented programming language Erlang
\cite{MostrousVasconcelos:featherweightErlang},
component-based object systems \cite{VallecilloA:typbsc},
object-oriented languages
\cite{CapecchiS:amasmo,Dezani-CiancagliniM:bousto,Dezani-CiancagliniMetal:ost,Dezani-CiancagliniM:sestoo,Dezani-CiancagliniM:disool,HuR:sesbdp},
operating system services \cite{FahndrichM:lansfr},
and more general servi\-ce-oriented systems \cite{CarboneM:strgpc,cruz-filipe.lanese.etal:stream-based-service-centered-calculus}.
Session types have also been generalised from two-party to multi-party
systems \cite{BonelliE:mulstd,HondaK:mulast}, although in the present
paper we will only consider the two-party case.



In previous work \cite{GaySJ:modstd} we proposed a new approach to
combining session-typed communication channels and distributed
object-oriented programming. Our approach extends earlier work and
allows increased programming flexibility. We adopted the principle
that it should be possible to store a channel in a field of an object
and allow the object's methods to use the field like any other; we
then followed the consequences of this idea. For example, consider a
field containing a channel of type $S$ above, and suppose that method
\lstinline|m| sends the integer and method \lstinline|n| 
receives the boolean. Because the session type of the channel requires
the send to occur first, it follows that \lstinline|m| must be called
before~\lstinline|n|. We therefore need to work with \emph{non-uniform
  objects}, in which the availability of methods depends on the state
of the object: method \lstinline|n| is not available until after
method \lstinline|m| has been called. In order to develop a static
type system for object-oriented programming with session-typed
channels, we use a form of typestate \cite{typestates}
that we have previously presented under the name of \emph{dynamic
  interfaces}~\cite{GaySJ:dyni}.  In this type system, the
availability of a class's methods (\ie, the possible sequences of
method calls) is specified in a style that itself resembles a form of
session type, giving a pleasing commonality of notation at both the
channel and class levels.

The result of this combination of ideas is a language that allows a
natural integration of programming with session-based channels
and with non-uniform objects. In particular, the implementation of a
session can be \emph{modularized} by dividing it into separate methods
that can be called in turn. This is not possible in SJ
\cite{HuR:sesbdp}, the most closely related approach to combining
sessions and objects (we discuss related work thoroughly in
Section~\ref{sec:related}). We believe that we have achieved a smooth
and elegant combination of three important high-level abstractions:
the object-oriented abstraction for structuring computation and data,
the typestate abstraction for structuring state-dependent method
availability, and the session abstraction for structuring
communication.



\paragraph{Contributions}
In the present paper we formalize a core \emph{distributed class-based
  object-oriented} language with a static type system that combines
session-typed channels and a form of typestate. The language is
intended to model programming with TCP/IP sockets in Java. The formal
language differs from that introduced in our previous work
\cite{GaySJ:modstd} by using structural rather than nominal
types. This allows several simplifications of the type system. We have
also simplified the semantics, and revised and extended the
presentation.  We prove that static typing guarantees two runtime
safety properties: first, that the sequence of method calls on every
non-uniform object follows the specification of its class's session
type; second, as a consequence (because channel operations are
implemented as method calls), that the sequence of messages on every
channel follows the specification of its session type. This paper
includes full statements and proofs of type safety, in contrast to the
abbreviated presentation in our conference paper. We also formalize a
type checking algorithm and prove its correctness, again with a
revised and expanded presentation in comparison with the conference
paper.

There is a substantial literature of related work, which we discuss in
detail in Section~\ref{sec:related}. Very briefly, the contributions
of our paper are the following.
\begin{longitem}
\item In contrast to other work on session types for object-oriented
  languages, we do not require a channel to be created and completely
  used (or delegated) within a single method. Several methods can
  operate on the same channel, thus allowing effective
  encapsulation of channels in objects, while retaining the usual
  object-oriented development practice. This is made possible by our
  integration of channels and non-uniform objects. This contribution
  was the main motivation for our work.
\item In contrast to other typestate systems, we use a global
  specification of method availability, inspired by session types,
  as part of a class definition. This replaces pre- and post-condition
  annotations on method definitions, except in the particular case of
  recursive methods.  
\item When an object's typestate depends on the \emph{result} (in an
  enumerated type) of a method call, meaning that the result must be
  case-analyzed before using the object further, we do not force the
  case-analysis to be done immediately by using a combined
  ``switch-call'' primitive. Instead, the method result can
  be stored in a field and the case-analysis can happen at any
  subsequent point. Although this feature significantly increases the
  complexity of the formal system and could be omitted for simplicity,
  it supports a natural programming style and gives more options to
  future programming language designers.
\item Our structural definition of subtyping provides a flexible treatment of
  relationships between typestates, which can also support inheritance; this is discussed further in Section~\ref{sec:nominal}.
\end{longitem}

\noindent The remainder of the paper is structured as follows. In
Section~\ref{sec:seq-example} we illustrate the concept of dynamic
interfaces by means of a sequential example. In
Section~\ref{sec:core-lang} we formalize a core sequential language
and in Section~\ref{sec:seq-extensions} we describe some extensions.
In Section~\ref{sec:distributed-example} we extend the sequential example to a
distributed setting and in Section~\ref{sec:distributed} we extend the
formal language to a core distributed language.
In Section~\ref{sec:results} we state and prove the key properties of
the type system. In Section~\ref{sec:algorithm} we present a
type checking algorithm and prove its
soundness and completeness, and describe a prototype implementation of
a programming language based on the ideas of the paper. 
Section~\ref{sec:related} contains a more
extensive discussion of related work; Section~\ref{sec:conclusion}
outlines future work and concludes.



\section{A Sequential Example}
\label{sec:seq-example}

\comment{SG: updated 31.3.11}

\begin{figure}
\begin{lstlisting}
class File {
  session Init 
  where Init  = { {OK, ERROR} open(String): <OK: Open, ERROR: Init>}
        Open  = { {TRUE, FALSE} hasNext(): <TRUE: Read, FALSE: Close>, 
                  Null close(): Init}
        Read  = {String read(): Open, Null close(): Init}
        Close = {Null close(): Init}
  open(filename) {...}
  hasNext() {...}
  read() {...}
  close() {...}
}
\end{lstlisting}
\caption{A class describing a file in some API}
\label{fig:file}
\end{figure}




A file is a natural example of an object for which the availability of
its operations depends on its state. The file must first be opened,
then it can be read repeatedly, and finally it must be closed. Before
reading from the file, a test must be carried out in order to
determine whether or not any data is present. The file can be closed
at any time.

There is a variety of terminology for objects of this kind.
Ravara and Vasconcelos \citeN{ravara.vasconcelos:typco} refer to them as
\emph{non-uniform}. We have previously used the term \emph{dynamic
  interface} \cite{GaySJ:dyni} to indicate that the interface, i.e.\
the set of available operations, changes with time. The term
\emph{typestate} \cite{typestates} is also well established.

Figure~\ref{fig:file} defines the class \lstinline|File|, which we
imagine to be part of an API for using a file system. The definition
does not include method bodies, as these would typically be
implemented natively by the file system. What it does contain is
method signatures and, crucially, a session type definition which
specifies the availability of methods. We will refer to a skeleton
class definition of this kind as an \emph{interface}, using the term
informally to mean that method definitions are omitted.
The figure shows the method signatures in the session type and not as
part of the method definitions as is normal in many programming
languages. This style is closer to the formal language that we define
in Section~\ref{sec:core-lang}.


Line 3 declares the initial session type \lstinline|Init| for the
class. This and other session types are defined on lines 4--7. We will
explain them in detail; they are types of objects, indicating which
methods are available at a given point and which is the type after
calling a method.  In a session type, the constructor
\lstinline|{...}|\hspace{-1ex}, 
which we call \emph{branch}, indicates that certain methods are
available. In this example, \lstinline|Init| declares the availability
of one method (\lstinline|open|), states \lstinline|Open| and
\lstinline|Read| allow for two methods each, and state
\lstinline|Close| for a single method (\lstinline|close|). For
technical convenience, the presence of data is tested by calling the
method \lstinline|hasNext|, in the style of a Java iterator, rather
than by calling an \lstinline|endOfFile| method. If desired, method
\lstinline|hasNext| could also be included in state \lstinline|Read|.

The constructor \lstinline|<...>|, which we call \emph{variant},
indicates that a method returns a value from an enumeration, and that
the subsequent type depends on the result. For example, from state
\lstinline|Init| the only available method is \lstinline|open|, and it
returns a value from an enumeration comprising the constants (or
labels) \lstinline|OK| and \lstinline|ERROR|. If the result is
\lstinline|OK| then the next state is \lstinline|Open|; if the result
is \lstinline|ERROR| then the state remains \lstinline|Init|.  It is
also possible for a session type to be the empty set of methods,
meaning that no methods are available; this feature is not used in the
present example, but would indicate the end of an object's useful
life.

\newcommand{\scl}{0.80} 
\newcommand{\sessionDiagram}{
\begin{tikzpicture}
[level distance=3.5em, sibling distance=3.5em, color=black, thick, scale=\scl]
\node[circle,shade,ball color=gray!30] (init node) {I} 
  child {node[circle,shade,ball color=gray!30] (plus node) {}
   child {node[circle,shade,ball color=gray!30] (open node) {O} 
      child {node[circle,shade,ball color=gray!30] {}
       child {node[circle,shade,ball color=gray!30] (read node) {R} 
        edge from parent[->] node[left] {\textsf{\tiny TRUE}}
        }
        child {node[circle,shade,ball color=gray!30] (close node) {C} 
          edge from parent[->] node[right] {\textsf{\tiny FALSE}}
        }
        edge from parent[->] node[right] {\textsf{hasNext()}}
      }
      edge from parent[->] node[left] {\textsf{\tiny OK}}
    }
  edge from parent[->] node[left] {\textsf{open()}}
  };
\draw[->] (open node.east) .. controls +(1:7em) and +(-15:15em)
.. node[left] {\textsf{close()}} (init node.north); 
\draw[->] (read node.west) .. controls +(1:-3em) and +(10:-3em)
.. node[left] {\textsf{read()}} (open node.west); 
\draw[->] (read node.south) .. controls +(50:-10em) and +(10:-20em)
.. node[left] {\textsf{close()}} (init node.west); 
\draw[->] (plus node.east) .. controls +(1:2em) and +(10:2em)
.. node[right] {\textsf{\tiny ERROR}} (init node.east); 
\draw[->] (close node.east) .. controls +(1:7em) and +(30:20em)
.. node[right] {\textsf{close()}} (init node.north); 
\end{tikzpicture} 
}

\begin{figure}
\vspace{-6em}
\begin{center}
      \sessionDiagram
\end{center}
\vspace{-5em}
\caption{Diagrammatic representation of the session type of
    class \lstinline|File| in Figure~\ref{fig:file}}
  \label{fig:session-diagram}
\end{figure}


The session type can be regarded as a finite state automaton whose
transitions correspond to method calls and results. This is
illustrated in Figure~\ref{fig:session-diagram}. Notice the two types
of nodes (\lstinline|{...}|\hspace{-1ex} and \lstinline|<...>|) and
the two types of labels in arcs (method names issuing from
\lstinline|{...}|\hspace{-1ex} nodes and enumeration constants issuing
from \lstinline|<...>| nodes).


Our language does not include constructor methods as a special
category, but the method \lstinline|open| must be called first and can
therefore be regarded as doing initialisation that might be included
in a constructor. Notice that \lstinline|open| has the filename as a
parameter. Unlike a typical file system API, creating an object of class
\lstinline|File| does not associate it with a particular file; instead this
happens when \lstinline|open| is called.

The reader might expect a declaration
\lstinline|void close()| rather than \lstinline|Null close()|; for
simplicity, we do not address procedures in this paper, instead working
with the type \lstinline|Null| inhabited by a single value,
\lstinline|null|.
Methods \lstinline|open| and \lstinline|hasNext| return a constant
from an enumeration: \lstinline|OK| or \lstinline|ERROR| for method
\lstinline|open|, and \lstinline|TRUE| or \lstinline|FALSE| for
method \lstinline|hasNext|. Enumerations are simply sets of labels,
and do not need to be declared with names.


\begin{figure}
\begin{lstlisting}
class FileReader {
  session Init
  where Init = {Null init (): {Null read (String): Final}}
        Final = {String toString(): Final}

  file; text;

  init() {
    file = new File(); 
    text = ""; // Evaluates to null
  }
  read(filename) {
    switch (file.open(filename)) {
      case ERROR:
        null;
      case OK:
        while (file.hasNext())
          text = text +++ file.read();
        file.close(); // Returns null
    } 
  }
  toString() { text; }
}
\end{lstlisting}
\caption{A client
  that reads from a \lstinline|File|}
\label{fig:filereader}
\end{figure}


Figure~\ref{fig:filereader} defines the class \lstinline|FileReader|,
which uses an object of class \lstinline|File|.
\lstinline|FileReader| has a session type of its own, defined on lines
2--3. It specifies that methods must be called in the sequence
\lstinline|init|, \lstinline|read|, \lstinline|toString|,
\lstinline|toString|, \dots.
Line 5 defines the fields of \lstinline|FileReader|. The formal
language does not require a type declaration for fields, since fields
always start with type \lstinline|Null|, and are initialised to value
\lstinline|null|.  Fields are always \emph{private} to a class, even
if we do not use a corresponding keyword. Lines 7--10 define the
method \lstinline|init|, which has initialisation behaviour typical of
a constructor.
Lines 12--19 illustrate the \lstinline|switch| construct. In this
particular case the \lstinline|switch| is on the result of a method
call. One of the distinctive features of our language is that it is
possible, instead, to store the result of the method call in a field
and later \lstinline|switch| on the field; we will explain this in
detail later. This contrasts with, for example, 
\emph{Sing\#}~\cite{FahndrichM:lansfr}, in which the
call/\lstinline|switch| idiom is the only possibility. The
\lstinline|while| loop (lines 16--17) is similar in the sense that the
result of \lstinline|file.hasNext| must be tested in order to find
out whether the loop can continue, calling \lstinline|file.read|, or
must terminate.  Line 21 defines the method \lstinline|toString| which
simply accesses a field.

Typechecking the class \lstinline{FileReader} according to our type
system detects many common mistakes. Each of the following code
fragments contains a type error.
\begin{itemize}
\item
\begin{lstlisting}[numbers=none]
file.open(filename); 
file.read() ...;     
\end{lstlisting}
The \lstinline|open| method returns either \lstinline|OK| or
\lstinline|ERROR|, and the type of \lstinline|file| is the variant type 
\lstinline|<OK: Open, ERROR: Init>|. The tag for this variant type
is the result of \lstinline|open|. Because the type of
\lstinline|file| is a variant, a method cannot be called on it; first
we must use a \lstinline|switch| statement to analyse the result of
\lstinline|open| and discover which part of the variant we are in.

\item
\begin{lstlisting}[numbers=none]
switch(file.open(filename))    
  case OK: text = file.read(); 
\end{lstlisting}
Here, a \lstinline|switch| is correctly used to find out whether or
not the file was successfully opened. However, if \lstinline|file| is
in state \lstinline|Open|, the \lstinline|read| method cannot be
called immediately. First, \lstinline|hasNext| must be called, with a
corresponding \lstinline|switch|.

\item
\begin{lstlisting}[numbers=none]
result = file.open(filename);
...
switch(result)
  case ERROR: file.close();
\end{lstlisting}
In state \lstinline|ERROR|, method \lstinline|close| is not available
(because the file was not opened successfully). The only available
method is \lstinline|open|.

\item
\begin{lstlisting}[numbers=none]
file.close();
if (file.hasNext()) ... 
\end{lstlisting}
After calling \lstinline|close|, the file is in state
\lstinline|Init|, so the method \lstinline|hasNext| is not
available. Only \lstinline|open| is available.
\end{itemize}
Clearly, correctness of the code in Figure~\ref{fig:filereader}
requires that the sequence of method calls on field \lstinline|file|
within class \lstinline|FileReader| matches the available methods in
the session type of class \lstinline|File|, and that the appropriate
\lstinline|switch| or \lstinline|while| loops are performed when
prescribed by session types of the form \lstinline|<...>| in class
\lstinline|File|.
Our static type system, defined in Section~\ref{sec:core-lang}, enables
this consistency to be checked at compile-time.
A distinctive feature of our type system is that methods are checked
in a precise order: that prescribed by the session type
(\lstinline{init, read, toString} in class \lstinline{FileReader},
Figure~\ref{fig:filereader}). As such the type of the private
reference \lstinline{file} always has the right type (and no further
annotations---pre/post conditions---are required when in presence of
non-recursive methods).
Also, in order to check statically that an object with a dynamic
interface such as \lstinline|file| is used correctly, our type system
treats the reference linearly so that aliases to it cannot be created.
This restriction is not a problem for a simple example such as this
one, but there is a considerable literature devoted to more flexible
approaches to unique ownership. We discuss this issue further in
Sections~\ref{subsec:shared-types}, \ref{sec:related}
and~\ref{sec:conclusion}.

\begin{figure}
\begin{lstlisting}[numbers=none]
class FileReadToEnd {
  session Init 
  where Init  = {{OK,ERROR} open(String): <OK: Open, ERROR: {Init}>}
        Open  = {{TRUE,FALSE} hasNext(): <TRUE: Read, FALSE: Close>}
        Read  = {String read(): Open}
        Close = {Null close(): {Init}}
}
\end{lstlisting}
\caption{Interface for class \lstinline|FileReadToEnd|}
\label{fig:filereader-api}
\end{figure}




In order to support \emph{separate compilation} we require only the
interface of a class, including the class name and the session type
(which in turn includes the signature of each method).
%
%
For example, in order to typecheck classes that are clients of
\lstinline|FileReader|, we only need its interface. Similarly, to
typecheck class \lstinline|FileReader|, which is a client of
\lstinline|File|, it suffices to use the interface for class
\lstinline|File|, thus effectively supporting typing clients of
classes containing \emph{native methods}.

Figure~\ref{fig:filereader-api} defines the interface for a class
\lstinline|FileReadToEnd|. This class has the same method definitions
as \lstinline|File|, but the \lstinline|close| method is not available
until all of the data has been read. According to the subtyping
relation defined in Section~\ref{sec:subtyping}, type
\lstinline|Init| of \lstinline|File| is a subtype of type
\lstinline|Init| of \lstinline|FileReadToEnd|, which we express as
\lstinline|File.Init| $\subt$
\lstinline|FileReadToEnd.Init|. Subtyping guarantees safe
substitution: an object of type \lstinline|File.Init| can be used
whenever an object of type \lstinline|FileReadToEnd.Init| is expected,
by forgetting that \lstinline|close| is available in more states. As
it happens, \lstinline|FileReader| reads all of the data from its
\lstinline|File| object and could use a \lstinline|FileReadToEnd|
instead.


\section{A Core Sequential Language}
\label{sec:core-lang}

\comment{NG: updated 4.9.2014}

We now present the formal syntax, operational semantics, and type
system of a core sequential language. As usual, the formal language
makes a number of simplifications with respect to the more practical
syntax used in the examples in Section~\ref{sec:seq-example}. We
summarise below the main differences with what was discussed in the
previous section; in Section~\ref{sec:seq-extensions}, we will discuss
in more detail how some usual programming idioms, which would be expected in
a full programming language, can be encoded into this formal core.
\begin{longitem}
\item Every method has exactly one parameter. This does not affect
  expressivity, as multiple parameters can be passed within an object,
  and a dummy parameter can be added if necessary: we consider a
  method call of the form $\methcal{f}{m}{}$ as an abbreviation for
  $\methcal{f}{m}{\nullterm}$.
\item Field access and assignment are defined in terms of a swap
  operation $\swap f e$ which puts the value of $e$ into the field $f$
  and evaluates to the former content of $f$. This operation is
  formally convenient because our type system forbids aliasing. In
  Java, the expression $f = e$ computes the result of $e$ and then
  both puts it into $f$ and evaluates to it, allowing
  expressions such as $g = (f = e)$ which create aliases; reading a
  field without removing its content also allows
  creation of aliases. The swap operation is a combined read-write
  which does not permit aliasing.

  The normal assignment operation $f = e$ is an
  abbreviation for $\seq{\swap f e}{\nullterm}$ (where the sequence
  operator explicitly discards the former content of $f$) and field
  read as the standalone expression $f$ is an abbreviation for
  $\swap{f}{\nullterm}$. They differ from usual semantics by the fact
  that field read is destructive and that the assigment expression
  evaluates to $\nullterm$. 
\item In the examples, all method signatures
  appearing in a branch session type
  indicate both a return type and a subsequent session type. In
  general, those types are two separate things. However,
  when the subsequent behaviour of the object depends on the returned
  value, like the case of \textsf{Open} in type
  \textsf{Init} on line 3 of Fig.~\ref{fig:file}, the return type is
  an enumerated set of labels and the subsequent session type is a
  variant which must provide cases for exactly these labels. To
  simplify definitions, we avoid this redundant specification in the
  formal language, and when the subsequent session type is a variant,
  the return type of the method is always the special type
  $\linkthis$, which indicates that the method will return a label
  from the variant.
\end{longitem}


\begin{figure}
  \begin{align*}
    \textrm{Class dec}&& D \bnf\ & \class{C}{S}{\vec f}{\vec M}
    \\
    \textrm{Class session types}&& S \bnf\ & 
    \branch{\methsign{m_i}{T_i}{T'_i}}{S_i}{i\in I} \alt \choice{l}{S_l}{l\in E}
    \alt X \alt \mu X.S
    \\
    \textrm{Method dec}&& M \bnf\ &
    \method{}{m}{x}{e}
    \\
    \textrm{Types}&& T \bnf\ & \nulltype\alt
    S \alt E \alt \linkthis
    \\
    \textrm{Label sets}&& E \bnf\ &  \{l,\ldots,l\} 
    \\
    \textrm{Values}&& v \bnf\ & \nullterm \alt l
    \\
    \textrm{Expressions}&& e \bnf\ & v \alt 
    \methcal{f}{m}{e} \alt x \alt \seq e e \alt \switch{e}{l}{e_l}{l\in E}
    \alt \swap{f}{e} \alt \new C
  \end{align*}
  %
  \caption{Top level syntax}
  \label{fig:syntax}
\end{figure}


\begin{figure}
  \begin{align*}
    \textrm{Types}&& T \bnf\ & \dots \alt \linktype E f \alt \objecttype C F
    \\
    \textrm{Field types}&& F \bnf\ & \iset{T_i\,f_i}{i\in I}
    \alt \vfield{l}{F_l}{l\in E} 
    \\
    \textrm{Values}&& v \bnf\ & \dots \alt o
    \\
    \textrm{Paths}&& r \bnf\ & o \alt r.f
    \\
    \textrm{Expressions}&& e \bnf\ &\dots \alt \return{e}{r}
    \\
    \textrm{Object records}&& R \bnf\ & \hentry{C}{\iset{f_i=v_i}{i\in I}}
    \\
    \textrm{Heaps}&& h \bnf\ &\varepsilon \alt
    \hadd{h}{o=R}
    \\
    \textrm{States}&& s \bnf\ &\state{h}{r}{e}
    \\
    \textrm{Contexts}&& \mathcal{E} \bnf\ & [\_] \alt
    \seq{\mathcal E}{e} \alt \switch{\mathcal E}{l}{e_l}{l\in E}
    \alt \return{\mathcal E}{r} \alt \swap{f}{\mathcal E}
    \alt \methcal{f}{m}{\mathcal E}
  \end{align*}
  The productions for types, values and expressions extend
  those in Figure~\ref{fig:syntax}. Session
  types may never contain types of the form $\linktype{}f$, even in
  the extended syntax.
  \caption{Extended syntax, used only in the type system and semantics}
  \label{fig:syntax-red}  
\end{figure}


\subsection{Syntax}
\label{sec:syntax}

We separate the syntax into the top-level language
(Figure~\ref{fig:syntax}) and the extensions required by the type
system and operational semantics (Figure~\ref{fig:syntax-red}).
Identifiers $C$, $m$, $f$ and $l$ are taken from disjoint countable
sets representing names of classes, methods, fields and labels
respectively. The vector arrow indicates a sequence of zero or more
elements of the syntactic class it is above. Similarly, constructs
indexed by a set denote a finite sequence. We use $E$ to
specifically denote finite sets of labels $l$, whereas $I$ is any
finite indexing set.

Field names always refer to fields of the current object; there is no
qualified field specification $o.f$.  In other words, all fields
are private.  Method call is only available on a field,
not an arbitrary expression. This is because calling a method changes
the session type of the object on which the method is called, and in
order for the type system to record this change, the object must be in
a specified location (field).

Conversely, there is no unqualified method call in the core language.
Calling a method on the current object $\this$, which we call a
\emph{self-call}, behaves differently from external calls with respect
to typing and will be discussed as an extension in
Section~\ref{subsec:self-calls},

A program consists of a sequence of class declarations $D$. In the core
language, types in a top-level program only occur in the
\emph{session} part of a class declaration: no type is declared for
fields because they can vary at run-time and are always initially
$\nullterm$, and method declarations are also typeless, as explained earlier.

A session type $S$ corresponds to a view of an object from outside. It
shows which methods can be called, and their signatures, but the
fields are not visible. We refer to
$\branch{\methsign{m_i}{T_i}{T'_i}}{S_i}{i\in I}$ as a \emph{branch}
type and to $\choice{l}{S_l}{l\in E}$ as a \emph{variant}
type. Session type $\End$ abbreviates the empty branch type
$\{\}$. The core language does not include named session types, or the
\lstinline|session| and \lstinline|where| clauses from the examples;
we just work with recursive session type expressions of the form
$\rectype{X}{S}$, which are required to be \emph{contractive},
i.e.\ containing no subexpression of the form
$\rectype{X_1}{\cdots\rectype{X_n}{X_1}}$. We require contractivity so
that every session type defines some behaviour. The $\mu$ operator is
a binder, giving rise, in the standard way, to notions of bound and
free variables and alpha-equivalence. A type is \emph{closed} if it
includes no free variables. We denote by $T\subs{U}{X}$ the
capture-avoiding substitution of $U$ for $X$ in~$T$.

Value types which can occur either as parameter or return type for a
method are: $\nulltype$ which has the single value $\nullterm$, a
session type $S$ which is the type of an object, or an enumerated type
$E$ which is an arbitrary finite set of labels $l$. Additionally, the
specific return type $\linkthis$ is used for method occurrences after
which the resulting session type is a variant, and means that the
method result will be the tag of the variant.  The set of possible labels
appears in the variant construct of the session type, so it is not
necessary to specify it in the return type of the
method. However, in the example code, the set of labels is written
instead of $\linkthis$, so that the method signature shows the return
type in the usual way.

The type system, which we will describe later, enforces the following
restrictions on session types: the immediate components of a variant
type are always branch types, and the session type in a class
declaration is always a branch. This is because a variant type is used
only to represent the effect of method calls and the dependency
between a method result and the subsequent session type, so it only
makes sense immediately within a branch type. 


Figure~\ref{fig:syntax-red} defines additional syntax that is needed
for the formal system but is not used in top-level programs. This
includes:\begin{longitem}
\item some extra forms of types, which are used internally to type
some subexpressions but cannot be the argument type or return type of
a method, and thus are never written in a program;
\item intermediate expressions that cannot appear in a program but
  arise from the operational semantics;
\item syntax for the heap.
\end{longitem}
\paragraph{Internal types.} The first internal type we add is the type
$\linktype{}{f}$, where $f$ is the name of a field. This type is
related to variant session types and the $\linkthis$ type, in the way
illustrated by the following example: suppose that, in some context,
field $f$ of the current object contains an object whose type is
$\branch{\methsign{m}{\nulltype}{\linkthis}}{\choice{l}{S_l}{l\in
    E}}{}$.
This means that the expression $\methcal{f}{m}{\nullterm}$ is allowed
in this context and will both: change the abstract state of the object
in $f$ to one of the $S_l$, and return the label $l$ corresponding to
that particular state. Thus, there is a link between the value of the
expression and the type of field $f$ after evaluating the expression,
and the type system needs to keep track of this link; to this end, the
expression $\methcal{f}{m}{\nullterm}$ is given, internally, the type
$\linktype{}{f}$, rather than $\linkthis$ which does not contain
enough information. The use of $\linktype{}{f}$ is also illustrated in
Figure~\ref{fig:reduction-ex}.

The second internal type is an alternative form of object type,
$\objecttype C F$, which has a field typing instead of a session type.
Recall that in our language, all object fields are private; therefore,
normally the type of an object is a session type which only refers to
methods. However, an object has access to its own fields, so for
typechecking a method definition, the type environment needs to provide
types for the fields of the current object ($\this$). Since the types
of the fields change throughout the life of the object, the class
definition is not enough to know their types at a particular point. We
thus use a field typing $F$, which is usually a record type associating one
type to each field of the object. For example, $\objecttype C
{\{\nulltype\,f_1; S\,f_2\}}$ represents an object of class $C$ with
exactly two fields, $f_1$, which currently contains a value of type
$\nulltype$, and $f_2$ which currently contains an object in state
$S$. Note that in $\objecttype C F$, $F$ must provide types for
exactly all the fields of class $C$.

The other form of field typing is a variant field typing, which
regroups several possible sets of field types indexed by labels. For
example, $\objecttype C {\vfield{l_1}{\{\nulltype\,f_1; S\,f_2\}; l_2 :
  \{E\,f_1; S'\,f_2\}}{}}$ represents an object of class $C$, whose two
fields are $f_1$ and $f_2$, and where either $f_1$ has type
$\nulltype$ and $f_2$ has type $S$, or $f_1$ has type $E$ and $f_2$
has type $S'$, depending on the value of the label.

These field typings cannot be the type of expressions (which cannot
evaluate to $\this$); they only represent the type of the current
object in type environments. The relation between field typings and
session types will be discussed in Section~\ref{sec:consistency}.

\paragraph{Internal expressions, heap, and states.}
These other additions are used to define the operational semantics.  A
heap $h$ maps object identifiers $o$, taken from yet another countable
set of names, to object records $R$.  We write $\dom(h)$ for the
set of object identifiers in $h$. The identifiers are values,
which may occur in expressions. The operation $\hadd{h}{o=R}$
represents adding a record for identifier $o$ to the heap $h$ and we
consider it to be associative and commutative, that is, $h$ is
essentially an unordered set of bindings. It is only defined if
$o\not\in\dom(h)$. \emph{Paths} $r$ 
represent locations in the heap.  A path consists of a top-level object
identifier followed by an arbitrary number of field specifications.
We use the following notation to interpret paths relative to a given heap.
\begin{definition}[Heap locations]\hfil
\label{def:heaplocations}
  \begin{itemize}
  \item If $R = \hentry{C}{\iset{f_i=v_i}{i\in I}}$, we define $R.f_i
    = v_i$ (for all $i$) and $R.\classterm = C$. For any value $v$ and
    any $j \in I$, we also define $\changeval{R}{f_j}{v} =
    \hentry{C}{\iset{f_i=v'_i}{i\in I}}$ where $v'_i = v_i$ for $i\neq
    j$ and $v'_j = v$.
  \item If $h = (\hadd{h'}{o = R})$, we define $h(o) = R$,
    and for any field $f$ of $R$, $\changeval{h}{o.f}{v} =
    (\hadd{h'}{o = \changeval{R}{f}{v}})$.
  \item If $r = r'.f$ and $h(r').f = o$, then we also define $h(r) =
    h(o)$ and $\changeval{h}{r.f'}{v} = \changeval{h}{o.f'}{v}$.
  \item In any other case, these operations are not defined. Note in
    particular that $h(r)$ is not defined if $r$ is a path that exists
    in $h$ but does not point to an object identifier.
  \end{itemize}
\end{definition}

\noindent There is a new form of expression, $\return e {}$, which is
used to represent an ongoing method call.

Finally, a state consists of a heap and an expression, and the
operational semantics will be defined as a reduction relation on states;
$\mathcal{E}$ are evaluation contexts in the style of Wright and
Felleisen \cite{WrightAK:synats}, used in the definition of reduction.

The semantic and typing rules we will present next are implicitly
parameterized by the set of declarations $D$ which constitute the
program.  It is assumed that the whole set is available at any
point
and that any class is declared only once.  We do not require the sets
of method or field names to be disjoint from one class to another.  We
will use the following notation: if $\class{C}{S}{\vec f}{\vec M}$ is
one of the declarations, $C.\sessterm$ means $S$ and $C.\fieldsterm$
means $\vec f$, and if $\method{}{m}{x}{e}\in\vec M$ then $C.m$ is
$e$.


\begin{figure}
\centering\setlength\lineskip{\typingRuleSkip}
\AxiomC{$\method{}{m}{x}{e}\in h(r.f).\classterm$}
\rulename{R-Call}
\UnaryInfC{$\state{h}{r}{\methcal{f}{m}{v}} \reduces
  \state{h}{r.f}{\return{e\subs{v}{x}}{}}$}
\DisplayProof\hfil 
\axiomname{R-Return}
$\state{h}{r.f}{\return v {}} \reduces \state{h}{r}{v}$
\hfil
\\~\\
%
%
%
\AxiomC{$o\not\in\dom(h)$}
\AxiomC{$C.\mkterm{fields} = \vec f$}
\rulename{R-New}
\BinaryInfC{$\state{h}{r}{\new C}\reduces
\state{\hadd{h}{o=\hentry{C}{\vec f = \overrightarrow\nullterm}}}{r}{o}$}
\DisplayProof\hfil
\AxiomC{$h(r).f = v$}
\rulename{R-Swap}
\UnaryInfC{$\state{h}{r}{\swap{f}{v'}}\reduces\state{\changeval{h}{r.f}{v'}}{r}{v}$}
\DisplayProof\hfil
\AxiomC{$l_0\in E$}
\rulename{R-Switch}
\UnaryInfC{$\state{h}{r}{\switch{l_0}{l}{e_l}{l\in E}} \reduces
  \state{h}{r}{e_{l_0}}$}
\DisplayProof\hfil
\axiomname{R-Seq}
$\state{h}{r}{\seq ve} \reduces \state{h}{r}{e}$
\hfil
\\~\\
%
%
\AxiomC{$\state h r e \reduces\state{h'}{r'}{e'}$}
\rulename{R-Context}
\UnaryInfC{$\state{h}{r}{\mathcal E[e]}\reduces\state{h'}{r'}{\mathcal E[e']}$}
\DisplayProof
\caption{Reduction rules for states}
\label{fig:reduction}
\end{figure}



\subsection{Operational Semantics}
\label{sec:reduction}

Figure~\ref{fig:reduction} defines an operational semantics on states
$\state h r e$ consisting of a heap $h$, a path $r$ in the heap
indicating the \emph{current object}, and an expression $e$. In
general, $e$ is an expression obtained by a series of reduction steps
from a method body, where the method was called on the object
identified by the path $r$. All rules have the implicit premise that
the expressions appearing in them must be defined. For example,
$\swap{f}{v}$ only reduces if $h(r)$ is an object record containing a
field named $f$. An example of reduction, together with typing, is
presented in Figure~\ref{fig:reduction-ex} and discussed at the end of
the present section.

The current object path $r$ is used to resolve field references
appearing in the expression~$e$. It behaves like a call
stack: as shown in \textsc{R-Call}, when a method call on a field $f$
(relative to the current object located at $r$) is entered, the object
in $r.f$ becomes the current object; this is indicated by changing the
path to $r.f$.  Additionally, the method body, with the actual
parameter substituted for the formal parameter, is wrapped in a
$\returnterm$ expression and replaces the method call. When the body
has reduced to a value, this value is unwrapped by \textsc{R-Return}
which also pops the field specification $f$ from the path, recovering the
previous current object $r$. This is illustrated in
Figure~\ref{fig:reduction-ex}, which also shows the typing of
expressions in a series of reductions.
\Rnew\ creates a new object in the heap, with $\nullterm$ fields. 
\textsc{R-Swap} updates the value of a field and reduces to its former value.

\Rswitch\ is standard.
\Rseq\ discards the result of the first part of a sequential
composition. \Rcontext\ is the usual rule for reduction in contexts.

To complete the definition of the semantics we need to define the
initial state. The idea is to designate a particular method $m$ of a
particular class $C$ as the \emph{main} method, which is called in
order to begin execution. The most convenient way to express this is
to have an initial heap that contains an object of class $C$, which is
also chosen as the current object, and an initial expression $e$ which is
the body of $m$. The initial state is therefore
\[
\state{\mkterm{top}=\hentry{C}{C.\fieldsterm=\vec\nullterm}}{\mkterm{top}}{e}
\]
where $\mkterm{top}$ is the identifier of the top-level object of
class $C$, which is the only object in the heap, and the current object
path is also $\mkterm{top}$.
Strictly speaking, method $m$ must have a parameter $x$; we take
$x$ to be of type $\nulltype$ and assume that it does not occur in $e$.

\subsection{Example of reduction}
\label{sec:reduction-ex}
\begin{figure*}[t]
\[
\begin{array}{c}
\state{\mkterm{top}=\hentry{C}{f=\nullterm, g=\nullterm}}{\mkterm{top}}{\seq{f = \new{C'}}{\seq{g = f.m_j()}{\switchterm~(g)~\{l\colon
    e_l\}_{l\in E}}}}
\\
\downarrow \scriptsize{(\mbox{Expand})}
\\
\state{\mkterm{top}=\hentry{C}{f=\nullterm, g=\nullterm}}{\mkterm{top}}{\seq{\seq{\swap{f}{\new{C'}}}{\nullterm}}{\seq{g = f.m_j()}{\switchterm~(g)~\{l\colon
    e_l\}_{l\in E}}}}
\\
\downarrow \scriptsize{(\Rnew)}
\\
\state{\mkterm{top}=\hentry{C}{f=\nullterm, g=\nullterm}, o=\hentry{C'}{}}{\mkterm{top}}{\seq{\seq{\swap{f}{o}}{\nullterm}}{\seq{g = f.m_j()}{\switchterm~(g)~\{l\colon
    e_l\}_{l\in E}}}}
\\
\downarrow \scriptsize{(\Rswap)}
\\
\state{\mkterm{top}=\hentry{C}{f=o, g=\nullterm}, o=\hentry{C'}{}}{\mkterm{top}}{\seq{\seq{\nullterm}{\nullterm}}{\seq{g = f.m_j()}{\switchterm~(g)~\{l\colon
    e_l\}_{l\in E}}}}
\\
\downarrow \scriptsize{(\Rseq,\Rseq)}
\\
\state{\mkterm{top}=\hentry{C}{f=o, g=\nullterm}, o=\hentry{C'}{}}{\mkterm{top}}{\seq{g = f.m_j()}{\switchterm~(g)~\{l\colon
    e_l\}_{l\in E}}}
\\
\downarrow \scriptsize{(\mbox{Expand})}
\\
\state{\mkterm{top}=\hentry{C}{f=o, g=\nullterm}, o=\hentry{C'}{}}{\mkterm{top}}{\seq{\seq{\swap{g}{f.m_j()}}{\nullterm}}{\switchterm~(\swap{g}{\nullterm})~\{l\colon e_l\}_{l\in E}}}
\\
\downarrow \scriptsize{(\mbox{Simplify})}
\\
\state{\mkterm{top}=\hentry{C}{f=o, g=\nullterm},
  o=\hentry{C'}{}}{\mkterm{top}}{\seq{\swap{g}{f.m_j()}}{\switchterm~(\swap{g}{\nullterm})~\{l\colon
    e_l\}_{l\in E}}}
\\
\downarrow \scriptsize{(\Rcall)}
\\
\state{\mkterm{top}=\hentry{C}{f=o, g=\nullterm},
  o=\hentry{C'}{}}{\mkterm{top}.f}{\seq{\swap{g}{\return{e}{}}}{\switchterm~(\swap{g}{\nullterm})~\{l\colon e_l\}_{l\in E}}}
\\
\downarrow \scriptsize{(\mbox{Evaluate $e$})}
\\
\state{\mkterm{top}=\hentry{C}{f=o, g=\nullterm},
  o=\hentry{C'}{}}{\mkterm{top}.f}{\seq{\swap{g}{\return{l_0}{}}}{\switchterm~(\swap{g}{\nullterm})~\{l\colon e_l\}_{l\in E}}}
\\
\downarrow \scriptsize{(\Rreturn)}
\\
\state{\mkterm{top}=\hentry{C}{f=o, g=\nullterm},
  o=\hentry{C'}{}}{\mkterm{top}}{\seq{\swap{g}{l_0}}{\switchterm~(\swap{g}{\nullterm})~\{l\colon e_l\}_{l\in E}}}
\\
\downarrow \scriptsize{(\Rswap,\Rseq)}
\\
\state{\mkterm{top}=\hentry{C}{f=o, g=l_0},
  o=\hentry{C'}{}}{\mkterm{top}}{\switchterm~(\swap{g}{\nullterm})~\{l\colon e_l\}_{l\in E}}
\\
\downarrow \scriptsize{(\Rswap)}
\\
\state{\mkterm{top}=\hentry{C}{f=o, g=\nullterm},
  o=\hentry{C'}{}}{\mkterm{top}}{\switchterm~(l_0)~\{l\colon e_l\}_{l\in E}}
\\
\downarrow \scriptsize{(\Rswitch)}
\\
\state{\mkterm{top}=\hentry{C}{f=o, g=\nullterm},
  o=\hentry{C'}{}}{\mkterm{top}}{e_{l_0}}
\end{array}
\]
\caption{A series of reduction steps.}
\label{fig:reduction-ex-no-types}
\end{figure*}


Assume that the top-level class is $C$, containing fields $f$ and $g$. Assume
also that there is another class $C'$ which defines the set of methods
$\{ m_i \mid i \in I \}$. Finally, assume that the body of the main
method of $C$ 
is
\begin{equation*}
\seq{f = \new{C'}}{\seq{g = f.m_j()}{\switchterm~(g)~\{l\colon
    e_l\}_{l\in E}}}
\end{equation*}
where $j$ is some element of $I$.

The initial state is
\begin{equation*}
\state{\mkterm{top}=\hentry{C}{f=\nullterm, g=\nullterm}}{\mkterm{top}}{\seq{f = \new{C'}}{\seq{g = f.m_j()}{\switchterm~(g)~\{l\colon
    e_l\}_{l\in E}}}}
\end{equation*}
where for simplicity we have ignored the parameter of $m_j$. 
Figure~\ref{fig:reduction-ex-no-types} shows the sequence of reduction
steps until one of the cases of the $\switchterm$ is reached.

The first step is expansion of the syntactic sugar for assignment,
translating it into a swap followed by $\nullterm$. Another similar
translation step occurs later. The first real reduction step is
$\Rnew$, creating an object $o$, followed by $\Rswap$ to complete the
assignment into field $f$ and then $\Rseq$ to tidy up.  Next, the step
labelled ``simplify'' informally removes $\nullterm$ in order to avoid
carrying it through to an uninteresting $\Rseq$ reduction later.

Now assume that the body of method $m_j$ is $e$. Reduction by $\Rcall$
changes the current object path to $\mkterm{top}.f$ because the
current object is now $o$. Several reduction steps convert $e$ to a
particular element of the enumerated type $E$, which we call
$l_0$. After that, $\Rreturn$ changes the current object path back to
$\mkterm{top}$, and then some $\Rswap$ and $\Rseq$ steps bring $l_0$
into the guard of the $\switchterm$, finally allowing the appropriate
case $e_{l_0}$ to be selected.

We will return to this example in Section~\ref{sec:example-red}, to
show how each state is typed.


\subsection{Subtyping}
\label{sec:subtyping}
\begin{figure}
{\centering\setlength\lineskip{\typingRuleSkip}
%
%
\AxiomC{$\forall i\in I, T_i\subt T'_i$}
\rulename{S-Record}
\UnaryInfC{$\iset{T_i\,f_i}{i\in I}\subt \iset{T'_i\,f_i}{i\in I}$} 
\DisplayProof\hfil
%
%
\AxiomC{$E\subset E'$}
\AxiomC{$\forall l\in E, F_l\subt F'_l$}
\rulename{S-Variant}
\BinaryInfC{$\vfield{l}{F_l}{l\in E} \subt \vfield{l}{F'_l}{l\in E'}$}
\DisplayProof\hfil
\AxiomC{$F\subt F'$}
\rulename{S-Field}
\UnaryInfC{$\objecttype C F\subt\objecttype{C}{F'}$}
\DisplayProof\hfil
%
%
\par}
\caption{Subtyping rules for fields}
\label{fig:subtyping}
\end{figure}

Subtyping is an essential ingredient of the theory of session types.
Originally proposed by Gay and Hole~\cite{GaySJ:substp}, it has been
widely used in other session-based systems, with subject-reduction and
type-safety holding. The guiding principle is the ``safe
substitutability principle'' of Liskov and
Wing~\cite{liskov.wing:behavioural-subtyping}, which states that, if
$S$ is a subtype of $T$, then objects of type $T$ in a program may be
safely replaced with objects of type $S$.

Two kinds of types in the top-level core language are subject to
subtyping: enumerated types and session types. The internal language
also has field typings; subtyping on them is derived from subtyping on
top-level types by the rules in Figure~\ref{fig:subtyping}.

Subtyping for enumerated types is defined as simple set inclusion:
$E\subt E'$ if and only if $E\subset E'$. We refer to subtyping for
session types as the \emph{sub-session} relation. Because session
types can be recursive, the sub-session relation is defined
coinductively, by defining necessary conditions it must satisfy and
taking the largest relation satisfying them. The definition involves
checking compatibility between different method signatures, which
itself is dependent on the whole subtyping relation. We proceed as
follows: given a candidate sub-session relation $\mathcal{R}$, we
define an $\mathcal{R}$-compatibility relation between types and
between method signatures which uses $\mathcal{R}$ as a sub-session
relation. We then use $\mathcal{R}$-compatibility in the structural
conditions that $\mathcal{R}$ must satisfy in order to effectively be
a sub-session relation.

Let $\mathcal{S}$ denote the set of contractive, closed, class session
types.  We deal with recursive types using the following
$\mkterm{unfold}$ operator:
\begin{definition}[Unfolding]
  The operator $\mkterm{unfold}$ is defined inductively on
  $\mathcal{S}$ by $\mkterm{unfold}(\rectype{X}{S}) =
  \mkterm{unfold}(S\subs{(\rectype{X}{S})}{X})$ and $\mkterm{unfold}(S) = S$ if $S$ is
  not of the form $\rectype{X}{S}$.  Since the types in $\mathcal{S}$
  are contractive, this definition is well-founded.
\end{definition}

We now define the two compatibility relations we need.
\begin{definition}[$\mathcal{R}$-Compatibility (Types)]
\label{def:Rcompatibility-types}
Let $\mathcal{R}$ be a binary relation on $\mathcal{S}$. We say that
type $T$ is $\mathcal{R}$-compatible with type $T'$ if one of the
following conditions is true.
\begin{enumerate}
\item $T=T'$
\item $T$ and $T'$ are enumerated types and $T\subset T'$ 
\item $T, T' \in \mathcal{S}$ and $(T, T') \in \mathcal{R}$.
\end{enumerate}
\end{definition}

\begin{definition}[$\mathcal{R}$-Compatibility (Signatures)]
\label{def:Rcompatibility-signatures}
Let $\mathcal{R}$ be a binary relation on $\mathcal{S}$.  Let $\sigma
= \methsign{m}{T}{U} : S$ and $\sigma' = \methsign{m}{T'}{U'} : S'$ be
components of branch types, both for the same method name $m$, i.e.\
method signatures with subsequent session types.  We say that $\sigma$
is $\mathcal{R}$-compatible with $\sigma'$ if $T'$ is
$\mathcal{R}$-compatible with $T$ and either:
\begin{enumerate}
\item $U$ is $\mathcal{R}$-compatible with $U'$ and $(S, S')\in
  \mathcal{R}$, or
\item $U$ is an enumerated type $E$, $U' = \linkthis$ and
  $(\choice{l}{S}{l\in E}, S')\in \mathcal{R}$.
\end{enumerate}
\end{definition}

The compatibility relation on method signatures is, as expected,
covariant in the return type and the subsequent session type and
contravariant in the parameter type, but with one addition: if a
method has an enumerated return type $E$ and subsequent session type
$S$, then it can always be used as if it had a return type of
$\linkthis$ and were followed by the uniform variant session type
$\choice{l}{S}{l\in E}$.  Indeed, both signatures mean that the method
can return any label in $E$ and will always leave the object in state
$S$.

We can now state the necessary conditions for a sub-session relation.
\begin{definition}[Sub-session]
\label{def:subsession}
Let $\mathcal{R}$ be a binary relation on $\mathcal{S}$. We say that
$\mathcal{R}$ is a sub-session relation if $(S, S')\in \mathcal{R}$
implies:
\begin{enumerate}
\item If $\mkterm{unfold}(S) =
  \branch{\methsign{m_i}{T_i}{U_i}}{S_i}{i\in I}$ then
  $\mkterm{unfold}(S')$ is of the form
  $\branch{\methsign{m_j}{T'_j}{U'_j}}{S'_j}{j\in J}$ with $J\subset
  I$, and for all $j\in J$, $\methsign{m_j}{T_j}{U_j}:S_j$ is
  $\mathcal{R}$-compatible with $\methsign{m_j}{T'_j}{U'_j}:S'_j$.
\item If $\mkterm{unfold}(S) = \choice{l}{S_l}{l\in E}$ then
  $\mkterm{unfold}(S')$ is of the form $\choice{l}{S'_l}{l\in E'}$ with
  $E\subset E'$ and for all $l\in E$, $(S_l, S'_l)\in\mathcal{R}$.
\end{enumerate}
For the sake of simplicity we will now, when we refer to this
definition later on, make the unfolding step implicit by assuming,
without loss of generality, that neither $S$ nor $S'$ is of the form
$\rectype{X}{S''}$.
\end{definition}

\begin{lemma}\label{lem:union} 
  The union of several sub-session relations is a sub-session
  relation.
\end{lemma}
\begin{proof}
Let $\mathcal{R} = \bigcup_{i\in I}\mathcal{R}_i$, where
the $\mathcal{R}_i$ are sub-session relations. Let $(S, S')\in
\mathcal{R}$. Then there is $j$ in
$I$ such that $(S, S')\in \mathcal{R}_j$. This implies that $(S, S')$
satisfies the conditions in Definition~\ref{def:subsession} with
respect to $\mathcal{R}_j$. Just notice that, because
$\mathcal{R}_j\subset\mathcal{R}$, the conditions are satisfied with
respect to $\mathcal{R}$ as well --- in particular,
$\mathcal{R}_j$-compatibility implies $\mathcal{R}$-compatibility.
Indeed, the conditions for $\mathcal{R}$ only differ from those for
$\mathcal{R}_j$ by requiring particular pairs of
session types to be in $\mathcal{R}$ rather than in $\mathcal{R}_j$, so
they are looser.
\end{proof}

We now define the subtyping relation $\subt$ on session types to be
the largest sub-session relation, i.e.\ the union of all sub-session
relations. The subtyping relation on general top-level types is just
$\subt$-compatibility.

Subtyping on session types means that either both are branches or both
are variants. In the former case, the supertype must allow fewer
methods and their signatures must be compatible; in the latter case,
the supertype must allow more labels and the common cases must be in
the subtyping relation.  Like the definition of subtyping for channel
session types \cite{GaySJ:substp}, the type that allows a choice to be
made (the branch type here, the $\oplus$ type for channels)
has contravariant subtyping in the set of choices.

The following lemma shows that the necessary conditions of Definition
\ref{def:subsession} are also sufficient in the case of $\subt$.
\begin{lemma}\label{lem:subsession}\hfill
  \begin{enumerate}
  \item Let $S = \branch{\methsign{m_i}{T_i}{U_i}}{S_i}{i\in I}$
    and $S' = \branch{\methsign{m_j}{T'_j}{U'_j}}{S'_j}{j\in J}$ with $J\subset
    I$. If for all $j\in J$, $\methsign{m_j}{T_j}{U_j}:S_j$ is
    $\subt$-compatible with $\methsign{m_j}{T'_j}{U'_j}:S'_j$,
    then $S\subt S'$.
  \item Let $S = \choice{l}{S_l}{l\in E}$ and
    $S' = \choice{l}{S'_l}{l\in E'}$ with
    $E\subset E'$. If for all $l$ in $E$ we have $S_l\subt S'_l$, then
    $S\subt S'$.
  \end{enumerate}
\end{lemma}
\begin{proof}
  The relation $\subt\cup\,\{(S, S')\}$ is a sub-session relation.
\end{proof}

Finally, we prove that this subtyping relation provides a preorder on types.
\begin{proposition}
  The subtyping relation is reflexive and transitive.
\end{proposition}
\begin{proof}
  First note that session types can only be related by subtyping to
  other session types; the same applies to enumerated types and, in
  the internal system, field typings. Since the relation for
  enumerated types is just set inclusion, we already know the result
  for it. We now prove the properties for session types; the fact that
  they hold for field typings is then a straightforward consequence.

  For reflexivity, just notice that the diagonal relation $\{(S,
  S)\mid S\in \mathcal{S}\}$ is a sub-session relation, hence included
  in $\subt$.

  For transitivity, what we need to prove is that the relation
  $\mathcal{R} = \{(S, S')\mid \exists S'', S\subt S'' \wedge S''\subt
  S'\}$ is a sub-session relation.  Let $(S, S')\in\mathcal{R}$ and
  let $S''$ be as given by the definition of $\mathcal{R}$.

  In case (1) where we have $S =
  \branch{\methsign{m_i}{T_i}{U_i}}{S_i}{i\in I}$, we know that:
\begin{itemize}
\item $S'' = \branch{\methsign{m_j}{T''_j}{U''_j}}{S''_j}{j\in J}$
  with $J\subset I$, and for all $j\in J$, $\sigma_j =
  \methsign{m_i}{T_i}{U_i}:S_j$ is $\subt$-compatible with $\sigma''_j
  = \methsign{m_j}{T''_j}{U''_j}:S''_j$.
\item Therefore, $S'$ is of the form
  $\branch{\methsign{m_k}{T'_k}{U'_k}}{S'_k}{k\in K}$ with $K\subset
  J$, and for all $k\in K$, $\sigma''_k$ is $\subt$-compatible with
  $\sigma'_k =\methsign{m_k}{T'_k}{U'_k}:S'_k$.
\end{itemize}
Straightforwardly $K\subset I$. For every $k$ in $K$, we have to prove
that $\sigma_k$ is $\mathcal{R}$-compatible with $\sigma'_k$. We
deduce it from the two $\subt$-compatibilities we know by looking into
the definition of compatibility point by point:
\begin{itemize}
\item We have $T''_k\subt T_k$ and $T'_k\subt T''_k$. Either these
  types are all session types, and then $(T'_k, T_k)\in\mathcal{R}$ by
  definition of $\mathcal{R}$, or none of them is and we have
  $T'_k\subt T_k$ by transitivity of subtyping on base types.  In both
  cases, $T'_k$ is $\mathcal{R}$-compatible with $T_k$.
\item We also have either:
  \begin{itemize}
  \item $U_k\subt U''_k$, $U''_k\subt U'_k$, $S_k\subt S''_k$ and
    $S''_k\subt S'_k$. In this case, the former two conditions imply,
    similarly to the above, that $U_k$ is $\mathcal{R}$-compatible
    with $U'_k$. The latter two imply $(S_k, S'_k)\in\mathcal{R}$.
  \item Or $U_k$ is an enumerated type $E$, $U''_k = U'_k =
    \linkthis$, $\choice{l}{S_k}{l\in E}\subt S''_k$ and $S''_k\subt
    S'_k$. Then we have $(\choice{l}{S_k}{l\in E},
    S'_k)\in\mathcal{R}$, which is all we need.
  \item Or, finally, $U_k$ is an enumerated type $E$, $U''_k$ is an
    enumerated type $E''$ such that $E\subset E''$, $S_k\subt S''_k$
    and $\choice{l}{S''_k}{l\in E''}\subt S'_k$. Then from $S_k\subt
    S''_k$ and $E\subset E''$ we deduce, using case (2) of
    Lemma~\ref{lem:subsession}, $\choice{l}{S_k}{l\in
      E}\subt\choice{l}{S''_k}{l\in E''}$. We thus have, again,
    $(\choice{l}{S_k}{l\in E}, S'_k)\in\mathcal{R}$ which is the
    required condition.
  \end{itemize}
\end{itemize}

\noindent In case (2) where we have $S = \choice{l}{S_l}{l\in E}$, we obtain
$S'' = \choice{l}{S''_l}{l\in E''}$ and $S' = \choice{l}{S'_l}{l\in
  E'}$, with $E\subset E''\subset E'$ and for any $l$ in $E$,
$S_l\subt S''_l$ and $S''_l\subt S'_l$, which imply by definition of
$\mathcal{R}$ that $(S_l, S'_l)$ is in $\mathcal{R}$.\end{proof}

\begin{definition}[Type equivalence]
\label{def:typeequivalence}
We define \emph{equivalence} of session types $S$ and $S'$ as $S \subt
S'$ and $S' \subt S$. This corresponds precisely to $S$ and $S'$
having the same infinite unfoldings (up to the ordering of cases in
branches and variants).
Henceforth types are understood up to type equivalence, so that, for
example, in any mathematical context, types $\rectype{X}{S}$ and
$S\subs{(\rectype{X}{S})}{X}$ can be used interchangeably,
effectively adopting the equi-recursive approach~\cite[Chapter
21]{PierceBC:typpl}.
\end{definition}







\subsection{Type System}
\label{sec:typing}

We introduce a static type system whose purpose is to ensure that
typable programs satisfy a number of safety properties. As usual, we
make use of a type preservation theorem, which states that reduction
of a typable expression produces another typable expression. Therefore
the type system is formulated not only for top-level expressions but for
the states (i.e.\ (heap, expression) pairs) on which the reduction
relation is defined. 

An important feature of the type system is that the method definitions
within a particular class are not checked independently, but are
analyzed in the order specified by the session type of the class.
This is expressed by rule \Tclass, the last rule in
Figure~\ref{fig:typingexpr}, which uses a consistency relation
$\typedsess{\overrightarrow\nulltype\,\vec f}{C}{S}$ between field
typings and session types, defined in
Section~\ref{sec:consistency}. Checking this relation requires
checking the definitions of the methods occurring in $S$, in
order. Checking method definitions uses the typing judgement for
expressions, which is defined by the other rules in
Figure~\ref{fig:typingexpr}.

In the following sections we describe the type
system in several stages.

\begin{figure}
\centering\setlength\lineskip{\typingRuleSkip}\setlength\baselineskip{0pt}
\axiomname{T-Null}
\judgment \envref\Gamma r > \nullterm : \nulltype < \envref\Gamma r /
\hfil
\axiomname{T-Label}
\judgment \envref\Gamma r > l : \{l\} < \envref\Gamma r /
\hfil
\axiomname{T-New}
\judgment \envref\Gamma r > \new C : C.\sessterm < \envref\Gamma r /
\hfil
\axiomname{T-LinVar}
\judgment \envref{\Gamma, x : S}{r} > x : S < \envref{\Gamma}{r} / 
\hfil
\AxiomC{$T$ is not an object type}
\rulename{T-Var}
\UnaryInfC{\judgment \envref{\Gamma, x : T}{r} > x : T <
  \envref{\Gamma, x : T}{r} /}
\DisplayProof\hfil
\AxiomC{\judgment \envref\Gamma r > e : T < \envref{\Gamma'}{r'} /}
\AxiomC{$\Gamma'(r'.f) = T'$}
\AxiomC{$T\not=\linkthis$\quad$T'$ is not a variant}
\rulename{T-Swap}
\TrinaryInfC{\judgment \envref\Gamma r > \swap{f}{e} : T' <
  \envref{\changetype{\Gamma'}{r'.f}{T}}{r'} /}
\DisplayProof\hfil
\let\oldvskip\extraVskip
\def\extraVskip{0.5pt} 
\AxiomC{$\judgment \envref\Gamma r > e : T'_j < \envref{\Gamma'}{r'} /$}
\AxiomC{$\Gamma'(r'.f) = \branch{\methsign{m_i}{T'_i}{T_i}}{S_i}{i\in I}$}
\noLine
\BinaryInfC{$\quad j\in I\qquad T = \linktype{}{f}$ if $T_j = \linkthis$,
  $T = T_j$ otherwise $\quad$}
\let\extraVskip\oldvskip
\rulename{T-Call}
\UnaryInfC{\judgment \envref\Gamma r > \methcal{f}{m_j}{e} :
  T < \envref{\changetype{\Gamma'}{r'.f}{S_j}}{r'} /}
\DisplayProof\hfil  
\AxiomC{\judgment \envref\Gamma r > e : T < \envref{\Gamma'}{r'} /}
\AxiomC{\judgment \envref{\Gamma'}{r'} > e' : T' < \envref{\Gamma''}{r'} /}
\AxiomC{$T \neq \linktype{}{\_}$ or $\linkthis$}
\rulename{T-Seq}
\TrinaryInfC{\judgment \envref\Gamma r > \seq{e}{e'} : T' <
  \envref{\Gamma''}{r'} /}
\DisplayProof\hfil
\AxiomC{\judgment \envref\Gamma r > e : E' < \envref{\Gamma'}{r'} /}
\AxiomC{$E'\subset E$}
\AxiomC{$\forall l\in E', \judgment \envref{\Gamma'}{r'} > e_l : T <
  \envref{\Gamma''}{r'} /$}
\rulename{T-Switch}
\TrinaryInfC{\judgment \envref\Gamma r > \switch{e}{l}{e_l}{l\in E} : T
< \envref{\Gamma''}{r'} /}
\DisplayProof\hfil
\let\oldvskip\extraVskip
\def\extraVskip{0.5pt} 
\AxiomC{\judgment \envref\Gamma r > e : \linktype{}{f} <
  \envref{\Gamma'}{r'} /}
\AxiomC{$\Gamma'(r'.f) = \choice{l}{S_l}{l\in E'}$}
\noLine
\BinaryInfC{$E'\subset E\qquad \forall l \in E',
  \judgment\envref{\changetype{\Gamma'}{r'.f}{S_l}}{r'}
  > e_l : T < \envref{\Gamma''}{r'} / \qquad$}
\let\extraVskip\oldvskip
\rulename{T-SwitchLink}
\UnaryInfC{\judgment \envref\Gamma r > \switch{e}{l}{e_l}{l\in E} : T < \envref{\Gamma''}{r'} /}
\DisplayProof\hfil
\AxiomC{\judgment \envref\Gamma r > e : E < \envref{\Gamma'}{r'} /}
\AxiomC{$\Gamma'(r') = \objecttype C{F'}$}
\AxiomC{$F'$ is a record}
\rulename{T-VarF}
\TrinaryInfC{\judgment \envref\Gamma r > e : \linkthis <
  \envref{\changetype{\Gamma'}{r'}{\objecttype C{\choice{l}{F'}{l\in E}}}}{r'} /}
\DisplayProof
\hfil
\AxiomC{\judgment \envref\Gamma r > e : T < \envref{\Gamma'}{r'} /}
\AxiomC{$T\subt T'$}
\rulename{T-Sub}
\BinaryInfC{\judgment \envref\Gamma r > e : T' < \envref{\Gamma'}{r'} /}
\DisplayProof\hfil
\AxiomC{\judgment\envref\Gamma r > e : T < \envref{\Gamma'}{r'} /}
\AxiomC{$\Gamma'\subt\Gamma''$}
\rulename{T-SubEnv}
\BinaryInfC{\judgment \envref\Gamma r > e : T < \envref{\Gamma''}{r'} /}
\DisplayProof\hfil
\AxiomC{$\typedsess{\overrightarrow\nulltype\,\vec f}{C}{S}$}
\rulename{T-Class}
\UnaryInfC{$\vdash\class{C}{S}
  {\vec f}
  {\vec M}$}
\DisplayProof
\caption{Typing rules for the top level language}
\label{fig:typingexpr}
\end{figure}
       


\subsubsection{Typing expressions}
\label{sec:typingexpressions}
\begin{definition}[Type environments]
We use type environments of the form $\Gamma =
\alpha_1:T_1,\ldots,\alpha_n:T_n$ where each $\alpha$ is either a
method parameter $x$ or an object identifier $o$.  
\end{definition}
The typing judgement for expressions is $\judgment \envref{\Gamma}{r}
> e:T < \envref{\Gamma'}{r'} /$. In such a judgement, $\Gamma$ and
$\Gamma'$ are type environments and $r$ and $r'$ are paths. The paths
parameters are in fact only necessary for typing runtime expressions;
they are needed for our type preservation theorem, but not
for type-checking a program, where they always have the value $\this$.
They will be discussed in Section~\ref{sec:typingfull}.

The expression $e$ and its type $T$ appear in the central part of the
judgement. The $\Gamma'$ on the right hand side shows the
change, if any, that $e$ causes in the type environment. There are
several reasons for $\Gamma'$ to differ from 
$\Gamma$; the most important is that if $e$ contains a method call on
an object, then the session type of that object is different in
$\Gamma'$ than it was in $\Gamma$. Another one is linearity: if a
linear parameter
$x$ is used in $e$, then $x$ does not appear in $\Gamma'$ because it
has been consumed.

When type-checking a program (as opposed to typing a runtime
expression, which needs not be implemented), the judgements for
expressions always have the particular form $\judgment
\envref{\this:\objecttype{C}{F}, V}{\this} > e:T <
\envref{\this:\objecttype{C}{F'}, V'}{\this} /$, with only one object
identifier in the environment, $\this$, representing the object to
which fields referred to in $e$ belong. The rest of the environment,
$V$, is either empty or has the form $x:U$ where $U$ is the type of
the parameter $x$ of the method currently being type-checked, and is
thus a top-level type. The initial type of $\this$ is the internal
type $\objecttype{C}{F}$, where $F$ is a field typing; the final type
is $\objecttype{C}{F'}$, as $e$ may change the types of the fields
(for example, by calling methods on them). The final parameter typing
$V'$ is either the same as $V$ or empty, depending whether the
parameter was consumed by $e$.





We extend subtyping to a relation on type environments, as follows.
\begin{definition}[Environment subtyping]
$\Gamma \subt \Gamma'$ if for every $\alpha\in\dom(\Gamma')$, where
$\alpha$ is either a parameter or an object identifier, we have
$\alpha\in\dom(\Gamma)$ and $\Gamma(\alpha)\subt\Gamma'(\alpha)$.
\end{definition}
Essentially $\Gamma\subt\Gamma'$ if $\Gamma$ is more precise (contains
more information) than $\Gamma'$: it contains types for everything in
$\Gamma'$ (and possibly more) and those types are more specific.

\subsubsection{Consistency between field typings and session types}
\label{sec:consistency}
There are two possible forms for the type of an object. One is a
session type $S$, which describes the view of the object from outside,
i.e.\ from the perspective of code in other classes. The session type
specifies which methods may be called, but does not reveal information
about the fields. The other form, $\objecttype{C}{F}$, contains a
field typing $F$, and describes the internal view of the object,
i.e.\ from the perspective of code in its own methods. Consider a
sequence of method calls in a particular class. There are two senses
in which it may be considered correct or incorrect. (1) In the sense
that it is allowed, or not allowed, by the session type of the
class. (2) In the sense that each call in the sequence leaves the fields of
the object in a state which ensures the next call does not produce a
type error. For example, if we consider the class
\lstinline|FileReader| of Figure~\ref{fig:filereader}, we see that the
session type allows calling \lstinline|read()| just after
\lstinline|init()|, making the sequence \lstinline|init(); read()|
correct in sense (1). It is correct in sense (2) if and only if the body
of \lstinline|read()| typechecks under the
precondition that the fields \lstinline|file| and \lstinline|text|
have the types produced by the evaluation of \lstinline|init()|.

In order to type a class definition, these two senses of correctness
must be consistent according to the following coinductive definition.

\begin{definition}
\label{def:sessionfield}
Let $C$ be a class and let $\mathcal{R}$ be a relation between field
typings $F$ and session types $S$. We say that $\mathcal{R}$ is a
$C$-consistency relation if $(F,S)\in\mathcal{R}$ implies:
\begin{enumerate}
\item If $S = \branch{\methsign{m_i}{T'_i}{T_i}}{S_i}{i\in I}$, then
  $F$ is not a variant and for all $i$ in $I$, there is a definition
  $\method{}{m_i}{x_i}{e_i}$ in the declaration of class $C$, and a field typing $F_i$, such that
$$\judgment \envref{\this: \objecttype C F, x_i : T'_i}{\this}
> e_i : T_i < \envref{\this: \objecttype C {F_i}}{\this} /$$ and
 $(F_i,S_i)\in\mathcal{R}$.
\item If $S = \choice{l}{S_l}{l\in E}$, then $F = \vfield{l}{F_l}{l\in
    E'}$ with $E'\subset E$ and for all $l$ in $E'$ we have
  $(F_l,S_l)\in\mathcal{R}$.
\end{enumerate}
\end{definition}

\noindent In clause (1), $S$ is the session type (external view) before calling
one of the methods $m_i$, and $F$ is the field typing (internal
view). If a particular $m_i$ is called then the subsequent session
type is $S_i$, and the subsequent field typing, arising from the
typing judgement for the method body, is $F_i$. These types must be
related. Clause (2) requires variant session types and field typings,
arising from a method call that returns an enumeration label, to
match. The inclusion $E'\subset E$ allows the method to return labels
from a smaller set than the one defined by the session type.
\begin{lemma}
  The union of several $C$-consistency relations is a $C$-consistency
  relation.
\end{lemma}
\begin{proof}
  Similar to (but simpler than) the proof of Lemma \ref{lem:union}.
\end{proof}
For any class $C$, we define the relation $\typedsess F C S$ between
field typings $F$ and session types $S$ to be the largest
$C$-consistency relation, i.e.\ the union of all $C$-consistency
relations.

The relation $\typedsess F C S$ represents the fact that an object of
class $C$ with internal (private) field typing $F$ can be safely
viewed from outside as having type $S$.  Clause (2) in
Definition~\ref{def:sessionfield} accounts for correspondence between
variant types. The main clause is clause (1): if the object's fields
have type $F$ and its session type allows a certain method to be
called, then it means that the method body is typable with an initial
field typing of $F$ and the declared type for the parameter.
Furthermore, the type of the expression must match the declared return
type and the final type of the fields must be compatible with the
subsequent session type. The parameter may or may not be consumed by
the method, but \textsc{T-SubEnv} at the end of
Figure~\ref{fig:typingexpr} (see next section) allows discarding it
silently in any case, hence its absence from the final environment.

The definition implies that a method with return type $\linkthis$ must
be followed by a variant session type, for the following
reason. Suppose that in clause (1), some $T_i$ is $\linkthis$. The
only way for $e_i$ to have type $\linkthis$ is by using rule \TvarF{}
(discussed in the next section), which implies that $F_i$ must be a
variant field typing. The condition $(F_i,S_i)\in\mathcal{R}$ implies,
by clause (1), that $S_i$ is a variant.

The rule $\Tclass$, last rule in Figure~\ref{fig:typingexpr}, checks that the
initial session type of a class is consistent with the initial
$\nulltype$ field typing. It refers to the above definition of
consistency, which itself refers to typing judgements built using the
other rules in the figure.

\subsubsection{Typing rules for top-level expressions}
The typing rules for top-level expressions (the syntax in
Figure~\ref{fig:syntax}) are in Figure~\ref{fig:typingexpr}. They use
the following notation for interpreting paths relative to type
environments, analogously to
Definition~\ref{def:heaplocations} for heaps.
\begin{definition}[Locations in environments]\hfil
\label{def:envlocations}
  \begin{itemize}
  \item If $\Gamma = \Gamma', o : T$ then we define $\Gamma(o) = T$
    and $\changetype{\Gamma}{o}{T'} = \Gamma', o : T'$
  \item Inductively, if $r = r'.f$, and if $\Gamma(r') =
    \objecttype{C}{F}$ where $F$ is a record field typing containing
    $f$, then $\Gamma(r)$ is defined as $F(f)$ and
    $\changetype{\Gamma}{r}{T'}$ as 
    $\changetype{\Gamma}{r'}{\objecttype C {\changetype{F}{f}{T'}}}$.
  \item In any other case, these operations are not defined. In
    particular, if $\Gamma(r')$ is defined but is a session type, then
    $\Gamma(r'.f)$ is not defined for any $f$.
  \end{itemize}
\end{definition}

\noindent A pair $\envref\Gamma r$ only makes sense if $\Gamma(r)$ is defined
and is of the form $\objecttype C F$.

We now comment on these rules: \Tnull\
and \Tlabel\ type constants. A label is given a singleton enumerated
type, which is the smallest type it can have, but subsumption can be
used to increase its type.  \Tnew\ types a new object, giving it the
initial session type from the class declaration.  \TlinVar\ and \Tvar\
are used to access a method's parameter, removing it from the
environment if it has an object type (which is linear). For
simplicity, this is the only way to use a parameter. In particular, we
do not allow calling methods directly on parameters: to call a method
on a parameter, it must first be assigned to a field. This is just a
simplification for this formal presentation and does not limit
expressivity; this will be discussed in Section~\ref{sec:seq-extensions}.

\Tswap\ types the combined read-write field access operation,
exchanging the types of the field and expression. There are two
restrictions on its use. $T$ is not allowed to be $\linkthis$, because
this particular type only makes sense as the return type of a method.
This condition effectively forbids the use of rule \TvarF\ in the
typing derivation for $e$, because $e$ is not what the method returns.
It also has the consequence that $\Gamma'(r')$ is not a variant type,
because the only rule that could produce one is \TvarF; hence
$\Gamma'(r'.f)$ is defined. The other condition is that $T'$ is not
allowed to be a variant type, because it is not allowed in our system
to extract from a field a variantly-typed value without having
$\switchterm$ed on the associated tag first. Indeed, $\linkterm$ types
refer to fields by name, so moving variantly-typed values around would
lose the connection between the value and its tag.

\Tcall\ checks that field $f$ has a session type that allows method
$m_j$ to be called. The type of the parameter is checked as usual, and
the final type environment is updated to contain the new session type
of the object in $f$. If the return type of the method is $\linkthis$
(method $\mathsf{open}()$ in Figure~\ref{fig:session-diagram} is one
such example), it means that the value returned is a label describing
the state of this object; since the object is in $f$, it is changed
into $\linktype{}{f}$. Because the return type appears in the session
type and is therefore expressed in the top-level syntax, it cannot
already be of the form $\linktype{}{f}$. Observe that although types
of the form $\linktype{}{f}$ are not written by the programmer, they
can nevertheless occur in a typechecking derivation as the types of
top-level expressions (as in the example in Figure~\ref{fig:reduction-ex}).

\Tseq\ accounts for the effects of the first expression on the
environment and checks that a label is not discarded, which would
leave the associated variant unusable.

\Tswitch\ types a $\switchterm$ whose expression $e$ does not have a
$\linkterm$ type. All relevant branches are required to have the same
type and final environment, and the whole $\switchterm$ expression
inherits them. A typical example is if the branches just contain
different labels: in that case they are given singleton types by
\Tlabel\ and then \Tsub\ is used to give all of them the same
enumerated type. If the type $E'$ of the parameter expression is
strictly smaller than the set $E$ of case labels in the $\switchterm$
expression, branches corresponding to the extra cases are ignored.

\TswitchLink\ is the only rule for deconstructing variants. It types a
$\switchterm$, similarly to the previous one, but the type of $e$ must
be a $\linkterm$ to a field $f$ with a variant session type. The
relevant branches are then typed with initial environments containing
the different case types for $f$ according to the value of the
label. As before, they must all have the same type and final
environment, and if the $\switchterm$ expression defines extra
branches for labels which do not appear in the variant type of $f$,
they are ignored.


\TvarF\ constructs a variant field typing for the current object. Here
$E$ is typically, but not necessarily, a singleton type, and $e$ is
typically a literal label. The field typing before applying the rule
must be a record as nested variants are not permitted, and the rule
transforms it into a variant with identical cases for all labels in
$E$. It can then be extended to a variant with arbitrary other cases
using rule \TsubEnv. This rule is used for methods leading to variant
session types, which, as Definition~\ref{def:sessionfield} implies,
must finish with a variant field typing.  As a simple example,
consider the following expression, which could end a method body in
some class $D$:
\[
\switchterm~(e)~\{\mkterm{TRUE} : \seq{\swap{f}{\new
    C}}{\mkterm{OK}}\quad\mid\quad \mkterm{FALSE} :
\seq{\swap{f}\nullterm}{\mkterm{ERROR}}\}
\] 
If $S$ is the declared session type of class $C$, we have, using rules
\Tnew, \Tswap, \Tlabel\ and \Tseq, the following judgements ($T$ is just the
initial type of $f$):
\[
\judgment \envref{\this:\objecttype D {T~f}}{\this} >
\seq{\swap{f}{\new C}}{\mkterm{OK}} : \{\mkterm{OK}\} <
\envref{\this:\objecttype D {S~f}}{\this} /
\] and
\[
\judgment \envref{\this:\objecttype D {T~f}}{\this} >
\seq{\swap{f}{\nullterm}}{\mkterm{ERROR}} : \{\mkterm{ERROR}\} <
\envref{\this:\objecttype D {\nulltype~f}}{\this} /.
\]
Then \TvarF\ can be applied to both these judgements, giving both
expressions the same type $\linkthis$, and giving at the same time to
$\this$, in the final environment, the variant types
$\objecttype{D}{\vfield{\mkterm{OK}}{\{S~f\}}{}}$ and
$\objecttype{D}{\vfield{\mkterm{ERROR}}{\{\nulltype~f\}}{}}$
respectively. These two types are both subtypes of the combined
variant $\objecttype{D}{\langle\mkterm{OK} : \{S~f\}, \mkterm{ERROR} :
  \{\nulltype~f\}\rangle}$ and \textsc{T-SubEnv} can thus be applied
to both judgements to increase the final type of $\this$ to this
common supertype. It is then possible to use \textsc{T-Switch} to type
the whole expression.  Note that the final type of the expression is
always $\linkthis$: as \textsc{T-VarF} is the only rule for
constructing variants, this is the only possible return type for a
method leading to a variant.

\Tsub\ is a standard subsumption rule, and \TsubEnv\ allows
subsumption in the final environment. The main use of the latter rule,
as illustrated above, is to enable the branches of a $\switchterm$ to
be given the same final environments.

\subsubsection{Typing rules for internal expressions, heaps and
  states}\label{sec:typingfull}
\begin{figure}
\centering\setlength\lineskip{\typingRuleSkip}\setlength\baselineskip{0pt}
\AxiomC{$r\neq o.\vec{f}$}
\rulename{T-Ref}
\UnaryInfC{\judgment \envref{\Gamma, o : T}{r} > o : T < \envref\Gamma r /}
\DisplayProof\hfil
\AxiomC{\judgment\envref\Gamma r > e : E < \envref{\Gamma'}{r'} /}
\AxiomC{$\Gamma'(r'.f) = S$}
\AxiomC{$S$ is a branch}
\rulename{T-VarS}
\TrinaryInfC{\judgment \envref\Gamma r> e : \linktype{}{f} <
  \envref{\changetype{\Gamma'}{r'.f}{\choice{l}{S}{l\in E}}}{r'} /}
\DisplayProof
\hfil
\let\oldvskip\extraVskip
\def\extraVskip{0.5pt} 
\AxiomC{\judgment \envref\Gamma r > e : T < \envref{\Gamma'}{r'.f} /}
\AxiomC{$\Gamma'(r'.f) = \objecttype{C}{F}$}
\AxiomC{$\typedsess F C S$}
\noLine
\TrinaryInfC{$\quad
T\neq\linktype{}{\_}\qquad T' = \linktype{}{f}$ if $T = \linkthis$,
$T' = T$ otherwise$\quad$}
\let\extraVskip\oldvskip
\rulename{T-Return}
\UnaryInfC{\judgment \envref\Gamma r > \return e {} : T' <
  \envref{\changetype{\Gamma'}{r'.f}{S}}{r'} /}
\DisplayProof
\caption{Typing rules for expressions in the internal language}
\label{fig:typingfull}
\end{figure}
       


\begin{figure}
\centering\setlength\lineskip{\typingRuleSkip}
\axiomname{T-Hempty}
\typedheap{\emptyset}{\varepsilon}
\hfil
\AxiomC{\typedheap{\Gamma}{h}}
\AxiomC{\judgment \envref{\Gamma, o : \objecttype{C}{\iset{\nulltype\,f_i}{1\leqslant i\leqslant n}}}{o} >
  \seq{\seq{\swap{f_1}{v_1}}{\ldots}}{\swap{f_n}{v_n}} : \nulltype
  < \envref{\Gamma'}{o} /}
\rulename{T-Hadd}
\BinaryInfC{\typedheap{\Gamma'}
  {\hadd{h}{o = \hentry{C}{\iset{f_i = v_i}{1\leqslant i\leqslant n}}}}
}
\DisplayProof\hfil
%
%
\AxiomC{\typedheap{\Gamma, o : \objecttype C F}{h}}
\AxiomC{\typedsess{F}{C}{S}}
\rulename{T-Hide}
\BinaryInfC{\typedheap{\Gamma, o : S}{h}}
\DisplayProof\hfil
%
\AxiomC{\typedheap\Gamma h}
\AxiomC{\judgment\envref\Gamma r > e : T < \envref{\Gamma'}{r'} /}
\rulename{T-State}
\BinaryInfC{\judgment\Gamma > \state h r e : T < \envref{\Gamma'}{r'} /}
\DisplayProof
\caption{Typing rules for states}
\label{fig:typingstates}
\end{figure}


The type system described so far is all we need to type check class
declarations and hence programs, which are sequences of class
declarations. In order to describe the runtime consequences of
well-typedness, we now introduce an extended set of typing rules for
expressions that occur only at runtime (Figure~\ref{fig:typingfull})
and for program states including heaps
(Figure~\ref{fig:typingstates}). When typing an expression $e$ as part
of a runtime state, the path $r$, which was always $\this$ when typing
programs, varies and indicates the currently active object (the one
the method at the top of the call stack belongs to). Any difference
between $r$ and $r'$ means that $e$ contains $\returnterm$; in that
case, $r$ and $r'$ represent the call stack during and after a method
call. Recall that $\returnterm$ expressions are what a method call
reduces to and are introduced by \Rcall\ and suppressed by
\Rreturn. Therefore, these expressions can be nested and can appear,
at runtime, in any part of the expression in which reduction can
happen. For example they can appear in the argument of a $\switchterm$
or of a function call, but not in the second term of a sequence (which
does not reduce until the sequence itself has reduced). In the rules
of Figure~\ref{fig:typingexpr}, the difference between $r$ and $r'$ in
some, but not all, of the premises, accounts for that fact.

Runtime expressions may contain object identifiers, typed by \Tref. In
this rule, the current object path $r$ must not be within $o$, meaning
that the current object or any object containing it cannot be used
within an expression. This is part of the linear control of objects:
somewhere there must be a reference to the object at $r$, in order for
a method to have been called on that object, which is what gives rise
to the evaluation of an expression whose current object path is
$r$. So obtaining another reference to the object at $r$, within the
active expression, would violate linearity.

Another new rule is \TvarS, which constructs a variant session type
for a field of the current object. At the top level, the only
expression capable of constructing a variant session type is a method
call, but once the method call has reduced into something else this
rule is necessary for type preservation.

The last additional rule for expressions is \Treturn, which types a
$\returnterm$ expression representing an ongoing method call. The
subexpression $e$ represents an intermediate state in a method of
object $r'.f$. If $e$ itself does not contain $\returnterm$, we have
$r = r'.f$; otherwise they are different. For example, if we have an
expression of the form $f.m()$, the body of $m$ is $f'.m'()$, and the
body of $m'$ is $e'$ (omitting parameters for simplicity), then
$\state{h}{r'}{f.m()}$ reduces to $\state{h}{r'.f.f'}{\return{(\return
  {e'}{})}{}}$ in two steps of \Rcall. The typing derivation for this
last expression would look like:

\smallskip
\noindent
\begin{small}
\AxiomC{$\judgment \envref{\Gamma}{r'.f.f'} > e' : T <
  \envref{\Gamma'}{r'.f.f'} /$}
\AxiomC{$\typedsess {F'} {C'} {S'}$}
\noLine
\UnaryInfC{$\Gamma'(r'.f.f') = \objecttype{C'}{F'}$}
\BinaryInfC{\judgment \envref{\Gamma}{r'.f.f'} > \return{e'}{} : T <
  \envref{\changetype{\Gamma'}{r'.f.f'}{S'}}{r'.f} /}
\AxiomC{$\typedsess F C S$}
\noLine
\UnaryInfC{$(\changetype{\Gamma'}{r'.f.f'}{S'})(r'.f) = \objecttype{C}{F}$}
\BinaryInfC{\judgment \envref{\Gamma}{r'.f.f'} > \return{(\return
    {e'}{})}{} : T < \envref{\changetype{\Gamma'}{r'.f}{S}}{r'} /}
\DisplayProof
\end{small}

\smallskip
As $e$ is an intermediate state in a method of $r'.f$, it is typed
with final current object $r'.f$ and a final environment where the
type of $r'.f$ is of the form $\objecttype C F$, representing an
inside view of the object, where the fields are visible. This
\Treturn\ rule
then steps out of the object, hides its fields and changes its type
into the outside view of a session type, which must be consistent with
the internal type ($\typedsess F C S$). The particular case where $T$
is $\linkthis$ is the same as in \Tcall. $T$ is not allowed to already
be of the form $\linktype{}{f'}$ since it would break encapsulation
($f'$ would refer to a field of $r'.f$ which is not known outside of
the object).

An important point is that the only expression that changes the
current object is $\returnterm$. Several rules besides \Treturn\ can
inherit in the conclusion a change of current object from a
subexpression in a premise, but they do not add further changes. Thus
the final current object path is always a prefix of the initial one,
and the number of field specifications removed is equal to the number
of $\returnterm$s contained in the expression.  Also note that the
second part of a sequence and the branches of a $\switchterm$ are not
reduction contexts; therefore, they should not contain $\returnterm$
and are not allowed by the rules to change the current object.

As we saw in Section \ref{sec:reduction}, a runtime state consists of
a heap, a current object path, and a runtime expression.
Figure~\ref{fig:typingstates} describes how these parts are related by
typing: by rule \Tstate, a typing judgement for the expression gives
one for the state provided the current object is the same and the
initial environment reflects the content of the heap; this last
constraint is represented by the judgement \typedheap \Gamma h. Such a
judgement is constructed starting from the axiom \THempty\ which types
an empty heap and adding objects into the heap one by one with rule
\THadd, converting their types into sessions using \THide\ as needed.
As \THadd\ is the only rule that adds to $\Gamma$, we have the
property that \typedheap \Gamma h implies that every identifier in
$\Gamma$ also appears in $h$.

\THadd\ essentially says that adding a new object with given field
values to the heap affects the environment in the same way as an
expression that starts from an empty object and puts the values into
the fields one by one.  The most important feature of this rule is
that whenever a $v_i$ is an object identifier, the typing derivation
for the expression has to use \Tref, which implies both that the
initial environment contains $v_i$ and that the final one, which
represents the type of the extended heap, does not. This means that a
type environment corresponding to a heap never contains entries for
object identifiers that appear in fields of other objects, and it also
implies that a heap with multiple references to the same object is not
typable.  The numbering of the fields in the rightmost premise is
arbitrary, meaning it must not be interpreted as requiring the
sequence of swaps to be done in any particular order; all possible
orders are valid instances of the premise. This is important if the
type of the object being added is to contain links and variants:
suppose that field $f$ contains an object $o$ and field $g$ a label
$l$; it must be possible to attribute a variant type to $f$ and the
type $\linktype{}f$ to $g$, but this can only be done as a result of
typing the sequence of swaps if $\swap f o$ occurs before $\swap g l$.


\subsection{Example of reduction and typing}\label{sec:example-red}

\begin{figure*}[t]
\[
\begin{array}{rll}
\envref{\mkterm{top}:\objecttype{C}{{\branch{m_i}{S_i}{i\in I}}~f,\nulltype~g}}{\mkterm{top}} & \triangleright & \seq{\swap{g}{f.m_j()}}{\switchterm~(\swap{g}{\nullterm})~\{l\colon e_l\}_{l\in E}}
\\
&\downarrow& \scriptsize{(\Rcall)}
\\
\envref{\mkterm{top}:\objecttype{C}{\objecttype{C'}{F}~f,\nulltype~g}}{\mkterm{top}.f} & \triangleright & \seq{\swap{g}{\return{e}{}}}{\switchterm~(\swap{g}{\nullterm})~\{l\colon e_l\}_{l\in E}}
\\
&\downarrow*
\\
\envref{\mkterm{top}:\objecttype{C}{\objecttype{C'}{F_{l_0}}~f,\nulltype~g}}{\mkterm{top}.f} & \triangleright & \seq{\swap{g}{\return{l_0}{}}}{\switchterm~(\swap{g}{\nullterm})~\{l\colon e_l\}_{l\in E}}
\\
&\downarrow& \scriptsize{(\Rreturn)}
\\
\envref{\mkterm{top}:\objecttype{C}{{S_{l_0}}~f,\nulltype~g}}{\mkterm{top}} & \triangleright & \seq{\swap{g}{l_0}}{\switchterm~(\swap{g}{\nullterm})~\{l\colon e_l\}_{l\in E}}
\\
&\downarrow& \scriptsize{(\Rswap,\Rseq)}
\\
\envref{\mkterm{top}:\objecttype{C}{{\choice{l}{S_l}{l\in E}}~f,(\linktype{}{f})~g}}{\mkterm{top}} & \triangleright & \switchterm~(\swap{g}{\nullterm})~\{l\colon e_l\}_{l\in E}
\\
&\downarrow& \scriptsize{(\Rswap)}
\\
\envref{\mkterm{top}:\objecttype{C}{{S_{l_0}}~f,\nulltype~g}}{\mkterm{top}} & \triangleright & \switchterm~(l_0)~\{l\colon e_l\}_{l\in E}
\\
&\downarrow& \scriptsize{(\Rswitch)}
\\
\envref{\mkterm{top}:\objecttype{C}{{S_{l_0}}~f,\nulltype~g}}{\mkterm{top}} & \triangleright & e_{l_0}
\end{array}
\]
\caption{Example of the interplay between method call, $\switchterm$
  and $\linkterm$ types. The heap and the rightmost typing environment
  are omitted.}
\label{fig:reduction-ex}
\end{figure*}


We now return to the example from Section~\ref{sec:reduction-ex}, to
illustrate the way in which the environment used to type an expression
changes as the expression reduces (see
Theorem~\ref{thm:subjectreductionSeq} on
page~\pageref{thm:subjectreductionSeq}).
To shorten the series of steps in which the current object path does
not change, Figure~\ref{fig:reduction-ex} starts from the point at
which the initial expression has reduced to
\begin{equation*}
\seq{g = f.m_j()}{\switchterm~(g)~\{l\colon e_l\}_{l\in E}}.
\end{equation*}
Recall that this
expression is an abbreviation for
\begin{equation*}
\seq{\seq{\swap{g}{f.m_j()}}{\nullterm}}{\switchterm~(\swap{g}{\nullterm})~\{l\colon e_l\}_{l\in E}}
\end{equation*}
which we simplify to
\begin{equation*}
\seq{\swap{g}{f.m_j()}}{\switchterm~(\swap{g}{\nullterm})~\{l\colon e_l\}_{l\in E}}
\end{equation*}
The initial typing environment is
\begin{equation*}
\mkterm{top}:\objecttype{C}{{\branch{m_i}{S_i}{i\in I}}~f,\nulltype~g}
\end{equation*}
with $\mkterm{top}$ as the current object,
where $S_j = \vfield{l}{S_l}{l\in E}$. The body of method $m_j$ is
$e$ with the typing
\begin{equation*}
\judgment \envref{\this : \objecttype{C'}{F}}{\this} > e : \linkthis <
    \envref{\this : \objecttype{C'}{\vfield{l}{F_l}{l\in E}}}{\this} /
\end{equation*}
and we assume that $m_j$ returns $l_0\in E$. According to
Definition~\ref{def:sessionfield} and the typing of the declaration of
class $C'$ we have
$\typedsess{F_{l_0}}{C'}{S_{l_0}}$ and
$\typedsess{F}{C'}{\branch{m_i}{S_i}{i\in I}}$.

The figure shows the environment in which each expression is typed;
the environment changes as reduction proceeds, for several reasons
explained below. The typing of an expression is $\judgment \Gamma > e
: T < \Gamma'/$ but we only show $\Gamma$ because $\Gamma'$ does not
change and $T$ is not the interesting part of this example. We also
omit the heap, showing the typing of expressions instead of
states. However, an important point to keep in mind is that $\Gamma$
corresponds to the typing environment obtained \emph{after} typing the
heap: $\typedheap h \Gamma$ is obtained after a number of
\textsc{T-Hadd} steps and corresponds to the final typing environment
for the heap.

Calling $f.m_j()$ changes the type of field $f$ to
$\objecttype{C'}{F}$ because we are now inside the object; the current
object path changes from $\mkterm{top}$ to $\mkterm{top}.f$. As $e$ reduces to $l_0$ the
type of $f$ may change, finally becoming $\objecttype{C'}{F_{l_0}}$ so
that it has the component of the variant field typing
$\objecttype{C'}{\vfield{l}{F_l}{l\in E}}$
corresponding to $l_0$. The reduction by
$\Rreturn$ changes the type of $f$ to ${S_{l_0}}$
because we are now outside the object again, but the type is still the
component of a variant typing corresponding to $l_0$. At this point
$f$ is popped from the current object path. Here the type of $l_0$ is
$\linktype{E}{f}$ (which was the expected result type for $f.m_j()$), and
this type is obtained by applying \textsc{T-VarS}, so in the
intermediate typing environment \emph{after} typing $l_0$, $f$ has the
variant type ${\choice{l}{S_l}{l\in E}}$.

The next step, swap, moves $l_0$ from the expression to the heap.
Therefore the application of \textsc{T-VarS} needed to type it now
occurs in the derivation for typing the heap, of which $\Gamma$ is the
result. This is why in $\Gamma$ the type of $f$ is now
${\choice{l}{S_l}{l\in E}}$, which is
${S_j}$, the type we were expecting after the method
call. At this point the information about which component of the
variant typing we have is stored in $\mkterm{top}.g$, the field the label was
swapped into: the type of the expression $f.m_j()$ is
$\linktype{E}{f}$, which appears as the type of $\mkterm{top}.g$ after the swap
is executed. When extracting the value of $g$ in order to
$\switchterm$ on it, the type $\linktype{E}{f}$ disappears from the
environment and becomes the type of the subexpression
$\swap{g}{\nullterm}$, at the same time resolving the variant type of
$f$ according to the particular enumerated value $l_0$.

\subsection{Typing the initial state}\label{sec:initial-typing}

Recall the discussion of the initial state for execution of a program,
from the end of Section~\ref{sec:reduction}. The initial state is
$\state{\mkterm{top}=\hentry{C}{C.\fieldsterm=\vec\nullterm}}{\mkterm{top}}{e}$
where class $C$ has a designated \emph{main} method $m$ with body
$e$. In order to type this initial state, we require that $m$ is
immediately available in $C.\sessterm$, and assume that the program is
typable, \ie that rule \Tclass\ is applicable to every class
definition. If $C.\fieldsterm = \vec{f}$ then the hypothesis of
\Tclass\ is
$\typedsess{\overrightarrow\nulltype\,\vec{f}}{C}{C.\sessterm}$; this
is what the type checking algorithm defined in
Section~\ref{sec:algorithm} checks. The definition of
$\typedsess{F}{C}{S}$ gives \judgment \envref{\this:
  \objecttype{C}{\overrightarrow\nulltype\,\vec{f}}, x :
  \nulltype}{\this} > e : T < \envref{\this: \objecttype{C}{F}}{\this}
/ for some field typing $F$. The type $T$ is
irrelevant. Lemma~\ref{lem:substitution} (Substitution), to be proved
later, gives \judgment \envref{\mkterm{top}:
  \objecttype{C}{\overrightarrow\nulltype\,\vec{f}}}{\mkterm{top}} > e
: T < \envref{\mkterm{top}: \objecttype{C}{F}}{\mkterm{top}} /, as we
assumed that $e\subs{\nullterm}{x} = e$. Straightforward use of
\THempty\ and \THadd\ gives
$\typedheap{\mkterm{top}:\objecttype{C}{\overrightarrow\nulltype\,\vec{f}}}{\mkterm{top}=\hentry{C}{C.\fieldsterm=\vec\nullterm}}$
and then \Tstate\ gives a typing for the initial state.

\subsection{Properties of the type system}
\label{sec:properties}

The main results in this sequential setting are standard: type
preservation under reduction (also known as Subject Reduction) and
absence of stuck states for well-typed programs. Furthermore, the
system also enjoys of a conformance property: all executions of
well-typed programs follow what is specified by the classes' session
types.

\subsubsection{Soundness of subtyping}

In this section we prove that the subtyping relation is sound with
respect to the type system, in the sense that it preserves not only
typing judgements but also consistency between field typings and
session typings, reflecting the safe substitution property.

\begin{lemma}\label{lem:moreweak}
  If \judgment \envref \Gamma r > e : T < \envref{\Gamma'}{r'} / and
  $\Gamma''\subt\Gamma$, then \judgment \envref{\Gamma''}{r} > e : T <
  \envref{\Gamma'}{r'} /.
\end{lemma}
\begin{proof}
  Straightforward induction on the derivation.
\end{proof}

\begin{proposition}[Soundness of subtyping for fields]\label{prop:subfield}
  If \typedsess F C S and $F'\subt F$ then \typedsess {F'} C S.
\end{proposition}
\begin{proof}
Straightforward using Lemma \ref{lem:moreweak}.
\end{proof}

Before proving the case where subtyping is on the right,
we first remark that, similarly to sub-session, the necessary
conditions in the definition of $C$-consistency (Definition
\ref{def:sessionfield}) become sufficient
once we consider the largest relation:
\begin{lemma}\label{lem:sessionfield}
  Let $C$ be a class and $F$ a field typing for that class.
  \begin{enumerate}
  \item Suppose $F$ is not a variant, and suppose there is a set of
    method definitions\linebreak $\iset{\method{}{m_i}{x_i}{e_i}}{i\in I}$ in
    the declaration of class $C$ such that, for all $i$, we have:
  $$\judgment \envref{\this: \objecttype C F, x_i : T'_i}{\this}  > e_i
  : T_i < \envref{\this: \objecttype C {F_i}}{\this} /$$
  with \typedsess{F_i}{C}{S_i}. Then
  \typedsess{F}{C}{\branch{\methsign{m_i}{T'_i}{T_i}}{S_i}{i\in I}}
  holds.
\item Suppose $F = \vfield{l}{F_l}{l\in E'}$ and let $(S_l)_{l\in E}$
  be a family of session types such that $E'\subset E$ and
  $\typedsess {F_l}{C}{S_l}$ for all $l\in E'$. Then \typedsess F C
  {\choice{l}{S_l}{l\in E}} holds.
  \end{enumerate}
\end{lemma}
\begin{proof}
  Let $S$ be either $\branch{\methsign{m_i}{T'_i}{T_i}}{S_i}{i\in I}$
  or $\choice{l}{S_l}{l\in E}$ depending on the case. Just notice that 
  $(\typedsess{\bullet}{C}{\bullet})\cup\{(F, S)\}$ is
  a $C$-consistency relation.
\end{proof}

\begin{proposition}[soundness of subtyping for sessions]\label{prop:subsess}
  If \typedsess F C S and $S\subt S'$ then \typedsess F C {S'}.
\end{proposition}
\begin{proof}
  For any class $C$, we define the following relation:
$$\mathcal{R}_C = \{(F,S')\mid\exists S,\, \typedsess F C S \text{ and }
S\subt S'\}$$
and prove that it is a $C$-consistency relation (Definition
  \ref{def:sessionfield}). Let $(F,S')\in\mathcal{R}_C$, and let
$S$ be as given by the definition of the relation. We have two cases
depending on the form of $S'$ (branch or variant).

The first one is $S' = \branch{\methsign{m_j}{T'_j}{U'_j}}{S'_j}{j\in J}$.
  Then $S\subt S'$ means (Definition \ref{def:subsession}) that we have:
  $S = \branch{\methsign{m_i}{T_i}{U_i}}{S_i}{i\in I}$ with
  $J\subset I$ and for all $j\in J$, $\methsign{m_j}{T_j}{U_j}:S_j$ is
  $\subt$-compatible with $\methsign{m_j}{T'_j}{U'_j}:S'_j$.
  Let $j\in J$, we know from \typedsess F C S
  that $C$ contains a method declaration $\method{}{m_j}{x}{e}$
  such that the following judgement:
  $$\judgment \envref{x : T_j, \this : \objecttype C F}{\this} > e : U_j <
  \envref{\this : \objecttype{C}{F_j}}{\this} /$$
  holds, with $\typedsess {F_j}{C}{S_j}$.
  $\subt$-compatibility between the two signatures
  of $m_j$ (Definition
  \ref{def:Rcompatibility-signatures}) gives us, first,
  $T'_j\subt T_j$, which
  allows us to apply Lemma~\ref{lem:moreweak} to this judgement and
  replace $T_j$ by $T'_j$ in it, and second, either:
  \begin{enumerate}
  \item $U_j\subt U'_j$ and $S_j\subt S'_j$. 
    The former allows us to use \textsc{T-Sub} to replace $U_j$ by
    $U'_j$ in the typing judgement for $e$, fulfilling the first
    condition in the definition of $C$-consistency. The latter, together with
    \typedsess{F_j}{C}{S_j}, implies $(F_j, S'_j)\in\mathcal{R}_C$,
    fulfilling the second one.
  \item $U_j$ is an enumerated type $E$, $U'_j = \linkthis$ and
    $\choice{l}{S_j}{l\in E}\subt S'_j$. In this case we first apply
    \textsc{T-VarF} to the judgement, yielding:
    $$\judgment \envref{x : T'_j, \this : \objecttype C F}{\this} > e :
    \linkthis < \envref{\this : \objecttype{C}{\vfield{l}{F_j}{l\in
          E}}}{\this}. /$$
    From \typedsess{F_j}{C}{S_j} we deduce
    \typedsess{\vfield{l}{F_j}{l\in E}}{C}{\choice{l}{S_j}{l\in E}}
    using Lemma \ref{lem:sessionfield}, and conclude
    $(\vfield{l}{F_j}{l\in E}, S'_j) \in \mathcal{R}_C$.
  \end{enumerate}

  In the second case, where $S' = \choice{l}{S'_l}{l\in E'}$, then
  $S = \choice{l}{S_l}{l\in E}$ with $E\subset E'$ and
  $\forall l\in E, S_l\subt S'_l$. From \typedsess F C S we know that
  $F = \vfield{l}{F_l}{l\in E''}$ with $E''\subset E$ and
  $\typedsess{F_l}{C}{S_l}$ for any $l$ in $E''$.  Just notice that
  $E''\subset E'$ (by transitivity of $\subseteq$) and
  $(F_l, S'_l)\in\mathcal{R}_C$ for any $l$ in $E''$.
\end{proof}

\subsubsection{Type preservation}

\begin{theorem}[Subject Reduction]
  \label{thm:subjectreductionSeq}
  Let $\mathcal{D}$ be a set of well-typed declarations, that is, such
  that for every class declaration $D$ in $\mathcal{D}$ we have
  $\vdash D$.

  If, in a context parameterised by $\mathcal{D}$, we have $\judgment
  \envref{\Gamma}{r} > \state{h}{r}{e} : T < \envref{\Gamma'}{r'} /$,
  and if $\state{h}{r}{e}\reduces \state{h'}{r''}{e'}$, then there
  exists $\Gamma''$ such that $\judgment \envref{\Gamma''}{r''} >
  \state{h'}{r''}{e'} : T < \envref{\Gamma'}{r'} /$.
\end{theorem}
\begin{proofnoqed}
  This theorem is a particular case of Theorem \ref{thm:threadsr}
  which will be proved in Section~\ref{sec:results}.
\end{proofnoqed}

\subsubsection{Type safety}

\begin{theorem}[No Stuck Expressions]
  \label{thm:typesafetySeq}
  Let $\mathcal{D}$ be a set of well-typed declarations, that is, such
  that for every class declaration $D$ in $\mathcal{D}$ we have
  $\vdash D$.

  If, in a context parameterised by $\mathcal{D}$, we have \judgment
  \envref{\Gamma}{r} > \state{h}{r}{e} : T < \envref{\Gamma'}{r'} /,
  then either $e$ is a value or there exists $\state{h'}{r''}{e'}$
  such that $\state{h}{r}{e}\reduces \state{h'}{r''}{e'}$.
\end{theorem}
\begin{proofnoqed}
This theorem is also a consequence of Theorem \ref{thm:threadsr}, so
we postpone its proof until Section \ref{sec:results}.
\end{proofnoqed}

\subsubsection{Conformance}
\comment{NG: updated 8.4.2011}

We show that, in well-typed programs, \emph{a sequence} of method calls
(interleaved with their respective return labels) of a given class \emph{is
a path} of its session type. In order to state this property
precisely, we introduce a few definitions.

\begin{definition}[Call trace]
\label{def:calltrace}
  A \emph{call trace} is a sequence $m_1l_1m_2l_2\ldots m_nl_n$ in
  which each $m_i$ is a method name and each $l_i$ may be absent or,
  if present, is a label.
\end{definition}

\begin{definition}[LTS on session types]
  Define a labelled transition relation on class session types by the
  following rules. $\alpha$ stands for $m$ or $l$.
  \begin{align*}
    \frac{j\in
      I}{\labtrans{\branch{\methsign{m_i}{T_i}{T'_i}}{S_i}{i\in
          I}}{m_j}{S_j}} & & \frac{l_0\in
      E}{\labtrans{\choice{l}{S_l}{l\in E}}{l_0}{S_{l_0}}}
    \\~\\
    \frac{S\text{ is not a variant}}{\labtrans{S}{l}{S}} & &
    \frac{\labtrans{S\subs{\mu X.S}{X}}{\alpha}{S'}}{\labtrans{\mu
        X.S}{\alpha}{S'}}
  \end{align*}
\end{definition}

\begin{definition}[Call trace mapping]
A \emph{call trace mapping} for a heap $h$ is a function $\tr$ from
$\dom(h)$ to call traces.
\end{definition}

\begin{definition}[Validity of mappings]
  A call trace mapping $\tr$ for a heap $h$ is \emph{valid} if for
  every entry $o = \hentry{C}{\ldots}$ in $h$, we have
  $\labtransstar{C.\sessterm}{\tr(o)}{}$. An element in a call trace
  which does not allow the corresponding session type to reduce is a
  \emph{type error}. (Thus a call trace is valid if and only if it
  does not contain type errors).
\end{definition}

\begin{definition}
If $\tr$ is a call trace mapping for a heap $h$ then we define
$\tr(h,r)$ for references $r$ such that $h(r)$ is defined, as follows:
\begin{align*}
\tr(h,o) & = \tr(o) \\
\tr(h,r.f) & = \tr(h(r).f)
\end{align*}
\end{definition}

\begin{definition}[Original reduction rule]
  If $\state{h}{r}{e}\reduces\state{h'}{r'}{e'}$ then the derivation
  of this reduction consists of a number of applications of
  \textsc{R-Context}, preceded by another rule which forms a unique
  leaf node in the derivation. We say that the rule at the leaf node
  is the \emph{original reduction rule} for the reduction, or that the
  reduction \emph{originates from} this rule.
\end{definition}

\begin{definition}[Extension of call traces]
\label{def:calltraceextension}
  Suppose $\tr$ is a call trace mapping for $h$ and
  $\state{h}{r}{e}\reduces\state{h'}{r'}{e'}$. Define a call trace mapping
  $\tr'$ for $h'$ as follows: 
\begin{itemize}
\item If the reduction originates from \textsc{R-Call} with method $m$
  and field $f$ then $\tr' = \tr\{h(r.f)\mapsto\tr(h(r.f))m\}$.
\item If the reduction originates from \textsc{R-Return} with value
  $v$, and $v$ is a label $l$, then $\tr' =
  \tr\{h(r)\mapsto\tr(h(r))l\}$.
\item If the reduction originates from \textsc{R-New} and the fresh
  object is $o$ then $\tr' = \tr\{o\mapsto\varepsilon\}$.
\item Otherwise, $\tr' = \tr$.
\end{itemize}
\end{definition}

The conformance property is the following: in a sequence of reductions
starting from the initial state of a well-typed program, the call
traces built using the extension mechanism defined above are valid
throughout the sequence. We need a couple of lemmas to
properly relate call traces and typings in the case of variant types.

\begin{lemma}\label{lem:heapfield}
  If \typedheap\Gamma h, then $\Gamma$ does not contain type
  $\linkthis$ or any variant field typing.
\end{lemma}
\begin{proof}
  It suffices to show that rule \textsc{T-VarF} cannot be used in the
  derivation of \typedheap\Gamma h, since it is the only rule that
  introduces $\linkthis$ or variant field typings.

  This rule can only be used on an expression of enumerated type,
  and the only place where such an expression can occur
  in the derivation of \typedheap\Gamma h is as the right member of a
  swap in the second premise of \textsc{T-Hadd} (the swap expression
  itself has type $\nulltype$ because of the initial environment).
  It corresponds to the first
  premise of \textsc{T-Swap}. However, the third premise of
  \textsc{T-Swap} forbids that the type of the expression be
  $\linkthis$, hence \textsc{T-VarF} cannot be used there.
\end{proof}

\begin{lemma}[Variant consistency]\label{lem:choiceconsistency}
  If \typedheap\Gamma h and $\Gamma(r) = \choice{l}{S_l}{l\in E}$,
  then:
  \begin{enumerate}
  \item $r$ is of the form $r'.f$
  \item there exists $f'$ such that $\Gamma(r'.f') = \linktype{}{f}$
    and $h(r'.f') \in E$.
  \end{enumerate}
\end{lemma}
\begin{proof}
  Consider how a variant session type can be introduced in the derivation of
  \typedheap\Gamma h. Because of Lemma \ref{lem:heapfield}, it cannot
  be a consequence of \textsc{T-Hide}: indeed, \typedsess F C S where
  $S$ is a variant can only hold if $F$ is a variant as well. Thus the only
  possibility is \textsc{T-VarS}, and it can only occur when typing
  one of the values in the right premise of \textsc{T-Hadd}. (1)
  follows from the fact that \textsc{T-VarS} acts on a field of the
  current object. Then \textsc{T-SubEnv} can be applied but the
  original label \textsc{T-VarS} was applied to is still in the final $E$.
  (2) follows from the structure of the derivation: the label 
  \textsc{T-VarS} is applied to is then swapped
  into a field of the same object.
\end{proof}

\begin{definition}[Actual session type]
  Let $\Gamma$ and $h$ be such that \typedheap\Gamma h. For any $r$ in
  $\Gamma$ such that $\Gamma(r)$ is a session type $S$, we define
  $S'$, the \emph{actual session type} of $r$ in $h$ according to
  $\Gamma$, as follows:
  \begin{itemize}
  \item If $S$ is a branch then $S' = S$.
  \item If $S$ is a variant $\choice{l}{S_l}{l\in E}$, then $S' =
    S_{h(r'.f')}$, where $r'$ and $f'$ are as given by
    Lemma~\ref{lem:choiceconsistency}.
  \end{itemize}
\end{definition}

\begin{definition}[Consistency of call traces]
  Let $\tr$ be a call trace mapping for a heap $h$ and let $\Gamma$ be
  a type environment such that \typedheap\Gamma h. We say that $\tr$
  is \emph{consistent with} $\Gamma$ if for every $r$ in $\Gamma$ with
  actual session type $S$ we have
  $\labtransstar{\classterm(h(r)).\sessterm}{\tr(h,r)}{S}.$
\end{definition}

%
\begin{theorem}[Conformance]
  \label{thm:conf}
  Suppose we are in a context parameterised by a set of well-typed
  declarations.

  Let $\state{h_1}{r_1}{e_1}$ be a program state together with a valid
  call trace mapping $\tr_1$, and suppose that
  $\state{h_1}{r_1}{e_1}\reduces\cdots\reduces\state{h_n}{r_n}{e_n}$
  is a reduction sequence such that $r_1$ is a prefix of all $r_i$.
  Definition~\ref{def:calltraceextension} gives a corresponding
  sequence of call traces $\tr_i$.

  If there exists $\Gamma$ such that $\tr_1$ is consistent with
  $\Gamma$ and \judgment \envref{\Gamma}{r_1} > \state{h_1}{r_1}{e_1}
  : T < \envref{\Gamma'}{r'}/ then for all $i$, $\tr_i$ is valid.
\end{theorem}

\begin{proofnoqed}
  Postponed, again, to Section \ref{sec:results} as it makes use of
  the proof of Theorem \ref{thm:threadsr} which will be proved there.
\end{proofnoqed}

\begin{corollary}
  Given a well-typed program, starting from the initial state
  described at the end of Section~\ref{sec:reduction} with
  the initial call trace mapping $\{\mkterm{top} \mapsto m\}$,
  and given a reduction sequence from there,
  the call trace mappings obtained by
  Definition \ref{def:calltraceextension} following the reductions
  are valid throughout the sequence.
\end{corollary}
\begin{proofnoqed}
  We just have to see that:
  \begin{enumerate}
  \item the initial call trace mapping is
    valid, as the main method $m$ is required to
    appear in the initial session type of the main class;
  \item  it is also
    consistent
    with the initial typing given in Section \ref{sec:initial-typing},
    as the initial $\Gamma$ contains no session type;
  \item the
    initial current object path is reduced to an object identifier and,
    therefore, stays a prefix of the current object path throughout any
    reduction sequence.
\qed
  \end{enumerate}
\end{proofnoqed}




\section{Towards a Full Programming Language}
\label{sec:seq-extensions}

\comment{NG: updated 9.1.15}
\comment{SG/VV: updated 13.1.15}

In this section, we show how the core calculus presented in the
previous section can be extended towards a full programming language.
The extensions include constructs which can be considered abbreviations and
may be translated into the core calculus without changing it, and
actual extensions to the formal system.

\subsection{Assignment and Field Access}
As explained in the introduction of Section~\ref{sec:core-lang}, we add to expressions:\begin{itemize}
\item the field access expression $f$ (not followed by a dot or by
  $\leftrightarrow$), which translates into the core expression $\swap
  f \nullterm$. This expression evaluates to the content of $f$ and
  has the side effect of setting $f$ to $\nullterm$;
\item the assignment expression $f = e$, which translates into the
  core expression $\seq{\swap f e}{\nullterm}$.
  This expression stores the value of $e$ in field $f$ and evaluates
  to $\nullterm$.
\end{itemize}
\subsection{Multiple Parameters}
It is straightforward to generalise the reduction and typing rules so that methods have multiple parameters. In rule $\Tcall$, the environments would be threaded through a series of parameter expressions, in the same way as in rule $\Tseq$.

\subsection{Local Variables}
Local variables can be simulated by introducing extra parameters.


\subsection{While Loops}
\label{subsec:while-loops}

The language can easily be extended to include $\whileterm$ loops, by
adding the rules in Figure~\ref{fig:while}.  The reduction rule
defines $\whileterm$ recursively in terms of $\switchterm$. There are
two typing rules, derived from $\Tswitch$ and $\TswitchLink$. The
first deals with a straightforward $\whileterm$ loop that has no
interaction with session types, and the second deals with the more
interesting case in which the condition of the loop is linked to the
session type of an object.

\begin{figure}
Top-level syntax (add to Figure~\ref{fig:syntax}):
\[
e \bnf\ \dots \alt \while{e}{e}
\]

Reduction rule (add to Figure~\ref{fig:reduction}):
\begin{center}
\axiomname{R-While}
$
\state{h}{r}{\while{e}{e'}}\reduces
\state{h}{r}{\switchterm~(e)~\{\true:\seq{e'}{\while{e}{e'}}, \false:\nullterm\}}
$
\end{center}

Top-level typing rules (add to Figure~\ref{fig:typingexpr}):
\begin{center}
\AxiomC{\judgment \envref\Gamma r > e : \{\true,\false\} < \envref{\Gamma'}{r'} /}
\AxiomC{\judgment \envref{\Gamma'}{r'} > {e'} : \nulltype < \envref{\Gamma}{r} /}
\rulename{T-While}
\BinaryInfC{\judgment \envref\Gamma r > \while{e}{e'} : \nulltype
< \envref{\Gamma'}{r'} /}
\DisplayProof\hfil
\end{center}
\begin{center}
\let\oldvskip\extraVskip
\def\extraVskip{0.5pt} 
\AxiomC{\judgment \envref\Gamma r > e : \linktype{}{f} <
  \envref{\Gamma'}{r'} /}
\AxiomC{$\Gamma'(r'.f) = \langle \true:S_\true, \false:S_\false \rangle$}
\noLine
\BinaryInfC{\judgment\envref{\changetype{\Gamma'}{r'.f}{S_\true}}{r'}
  > e' : \nulltype < \envref{\Gamma}{r} /}
\let\extraVskip\oldvskip
\rulename{T-WhileLink}
\UnaryInfC{\judgment \envref\Gamma r > \while{e}{e'} : \nulltype < \envref{\changetype{\Gamma'}{r'.f}{S_\false}}{r'} /}
\DisplayProof\hfil
\end{center}
\caption{Rules for \lstinline|while|}
\label{fig:while}
\end{figure}

\begin{figure}
Top-level syntax (add to Figure~\ref{fig:syntax}) :
\begin{align*}
M \bnf\ & \ldots\alt\annotmethod{F}{F}{T}{m}{T~x}{e} \\
e \bnf\  & \dots\alt\selfcal{m}{e}
\end{align*}

Reduction rule (add to Figure~\ref{fig:reduction}):
\begin{center}
\AxiomC{$\annotmethod{\_}{\_}{\_}{m}{\_\, x}{e} \in h(r).\classterm$}
\rulename{R-SelfCall}
\UnaryInfC{$\state{h}{r}{\selfcal{m}{v}} \reduces
  \state{h}{r}{e\subs{v}{x}}$}
\DisplayProof
\end{center}

Top-level typing rules (add or replace in Figure~\ref{fig:typingexpr}):
\begin{center}
\AxiomC{$\judgmentc \envref\Gamma r > e : T < \envref{\Gamma'}{r'} /$}
\AxiomC{$\Gamma'(r') = \objecttype C F$}
\AxiomC{$\annotmethod{F}{F'}{T'}{m}{T~x}{e'}\in C$}
\rulename{T-SelfCall}
\TrinaryInfC{\judgmentc \envref\Gamma r > \selfcal{m}{e} :
  T' < \envref{\changetype{\Gamma'}{r'}{\objecttype{C}{F'}}}{r'} /}
\DisplayProof\hfil
\end{center}
\begin{center}
  \AxiomC{$\judgmentc \envref{\this:\objecttype{C}{F}, x:T'}{\this} > e : T <
    \envref{\this:\objecttype{C}{F'},x:T''}{\this} /$}
  \AxiomC{$F'\not=\langle\_\rangle$}
  \rulename{T-AnnotMeth}
  \BinaryInfC{$\typedmeth{C}{\annotmethod{F}{F'}{T}{m}{T'~x}{e}}$}
  \DisplayProof
\end{center}
\begin{center}
\AxiomC{$\typedsess{\overrightarrow\nulltype\,\vec f}{C}{S}$}
\AxiomC{$\forall M\in\vec{M}.~(M~\text{has $\reqterm/\ensterm$} \Rightarrow~\typedmeth{C}{M})$}
\rulename{T-Class}
\BinaryInfC{$\vdash\class{C}{S}
  {\vec f}
  {\vec M}$}
\DisplayProof
\end{center}








\caption{Rules for recursive methods and other self-calls}
\label{fig:self-calls}
\end{figure}


\subsection{Self-Calls and Recursive Methods}
\label{subsec:self-calls}

The rules in Figure~\ref{fig:self-calls} extend the language to
include self-calls (method calls on $\this$). This extension also
supports recursive calls, which are necessarily self-calls. Self-calls
do not check or advance the session type, and a method that is only
self-called does not appear in the session type. A method that is
self-called and called from outside appears in the session type, and
calls from outside do check and advance the session type. The reason
why it is safe to not check the session type for self-calls is that
the effect of the self-call on the field typing is included in the
effect of the method that calls it. All of the necessary checking of
session types is done because of the original outside call that
eventually leads to the self-call.

Because they are not in the session type, self-called methods must be
explicitly annotated with their initial ($\reqterm$) and final
($\ensterm$) field typings. The annotations are used to type
self-calls ($\TselfCall$) and method definitions ($\TannotMeth$). The
result type and parameter type are also specified as part of the
method definition, again because the method is not in the session type.

If a method is in the session type then its body is checked by the
first hypothesis of \textsc{T-Class}, but the annotations (if present)
are ignored except when they are needed to check recursive calls.  If
a method has an annotation then its body is checked by the second
hypothesis of \textsc{T-Class}.  If both conditions apply then the
body is checked twice. An implementation could optimize this.

An annotated method cannot produce a variant field typing or have a
$\linkterm$ type, because $\TswitchLink$ can only analyze a variant
session type, not a variant field typing. 


\subsection{Shared Types and Base Types}
\label{subsec:shared-types}

The formal language described in this paper has a very strict linear
type system. It is straightforward to add non-linear classes as an
orthogonal extension: they would not have session types and their
instances would be shared objects, treated in a completely standard
way.
Including them in
the formalisation, however, would only complicate the typing rules.

More interesting, and more challenging, is the possibility of
introducing a more refined approach to aliasing and ownership, for
example along the lines of the systems discussed in
Section~\ref{sec:related}. We intend to investigate this in the
future.

Base types such as \lstinline|int| are also straightforward to add,
and would be treated non-linearly. 

\subsection{Inheritance}
\label{sec:nominal}

The formal language uses a structural type system in which class
names are only used in order to obtain their session types; method
availability is determined solely by the session type, and method
signatures are also in the session type. In particular, the subtyping
relation is purely structural and makes no reference to class
names. It is straightforward to adapt the language to include features
associated with nominal subtyping, such as an explicitly declared
inheritance hierarchy for classes with inheritance and overriding of
method definitions. In this case, if class $C$ is declared to inherit
from class $D$, and both define session types (alternatively, $C$
might inherit its session type from $D$), then the condition
$C.\sessterm \subt D.\sessterm$ would be required in order for the
definition of $C$ to be accepted.


\section{A Distributed Example}
\label{sec:distributed-example}
\comment{SG: updated 28.3.2011}

We now present an example of a distributed system, illustrating the
way in which our language unifies session-typed channels and more
general typestate. Recall that our programming model is based on
communication over TCP/IP-style socket connections, which we refer to as
channels. The scenario is a file server, which clients can communicate
with via a channel. The file server uses a local file, represented by
a \lstinline|File| object as defined in Section~\ref{sec:seq-example},
and responds to requests such as \lstinline|OPEN| and
\lstinline|HASNEXT| on the channel. On the client side, the remote
file is represented by an object of class \lstinline|RemoteFile|,
whose interface is similar to \lstinline|File|. In this ``stub''
object, methods such as \lstinline|open| are implemented by
communicating with the file server.

The channel between the client and the server has a session type in
the standard sense \cite{HondaK:intblt}, which defines a communication
protocol. In our language, each endpoint of the channel is represented
by an object of class \lstinline|Chan|, with a class session type
derived from the channel session type. This class session type also
expresses the definition of the communication protocol, by specifying
when the methods \lstinline|send| and \lstinline|receive| are
available.

For the purpose of this example, we imagine that the communication
protocol (channel session type) is defined by the provider of the file
server, while the class session type of \lstinline|RemoteFile| is
defined by the implementor of a file system API. We therefore present
two versions of the example: one in which the channel session type,
and the class session type of \lstinline|RemoteFile|, have the same
structure; and one in which they have different structures.

\subsection{Distributed Example Version 1} 
\begin{figure}
\begin{lstlisting}
FileReadCh = &{OPEN: ?String.+{OK: OpenCh, ERROR: FileReadCh}, QUIT: End}
OpenCh = &{HASNEXT: +{TRUE: CanReadCh, FALSE: MustCloseCh}, CLOSE: FileReadCh}
MustCloseCh = &{CLOSE: FileReadCh}
CanReadCh = &{READ: !String.OpenCh, CLOSE: FileReadCh }
\end{lstlisting}
\caption{Remote file server version 1: channel session type (server side)}
\label{fig:remotefileserver1a}
\end{figure}


\begin{figure}
\begin{lstlisting}
FileRead_cl = {Null send({OPEN}): {Null send(String): {{OK, ERROR} receive():
                                                         <OK: Open_cl, 
                                                          ERROR: FileRead_cl>}},
               Null send({QUIT}): {}}
where
  Open_cl = {Null send({HASNEXT}): {{TRUE, FALSE} receive(): 
                                      <TRUE: CanRead_cl, FALSE: MustClose_cl>},
             Null send({CLOSE}): FileRead_cl}
  MustClose_cl = {Null send({CLOSE}): FileRead_cl}
  CanRead_cl = {Null send({READ}): {String receive(): Open_cl},
                Null send({CLOSE}): FileRead_cl}
\end{lstlisting}
~\\~\\
\begin{lstlisting}
FileRead_s = {{OPEN, QUIT} receive(): <OPEN: {String receive():
                                                {Null send({OK}): Open_s,
                                                 Null send({ERROR}): FileRead_s}, 
                                       QUIT: {}>}
where
  Open_s = {{HASNEXT, CLOSE} receive():
              <HASNEXT: {Null send({TRUE}): CanRead_s,
                         Null send({FALSE}): MustClose_s},
               CLOSE: FileRead_s>}
  MustClose_s = {{CLOSE} receive(): <CLOSE: FileRead_s>}
  CanRead_s = {{READ, CLOSE} receive(): <READ: {Null send(String): Open_s},
                                         CLOSE: FileRead_s>}
\end{lstlisting}
\caption{Remote file server version 1: client and server class session types
  generated from channel session type \lstinline|FileReadCh|}
\label{fig:remotefileserver1cd}
\end{figure}

 \begin{figure}
\begin{lstlisting}
class RemoteFile {
 session {connect: Init}
  where Init = {open: <OK: Open, ERROR: Init>}
        Open = {hasNext: <TRUE: Read, FALSE: Close>, close: Init }
        Read = {read: Open, close: Init }
        Close = {close: Init }

  channel;

  Null connect(<FileReadCh> server) {
    channel = server.request();
  }
  {OK,ERROR} open(String name) {
    channel.send(OPEN);
    channel.send(name);
    switch (channel.receive()) {
      OK: OK;
      ERROR: ERROR;
   }
  }
  {TRUE,FALSE} hasNext () {
    channel.send(HAS_NEXT);
    switch (channel.receive()) {
      TRUE: TRUE;
      FALSE: FALSE;
    }
  }
  String read() {
    channel.send(READ);
    channel.receive();
  }
  Null close() { 
    channel.send(CLOSE);
  }
}
\end{lstlisting}
\caption{Remote file server version 1: client side stub}
\label{fig:remotefileserver1b}
\end{figure}


\begin{figure}
\begin{lstlisting}
class FileServer {
  session { Null main(<FileReadCh> port): {} }

  channel; file;

  Null main(<FileRead_c> port) {
    file = new File();
    channel = port.accept();
    fileRead();
  }
  req FileRead_s channel, Init file
  ens {} channel, Init file
  Null fileRead() {
    switch (channel.receive()) {
      OPEN:
        switch (file.open(channel.receive())) {
          OK: open();
          ERROR: fileRead();
       }
      QUIT: null;
    }
  }
  req Open_s channel, Open file
  ens {} channel, Init file
  Null open() {
    switch (channel.receive()) {
      HASNEXT:
        switch (file.hasNext()) {
          TRUE: channel.send(TRUE); canRead();
          FALSE: channel.send(FALSE); mustClose();
        }
      CLOSE: file.close(); fileRead();
  }
  req MustClose_s channel, Close file
  ens {} channel, Init file
  Null mustClose() {
    switch (channel.receive()) {
      CLOSE: file.close(); fileRead();
    }
  }
  req CanRead_s channel, Read file
  ens {} channel, Init file
  Null canRead() {
    switch (channel.receive()) {
      READ: channel.send(file.read()); open();
      CLOSE: file.close(); fileRead();
    }
  }
}
\end{lstlisting}
\caption{Remote file server version 1: server code}
\label{fig:remotefileserver1e}
\end{figure}

 Figure~\ref{fig:remotefileserver1a}
defines a channel session type for interaction between a file server
and a client. The type of the server's endpoint is shown, and the type
\lstinline|FileReadCh| is the starting point of the protocol. The type
constructor \lstinline|&| means that the server offers a choice, in
this case between \lstinline|OPEN| and \lstinline|QUIT|; the client
makes a choice by sending one of these labels. If \lstinline|OPEN| is
selected, the server receives (constructor \lstinline|?|) a
\lstinline|String| and then (the \lstinline|.| constructor means
sequencing) the constructor \lstinline|+| indicates that the server
can choose either \lstinline|OK| or \lstinline|ERROR| by sending the
appropriate label. The remaining definitions are read in the same way;
\lstinline|End| means termination of the protocol. The type
of the client's endpoint is dual, meaning that receive (\lstinline|?|)
and send (\lstinline|!|) are exchanged, as are offer (\lstinline|&|)
and select (\lstinline|+|). When the server offers a choice, the
client must make a choice, and \emph{vice versa}.

The structure of the channel session type is similar to that of the
class session type of \lstinline|File| from
Section~\ref{sec:seq-example}, in the sense that \lstinline|HASNEXT|
is used to discover whether or not data can be read.

We regard each endpoint of a channel as an object with
\lstinline|send| and \lstinline|receive| methods. For every channel
session type there is a corresponding class session type that
specifies the availability and signatures of \lstinline|send| and
\lstinline|receive|. The general translation is defined in
Section~\ref{sec:distributed}, Figure~\ref{fig:channelobjects}. For
the particular case of \lstinline|FileReadCh|, the client and server
class session types are as defined in
Figure~\ref{fig:remotefileserver1cd}: \lstinline|FileRead_cl| for the client and \lstinline|FileRead_s| for the server.

The requirement to make a choice (\lstinline|+|) in the channel
session type corresponds to availability of \lstinline|send| with a
range of signatures, each with a parameter type representing one of
the possible labels; here we are taking advantage of overloading,
disambiguated by parameter type. The requirement to offer a choice (\lstinline|&|) in the channel session type 
corresponds to availability of \lstinline|receive|, with the subsequent session
depending on the label that is received. Sending (\lstinline|!|) and receiving (\lstinline|?|) data in the channel session type
correspond straightforwardly to \lstinline|send| and
\lstinline|receive| with appropriate signatures.

Figure~\ref{fig:remotefileserver1b} defines the class
\lstinline|RemoteFile|, which acts as a local proxy for a remote file
server. Its interface is similar to that of the class \lstinline|File|
from Section~\ref{sec:seq-example}; the only difference is that
\lstinline|RemoteFile| has an additional method \lstinline|connect|,
which must be called in order to establish a connection to the file
server. The types \lstinline|RemoteFile.Init| and
\lstinline|File.Init| are equivalent
(Definition~\ref{def:typeequivalence}): each is a subtype of the
other, and they can be used interchangeably.

The methods of \lstinline|RemoteFile| are implemented by communicating
over a channel to a file server. The \lstinline|connect| method has a
parameter of type \lstinline|<FileReadCh>|. A value of this type
represents an access point, analogous to a URL, on which a connection
can be requested by calling the \lstinline|request| method (line 11);
the resulting channel endpoint has type \lstinline|FileRead_cl|.

The remaining methods communicate on the channel, and thus advance the
type of the field \lstinline|channel|. The similarity of structure
between the channel session type \lstinline|FileReadCh| and the class
session type \lstinline|Init| is reflected in the simple definitions
of the methods, which just copy information between their parameters
and results and the channel. There is one point of interest in
relation to the \lstinline|close| method. It occurs three times in the
class session type, and according to our type system, its body is type
checked once for each occurrence. Each time, the initial type
environment in which the body is checked has a different type for the
\lstinline|channel| field: \lstinline|Open_cl|,
\lstinline|MustClose_cl| or \lstinline|CanRead_cl|. Type checking is
successful because all of these types allow \lstinline|send({CLOSE})|.

Figure~\ref{fig:remotefileserver1e} defines the class
\lstinline|FileServer|, which accesses a local file system and uses
the server endpoint of a channel of type \lstinline|FileReadCh|. The
session type of this class contains the single method
\lstinline|main|, with a parameter of type
\lstinline|<FileReadCh>|. We imagine this \lstinline|main| method to
be the top-level entry point of a stand-alone application, with the
parameter value (the access point or URL for the server) being
provided when the application is launched. The server uses
\lstinline|accept| to listen for connection requests, and when a
connection is made, it obtains a channel endpoint of type
\lstinline|FileRead_s|.

\begin{figure}
\begin{lstlisting}
FileChannel = &{OPEN: ?String.+{OK: CanRead, ERROR: FileChannel}, QUIT: End}
CanRead = &{READ: +{EOF: FileChannel, DATA: !String.CanRead}, CLOSE: FileChannel}
\end{lstlisting}
\caption{Remote file server version 2: channel session type (server side)}
\label{fig:file-channel-type}
\end{figure}

\begin{figure}
\begin{lstlisting}
ClientCh = {Null send({OPEN}): {Null send(String): {{OK, ERROR} receive():
                                                      <OK: CanRead_cl,
                                                       ERROR: ClientCh>}},
            Null send({QUIT}): {}}
where CanRead_cl = {Null send({READ}): {{EOF, DATA} receive():
                                          <EOF: ClientCh,
                                           DATA:{String receive(): CanRead_cl}>},
                    Null send({CLOSE}): ClientCh}
\end{lstlisting}
\begin{lstlisting}
ServerCh = {{OPEN, QUIT} receive(): <OPEN: {String receive():
                                             {Null send({OK}): CanRead_cl,
                                              Null send({ERROR}): ServerCh}},
                                     QUIT: {}>}
where CanRead_cl = {{READ, CLOSE} receive():
                      <READ: {Null send({EOF}): ServerCh,
                              Null send({DATA)}:{Null send(String): CanRead_cl}},
                       CLOSE: ServerCh>}
\end{lstlisting}
\caption{Remote file server version 2: client and server
  class session types generated from channel session
  type \lstinline|FileChannel|}
\label{fig:client-server-file-classes}
\end{figure}

\begin{figure}
\begin{lstlisting}
class RemoteFile {
  session {connect: Init} 
  where Init  = {open: <OK: Open, ERROR: Init>}
        Open  = {hasNext: <TRUE: Read, FALSE: Close>, close: Init}
        Read  = {read: Open, close: Init}
        Close = {close: Init}
 
  channel; state;
 
  Null connect(<FileChannel> c) {
    channel = c.request();
  }
  {OK,ERROR} open(String name) {
    channel.send(OPEN);
    channel.send(name);
    switch (channel.receive()) {
      OK: state = READ; OK;
      ERROR: ERROR;
    }
  }
  {TRUE,FALSE} hasNext() {
    channel.send(READ);
    switch (channel.receive()) {
      EOF: state = EOF; FALSE;
      DATA: state = DATA; TRUE;
    }
  }
  String read() {
    state = READ;
    channel.receive();
  }
  Null close() {
    switch (state) {
      EOF: null;
      READ: channel.send(CLOSE);
      DATA: channel.receive(); channel.send(CLOSE);
    }
  }
}
\end{lstlisting}
\caption{Remote file server version 2: client side stub}
\label{fig:remote-file}
\end{figure}

\begin{figure}
\begin{lstlisting}
class FileServer {
  session {main: {}}

  channel; file;

  Null main(<FileChannel> c) {
    file = new File();
    channel = c.accept();
    serverCh();
  }
  req ServerCh channel, Init file
  ens {} channel, Init file
  Null serverCh() {
    switch (channel.receive()) {
      OPEN:
        switch (file.open(channel.receive())) {
          OK: canRead();
          ERROR: serverCh();
       }
      QUIT: null;
    }
  }
  req CanRead_s channel, Open file
  ens {} channel, Init file
  Null canRead() {
    switch (channel.receive()) {
      READ:
        switch (file.hasNext()) {
          TRUE: channel.send(DATA); channel.send(file.read());
                canRead();
          FALSE: channel.send(EOF); file.close(); serverCh();
       }
      CLOSE: file.close(); serverCh();
  }
}
\end{lstlisting}
\caption{Remote file server version 2: server code}
\label{fig:file-server}
\end{figure}


The remaining methods of \lstinline|FileServer| are mutually recursive
in a pattern that matches the structure of \lstinline|FileRead_s|. The
methods are self-called, and do not appear in the class session type;
instead, they are annotated with pre- and post-conditions on the types
of the fields \lstinline|channel| and \lstinline|file|. The direct
correspondence between the structure of the channel session type and
the class session type of \lstinline|File| is again reflected in the
code, for example on lines 29 and 30 where the result of calling
\lstinline|hasNext| on \lstinline|file| directly answers the
\lstinline|HASNEXT| query on \lstinline|channel|.

Most systems of session types support delegation, which is the ability
to send a channel as a message on another channel. It is indicated by
the occurrence of a session type as the type of the message in a send
(\lstinline|!|) or receive (\lstinline|?|) constructor. In our
language, delegation is realised by sending an object representing a
channel endpoint; it corresponds to a \lstinline|send| method with a
parameter of type, for example,
\lstinline|Chan[FileRead_cl]|. Transfer of channel endpoints from one
process to another is supported by the operational semantics in
Section~\ref{sec:distributed}.

\subsection{Distributed Example Version 2} 

This version has a different channel session type,
\lstinline|FileChannel|, defined in
Figure~\ref{fig:file-channel-type}, which does not match the class
session type \lstinline|FileRead|. The difference is that there is no
\lstinline|HASNEXT| option; instead, the \lstinline|READ| option is
always available. If there is no more data then \lstinline|EOF| is
returned in response to \lstinline|READ|; alternatively,
\lstinline|DATA| is returned, followed by the desired data. The
corresponding class session types for the client and server endpoints
are defined in Figure~\ref{fig:client-server-file-classes}.

The implementation of \lstinline|RemoteFile| must now mediate between
the different structures of the class session type
\lstinline|FileRead| and the channel session type
\lstinline|FileChannel|. The new definition is in
Figure~\ref{fig:remote-file}. The main point is that the
definition of the \lstinline|close| method must depend on the state of
the channel. For example, if \lstinline|close| is called immediately
after a call of \lstinline|hasNext| that returns \lstinline|TRUE|,
then the channel session type requires data to be read before
\lstinline|CLOSE| can be sent. We therefore introduce the field
\lstinline|state|, which stores a value of the enumerated type
\lstinline|{EOF, READ, DATA}|. This field represents the state of the
channel (equivalently, the session type of the \lstinline|channel|
field): \lstinline|EOF| corresponds to \lstinline|ClientCh|,
\lstinline|READ| corresponds to \lstinline|CanRead|, and
\lstinline|DATA| corresponds to the point after the \lstinline|DATA|
label in \lstinline|CanRead|. The definition of \lstinline|close|
contains a \lstinline|switch| on \lstinline|state|, with appropriate
behaviour for each possible value. It is also possible for
\lstinline|state| to be \lstinline|null|, but this only occurs before
\lstinline|open| has been called, and at this point \lstinline|close|
is not available.

In order to type check this example we take advantage of the fact that
the body of the close method is repeatedly checked, according to its
occurrence in the class session type. The value of \lstinline|state|
always corresponds to the state of \lstinline|channel|. This
correspondence is not represented in the type system --- that would
require some form of dependent type --- but whenever the body of
\lstinline|close| is type checked, the \emph{type} of
\lstinline|channel| is compatible with the \emph{value} of
\lstinline|state|, and so typechecking succeeds. More precisely, each
possible value of \lstinline|state| corresponds to a different
singleton type for \lstinline|state| (typing rule \Tlabel), and rule
\Tswitch\ only checks the branches that correspond to possible values
in the enumerated type of the condition. So each time the body of
\lstinline|close| is type checked, only \emph{one} branch (because the
type of \lstinline|state| is a singleton) of the
\lstinline|switch| is checked, corresponding to the value of
\lstinline|state| for that occurrence of \lstinline|close|. 
 



\section{A Core Distributed Language}\label{sec:distributed}

\comment{SG: updated 29.3.2011}
\comment{NG: updated 8.4.2011}

We now define the core of the distributed language illustrated in
Section~\ref{sec:distributed-example}. For simplicity, communication
is synchronous. Formalising asynchronous communication is
well-understood (for example, Gay and Vasconcelos \cite{GaySJ:lintta}
define a functional language with similar communication primitives but
adds the complication of message buffers to the operational semantics). 

Our integration of channel session types and the typestate system of
this paper is based on binary session types \cite{HondaK:intblt}
(actually, we adopt the now standard constructs of session types).
It should be straightforward to adapt the technique to multi-party
session types \cite{HondaK:mulast}, because that system also depends
on specifying the sequence of send and receive operations on each
channel endpoint. 

The only additions to the top-level language are access points and
their types $\access{T}$, channel session types and their translation
to class session types, and the $\spawnterm$ primitive. However, there
are significant changes to the internal language, in order to
introduce a layer of concurrently executing components that
communicate on channels.

\subsection{Syntax}

\begin{figure}
\begin{align*}
  \textrm{Declarations}&& D \bnf\ & \ldots \alt \mkterm{access}~\langle\Sigma\rangle~n\\
  \textrm{Values}&& v \bnf\ & \ldots \alt c^+ \alt c^- \alt n\\
  \textrm{Expressions}&& e \bnf\ & \ldots \alt \spawn{C.m(e)}\\
  \textrm{Contexts}&& \mathcal{E} \bnf\ & \ldots \alt \spawn{C.m(\mathcal{E})}\\
  \textrm{Channel session types}&& \Sigma \bnf\ & \End \alt
    X \alt \mu X.\Sigma \alt \rcv{T}\chanseq\Sigma \alt \send{T}\chanseq\Sigma\\
     && \alt &\chanbranch{l : \Sigma_l}_{l\in E}
    \alt \chanchoice{l : \Sigma_l}_{l\in E}\\
  \textrm{States}&& s \bnf\ &
    \ldots \alt s\parcomp s \alt (\nu c)\,s
\end{align*}
\caption{Additional syntax for channels and states}
\label{fig:concsyntax}
\end{figure}


Figure~\ref{fig:concsyntax} defines the new syntax. The types of
access points are top-level declarations. Of the new values, access
points $n$ can appear in top-level programs, but channel endpoints,
$c^+$ and $c^-$, are part of the internal language. If an access point $n$ is declared with $\mkterm{access}~\langle\Sigma\rangle~n$ then we define $n.\mkterm{protocol}$ to mean $\Sigma$.
The $\spawnterm$ primitive was not used in the example in
Section~\ref{sec:distributed-example}, but its behaviour is to start a
new thread executing the specified method on a new instance of the
specified class (just like it happens in Java; other works,
e.g.~\cite{Dezani-CiancagliniM:bousto}, use similar
approaches). Although a parameter is required as in any method
call, for simplicity the type system restricts the parameter's type to
be $\nulltype$ in this case, so that there is only one form of
inter-thread communication.  The syntax of channel session types
$\Sigma$ is included so that the types of access points can be
declared.
Channels are created by the interaction of methods $\requestterm$ and
$\acceptterm$ in different threads, one thread keeps the $c^-$
endpoint whereas the other keeps $c^+$. The two threads then
communicate on channel~$c$ by reading and writing on their channel
ends.
The syntax of states is extended to include parallel composition and a
channel binder $\nu c$, which binds both endpoints $c^+$ and $c^-$ in
the style of Gay and Vasconcelos \citeN{GaySJ:lintta}. In a parallel
composition, the states are exactly states from the semantics of the
sequential language; in particular, each one has its own heap. This
means that $\spawnterm$ generates a new heap as well as a new
executing method body. Communication between parallel expressions is
only via channels.

The syntax extensions do not include $\requestterm$, $\acceptterm$,
$\sendterm$ and $\rcvterm$, as they are treated as method names.

\subsection{Semantics}
\begin{figure}
Structural congruence: 
\begin{gather*}
    \tag{\Ecomm,\Eassoc}
    s_1 \parallel s_2\ \equiv\  s_2 \parallel s_1
    \qquad
    s_1 \parallel (s_2 \parallel s_3)\ \equiv\  (s_1 \parallel s_2)
    \parallel s_3
    \\[\typingRuleSkip]
    \tag\Escope
    s_1 \parallel (\nu c)s_2\ \equiv\ (\nu c)(s_1\parallel s_2) 
    \;\;\text{if $c^+,c^-$ not free in $s_1$}
\end{gather*}

Reduction rules:
\begin{center}
\setlength\lineskip{\typingRuleSkip}
\AxiomC{$h(r).f = n$}
\AxiomC{$h'(r').f' = n$}
\AxiomC{$c$ fresh}
\rulename{R-Init}
\TrinaryInfC{$ \state{h}{r}{\mathcal
    E[\accept{f}]}\parcomp\state{h'}{r'}{\mathcal E'[\request {f'}]}
  \vreduces (\nu c) \left(\state{h}{r}{\mathcal
      E[c^+]}\parcomp\state{h'}{r'}{\mathcal E'[c^-]}\right)$}
\DisplayProof\hfil
\AxiomC{$h(r).f = c^p$}
\AxiomC{$h'(r').f' = c^{\overline p}$}
\AxiomC{$v\not\in\mathcal{O}$}
\rulename{R-ComBase}
\TrinaryInfC{$\state{h}{r}{\mathcal E[\methcal{f}{\sendterm}{v}]}
      \parcomp\state{h'}{r'}{\mathcal E'[\methcal{f'}{\rcvterm}{}]}
      \vreduces
      \state{h}{r}{\mathcal E[\nullterm]}\parcomp\state{h'}{r'}{\mathcal E'[v]}$}
\DisplayProof\hfil
\AxiomC{$h(r).f = c^p$}
\AxiomC{$h'(r').f' = c^{\overline p}$}
\AxiomC{$\varphi\in\mkterm{Inj}(\dom(h\downarrow o),\mathcal{O}\setminus\dom(h'))$}
\rulename{R-ComObj}
\TrinaryInfC{$\begin{array}{l}\state{h}{r}{\mathcal E[\methcal{f}{\sendterm}{o}]}
      \parcomp\state{h'}{r'}{\mathcal E'[\methcal{f'}{\rcvterm}{}]}
      \vreduces\\
      \hspace{10mm}\state{h\uparrow o}{r}{\mathcal E[\nullterm]}\parcomp\state{h'+\varphi(h\downarrow o)}{r'}{\mathcal E'[\varphi(o)]}\end{array}$}
\DisplayProof\hfil
\AxiomC{$o$ fresh}
\AxiomC{$C.\fieldsterm = \vec{f}$}
\AxiomC{$\method{}{m}{x}{e}\in C$}
\rulename{R-Spawn}
\TrinaryInfC{$\state{h}{r}{\mathcal{E}[\spawn{C.m(v)}]}\vreduces
  \state{h}{r}{\mathcal{E}[\nullterm]}\parcomp\state{\smash{o =
      \hentry{C}{\vec{f}=\vec\nullterm}}}{o}{e\subs{\nullterm}{x}}$}
\DisplayProof\hfil
%
%
%
\AxiomC{$s\vreduces s'$}
\rulename{R-Par}
\UnaryInfC{$s\parcomp s''\vreduces s'\parcomp s''$}
\DisplayProof\hfil
\AxiomC{$s\equiv s' \quad s'\vreduces s'' \quad s''\equiv s'''$}
\rulename{R-Str}
\UnaryInfC{$s\vreduces s'''$}
\DisplayProof\hfil
\AxiomC{$s\vreduces s'$}
\rulename{R-NewChan}
\UnaryInfC{$(\nu c)\,s\vreduces(\nu c)\,s'$}
\DisplayProof\hfil
\end{center}

Top-level typing rules:
\begin{center}\setlength\lineskip{\typingRuleSkip}\setlength\baselineskip{0pt}
  \AxiomC{\judgment\envref{\Gamma}{r} > e : \nulltype < \envref{\Gamma'}{r'} /}
  \AxiomC{$C.\sessterm = \{\methsign{m}{\nulltype}{\_} : \_; \ldots\}$}
  \rulename{T-Spawn}
  \BinaryInfC{\judgment\envref{\Gamma}{r} > \spawn{C.m(e)} : \nulltype < \envref{\Gamma'}{r'} /}
\DisplayProof\hfil
%
\AxiomC{$\mkterm{access}~\langle\Sigma\rangle~n$}
\rulename{T-Name}
\UnaryInfC{\judgment\envref{\Gamma}{r} > n : \sem{\access\Sigma} < \envref{\Gamma}{r} /}
\DisplayProof\hfil
%
\end{center}
\caption{Reduction and top-level typing rules for concurrency and channels} 
\label{fig:concurrency}
\end{figure}


Figure~\ref{fig:concurrency} defines the reduction rules for the
distributed language, as well as the top-level typing rules. The
reduction rules make use of a pi-calculus style structural congruence
relation, again following Gay and Vasconcelos \citeN{GaySJ:lintta}.  It is the smallest
congruence (with respect to parallel and binding) that is also closed
under the given rules.

Rule \Rinit\ defines interaction between $\acceptterm$ and
$\requestterm$, which creates a fresh channel $c$ and substitutes one
endpoint into each expression.

There are two rules for communication, involving interaction between
$\sendterm$ and $\rcvterm$. Rule \RcomBase\ is for communication of
non-objects and rule \RcomObj\ is for communication of objects. Let
$\mathcal{O}$ be the set of all object identifiers. \RcomBase\
expresses a straightforward transfer of a value, while \RcomObj\ also
transfers part of the heap corresponding to the contents of a
transferred object. In \RcomObj, $\varphi$ is an arbitrary renaming
function which associates to every identifier in $\dom(h)$ an
identifier not in $\dom(h')$.  This rule can easily be made
deterministic in practice by using a total ordering on identifiers and
a mechanism to generate fresh ones.

\Rspawn\ creates a new parallel state whose heap contains a single
instance of the specified class. As discussed above, communication
between threads is only through channels in order to keep the formal
system a reasonable size; therefore, no data is transmitted to the new
thread and the body of the method being spawned always has its
parameter replaced by the literal $\nullterm$. The type system will
ensure that $v=\nullterm$, so that this semantics makes sense.
The remaining rules are standard. 

Returning to \RcomObj, there is some additional notation associated
with identifying the part of the heap that must be transferred; we now
define it.

\begin{definition}
  Let $h$ be a heap.  For any entry $o = \hentry{C}{\iset{f_i =
      v_i}{i\in I}}$ in $h$, we define the \emph{children} of $o$ in
  $h$ to be the set of all $v_i$ which are object identifiers:
  $\children{h}(o) = \{v_i\mid i\in I\}\cap\mathcal{O}$.

  We say that an object identifier $o$ in $\dom(h)$ is a \emph{root}
  in $h$ if there is no $o'$ in $\dom(h)$ such that
  $o\in\children{h}(o')$. We note $\roots(h)$ the set of roots in $h$.

  We say that $h$ is \emph{complete} if for any $o$ in $\dom(h)$ we have
  $\children{h}(o) \subset \dom(h)$.

  If $h$ is complete, we define the descendants of $o$ in $h$ to be
  the smallest set containing $o$ and the children of any object it
  contains.  Formally, let $\children{h}^0(o) = \{o\}$ and for
  $i\geqslant 1$, $\children{h}^i(o) =\!\!\!\!\!\!\!\!\!\!\!\!\!\!\!\!
  \bigcup\limits_{\null\qquad\omega\in\children{h}^{i-1}(o)}
  \!\!\!\!\!\!\!\!\!\!\!\!\!\!\!\!\children{h}(\omega)$. Then
  $\desc{h}(o) = \bigcup\limits_{i\in\mathbb{N}}\children{h}^i(o)$.
\end{definition}

\begin{definition}[Heap separation]
  Let $h$ be a complete heap and $o$ a root of $h$. We define
  $h\downarrow o$ to be the sub-heap obtained by restricting $h$ to
  the descendants of $o$, and $h\uparrow o$ to be the sub-heap obtained by
  removing from $h$ the descendants of $o$. Note that $h\downarrow o$
  is complete and has the property that $o$ is its only root.
\end{definition}

\begin{definition}[Additional notation]\
\begin{itemize}
\item Let $h$ be a heap and let $\varphi$ be a function from
  $\mathcal{O}$ to $\mathcal{O}$. We denote by $\varphi(h)$ the result
  of applying $\varphi$ to all object identifiers in $h$, including
  inside object records.
\item We denote by $\mkterm{Inj}(A,B)$ the set of injective functions
  from $A$ to $B$.
\item  We denote by $+$ the disjoint union of heaps or environments, \ie the
  operation $h+h'$ is defined by merging $h$ and $h'$ if their domains
  are disjoint and undefined otherwise.
\end{itemize}
\end{definition}

\subsection{Type System}
\begin{figure}
Given a channel session type $\Sigma$, define a class session type
$\sem\Sigma$ as follows.
\begin{align*}
\sem{\End} & = \{\}\\
\sem{X} & = X \\
\sem{\mu X.\Sigma} & = \mu X.\sem\Sigma \\
\sem{\rcv{T}\chanseq\Sigma} & = \branch{\methsign{\rcvterm}{\nulltype}{T}}{\sem\Sigma}{}\\
\sem{\send{T}\chanseq\Sigma} & = \branch{\methsign{\sendterm}{T}{\nulltype}}{\sem\Sigma}{}\\
\sem{\chanbranch{l : \Sigma_l}_{l\in E}} & = \branch{\methsign{\rcvterm}{\nulltype}{\linkthis}}{\choice{l}{\sem{\Sigma_l}}{l\in E}}{}\\
\sem{\chanchoice{l :\Sigma_l}_{l\in E}} & = \branch{\methsign{\sendterm}{\{l\}}{\nulltype}}{\sem{\Sigma_l}}{l\in E}
\end{align*}


In the type system, a channel endpoint with session type $\Sigma$ is treated as an object with type $\sem\Sigma$. Calls of $\mkterm{send}$ and $\mkterm{receive}$ are typed as standard method calls.
\\~\\
Given an access type $\access\Sigma$, define a class session type
$\sem{\access\Sigma}$ by
\begin{align*}
\sem{\access\Sigma} & = \mu X.\{\methsign{\requestterm}{\nulltype}{\sem{\overline\Sigma}}:X,\methsign{\acceptterm}{\nulltype}{\sem\Sigma}:X\}.
\end{align*}
In the type system, an access point with type $\access\Sigma$ is
treated as an object with type $\sem{\access\Sigma}$. Calls of
$\mkterm{request}$ and $\mkterm{accept}$ are typed as standard method
calls. 

\caption{Object types for channels and access points}
\label{fig:channelobjects}
\end{figure}

\begin{figure}
These rules add to or replace the rules in
Figure~\ref{fig:typingstates}.
~\\

\centering\setlength\lineskip{\typingRuleSkip}
%
\axiomname{T-Chan}
\judgment\envref{\Gamma, c^p : T}{r} > c^p : T < \envref{\Gamma}{r} /
\hfil
%
%
%
\axiomname{T-Hempty}
\typedheapchan{\Theta}{\sem\Theta}{\varepsilon}
\hfil
%
%
%
\AxiomC{\typedheapchan{\Theta}{\Gamma}{h}}
\AxiomC{\judgment \envref{\Gamma, o : \objecttype{C}{\iset{\nulltype\,f_i}{1\leqslant i\leqslant n}}}{o} >
  \seq{\seq{\swap{f_1}{v_1}}{\ldots}}{\swap{f_n}{v_n}} : \nulltype
  < \envref{\Gamma'}{o} /}
\rulename{T-Hadd}
\BinaryInfC{\typedheapchan{\Theta}{\Gamma'}
  {\hadd{h}{o = \hentry{C}{\iset{f_i = v_i}{1\leqslant i\leqslant n}}}}
}
\DisplayProof\hfil
\AxiomC{\typedheapchan{\Theta}{\Gamma, o : \objecttype C F}{h}}
\AxiomC{\typedsess{F}{C}{S}}
\rulename{T-Hide}
\BinaryInfC{\typedheapchan{\Theta}{\Gamma, o : \objecttype C S}{h}}
\DisplayProof\hfil
\AxiomC{\typedheapchan{\Theta}{\Gamma}{h}}
\AxiomC{\judgment\envref{\Gamma}{r} > e : T < \envref{\Gamma'}{r'} /}
\rulename{T-State}
\BinaryInfC{\judgment\Theta;\Gamma > \state{h}{r}{e} : T < \envref{\Gamma'}{r'} /}
\DisplayProof
\AxiomC{\judgment \Theta;\Gamma > \state{h}{r}{e} : T < \envref{\Gamma'}{r'} /}
\rulename{T-Thread}
\UnaryInfC{\typedconf{\Theta}{\state{h}{r}{e}}}
\DisplayProof\hfil
\AxiomC{\typedconf\Theta s}
\AxiomC{\typedconf{\Theta'}{s'}}
\rulename{T-Par}
\BinaryInfC{\typedconf{\Theta+\Theta'}{s\parcomp s'}}
\DisplayProof\hfil
\AxiomC{\typedconf{\Theta,
    c^+:\Sigma,
    c^-:\overline\Sigma}{s}}
\rulename{T-NewChan}
\UnaryInfC{\typedconf{\Theta}{(\nu c)\,s}}
\DisplayProof
\caption{Internal typing rules for the distributed language}
\label{fig:typingproofs}
\end{figure}


The type system treats $\sendterm$, $\rcvterm$, $\requestterm$ and
$\acceptterm$ as method calls on objects whose session types are
defined by the translations in Figure~\ref{fig:channelobjects}. A
channel endpoint with (channel) session type $\Sigma$ is treated as an
object with (class) session type $\sem\Sigma$. The type constructor
$\&$ (offer) is translated into a $\rcvterm$ method with return
type $\linkthis$ in order to capture the relationship between the
received label and the subsequent type. The type constructor $\oplus$
(select) is translated into a collection of $\sendterm$ methods with
different parameter types, each being a singleton type for the
corresponding label.

In a similar but much simpler way, an access type $\access\Sigma$ is
translated into a (class) session type that allows both $\requestterm$
and $\acceptterm$ to be called repeatedly and at any time. These two
methods need to return dual channel endpoints, which requires the
following definition.

\begin{definition}[Dual channel type]
  The dual type $\overline\Sigma$ of a channel session type $\Sigma$
  is defined by
\[
\overline{\Sigma} = \dl(\Sigma,\iota)
\]
where $\iota$ is the identity substitution, $\sigma\subs{\Sigma'}{X}$
denotes the extension of substitution $\sigma$ by the mapping
$X \mapsto\Sigma'$, and $\Sigma\sigma$ denotes the application of the
substitution $\sigma$ to the session type $\Sigma$; the auxiliary
function $\dl(\Sigma,\sigma)$ is defined on session types $\Sigma$ and
substitutions $\sigma$ by:

  \begin{align*}
    \dl(\End,\sigma) &= \End\\
    \dl(X,\sigma) &= X\\
    \dl(\mu X.\Sigma,\sigma) &= \mu X.\dl(\Sigma,\sigma\subs{\Sigma\sigma}{X})\\
    \dl(\send T\chanseq\Sigma,\sigma) &= \rcv{T\sigma}\chanseq\dl(\Sigma,\sigma)\\
    \dl(\rcv T\chanseq\Sigma,\sigma) &= \send{T\sigma}\chanseq\dl(\Sigma,\sigma)\\
    \dl(\chanchoice{l : \Sigma_l}_{l\in E},\sigma) &= \chanbranch{l :
      \dl(\Sigma_l,\sigma)}_{l\in E}\\
    \dl(\chanbranch{l : \Sigma_l}_{l\in E},\sigma) &= \chanchoice{l :
      \dl(\Sigma_l,\sigma)}_{l\in E}
  \end{align*}
With this definition, duality commutes with unfolding \cite{BernardiG:private}; this property is essential in order to use the equi-recursive convention (Definition~\ref{def:typeequivalence}).
\end{definition}
By convention, $\requestterm$ returns a channel endpoint of type
$\sem{\overline\Sigma}$ and $\acceptterm$ returns an endpoint of type
$\sem\Sigma$.

Because access points $n$ are global constants, they can be used
repeatedly even though their session types are linear; there is no
restriction to a single occurrence of a given name.  

The only new typing rules for the top-level language are in
Figure~\ref{fig:concurrency}. \Tspawn\ allows a method to be used in a
$\spawnterm$ expression if it is available in the initial session type
of the specified class. \Tname\ obtains the type of an access point
from its declaration, and assigns an object type according to the
translation described above.

Figure~\ref{fig:typingproofs} contains typing rules for the internal
language with the concurrency extensions. Rules with the same names as
rules in Figure~\ref{fig:typingstates} are replacements.

Rule \Tchan\ takes the type of a channel endpoint from the typing
environment. The remaining rules involve a new typing environment
$\Theta$, which maps channel endpoints to channel session types
$\Sigma$; these are indeed \emph{channel} session types, not their
translations into class session types. \THadd, \THide\ and \Tstate\
are just the corresponding rules from Figure~\ref{fig:typingstates}
with $\Theta$ added. In \THempty\ the notation $\sem\Theta$ means that
the translation from channel session types to class session
types is applied to the type of each channel endpoint. In combination
with \Tstate, this means that the typing of expressions uses class
session types for channel endpoints; the $T$ in \Tchan\ is a class
session type.

\Tthread\ lifts a typed state to a typed concurrent component,
preserving only the channel typing $\Theta$, which is used in \Tpar\
and \TnewChan. In \Tpar, $\Theta+\Theta'$ means union, with
the assumption that $\Theta$ and $\Theta'$ have disjoint
domains. \TnewChan\ requires the complementary endpoints of each
channel to have dual session types.


\subsection{Subtyping}

We have two subtyping relations between channel session types:
$\Sigma\subt\Sigma'$ as defined by Gay and Hole \citeN{GaySJ:substp}, and
$\sem{\Sigma}\subt\sem{\Sigma'}$ as defined in this paper. To avoid a
detour into the definition of $\Sigma\subt\Sigma'$, we state the
following result without proof.

\begin{proposition}
$\Sigma\subt\Sigma' \Rightarrow \sem{\Sigma}\subt\sem{\Sigma'}$.
\end{proposition}

Interestingly, the converse is not true, as subtyping between
translations of channel session types is a larger relation. For
example:
\begin{itemize}
\item for all $\Sigma$, $\sem\Sigma \subt \sem\End$
\item if $E$ is an enumeration then $\sem{\rcv{E}\chanseq\Sigma} \subt
  \sem{\chanbranch{l : \Sigma}_{l\in E}}$
\item $\sem{\chanchoice{l : \Sigma}} =
  \sem{\send{\{l\}}\chanseq\Sigma}$ and therefore by transitivity any
  translated $\oplus$ type is a subtype of all the corresponding
  individual translated send
  types.
\end{itemize}

\noindent This suggests the possibility, in the context of the work of Gay and Hole \citeN{GaySJ:substp},
of generalizing the subtyping relation between channel session types
by considering branch/select labels as values in an enumerated
type. We do not explore this idea further in the present paper.


\section{results}\label{sec:results}

\comment{SG: updated text 31.3.2011.} 
\comment{NG: updated proofs 7.4.2011, everything is there.
  Moved those of Section 7.1 to an
  appendix. techlemmas.tex contains section 7.1 with the proofs and
  explanatory text, proofs are between \textbackslash ifproofs and
  \textbackslash fi and explanatory text (of which there is currently
  little) should be put between \textbackslash ifexplanatorytext and
  \textbackslash fi, so that the file can be used both for typesetting
  7.1 and the appendix.}

%
The key results concerning the distributed language supporting
self-calls are, again, Subject-Reduction, Type Safety, and
Conformance. Notice that we can no longer guarantee the
absence of stuck states for all well-typed programs, as one endpoint of a
channel may try to send when the other endpoint is not available to receive.

\subsection{Properties of typing derivations}\label{sec:techlemmas}
This subsection is mostly a collection of lemmas which will be used to
prove the main theorems in the following subsections. They draw various
useful consequences from the fact that a program state is well-typed.
Their proofs can be found in Appendix \ref{app:proofs}.

We define $\chans(h)$ as the set of channel endpoints appearing in
object records in $h$. We define $\chans(\Gamma)$ and $\objs(\Gamma)$
as the sets of, respectively, channel endpoints and object identifiers
in $\dom(\Gamma)$. We have $\dom(\Gamma) = \chans(\Gamma)\cup
\objs(\Gamma)$.

\newif\ifinlineproofs
\newif\ifexplanatorytext
\inlineproofsfalse
\explanatorytexttrue

\begin{lemma}\label{lem:heapproperties}
  Suppose $\typedheapchan\Theta\Gamma h$. Then \emph{(a)} $h$ is
  complete,
  \emph{(b)} $\chans(\Gamma)\subset\dom(\Theta)\setminus \chans(h)$ and
  \emph{(c)} $\objs(\Gamma)\subset\roots(h)$.
\end{lemma}
\ifinlineproofs\begin{proof}
  By induction on the derivation of $\typedheapchan\Theta\Gamma h$.
  The only axiom is \textsc{T-Hempty} for which the properties are true.
  Then \textsc{T-Hide} does not change either $h$ or
  $\dom(\Gamma)$ so it preserves all three properties. The other case is
  \textsc{T-Hadd}. Let $h'$ be the heap in the conclusion. Then
  $\children{h'}(o)$ is the set of $v_i$ which are object
  identifiers. Let $K$ be the set of $v_i$ which are channel endpoints.
  The typing derivation for the sequence
  of swaps in the right premise must include an occurrence of
  \textsc{T-Ref} for each object identifier, and of \textsc{T-Chan}
  for each channel endpoint, 
  each followed by \textsc{T-Swap} and a number of occurrences of
  \textsc{T-Seq}. Looking at these rules, we can see that this
  implies:\begin{enumerate}
  \item $\children{h'}(o)\cup K\subset\dom(\Gamma)$ and
  \item $\dom(\Gamma')\subset(\dom(\Gamma)\setminus
  (\children{h'}(o)\cup K))\cup\{o\}$. (Note that $o$ cannot be one of the
  $v_i$ because it is the current object in the judgement: the premise
  of \textsc{T-Ref} forbids it.)
  \end{enumerate}

  From (1) and induction hypothesis (c) we get
  $\children{h'}(o)\subset\roots(h)$. We have
  $\roots(h)\subset\dom(h)\subset\dom(h')$ and $o$ is the only new
  object in $h'$, so $h'$ is complete.

  If we project (2) onto just channel endpoints, we get
  $\chans(\Gamma')\subset \chans(\Gamma)\setminus K$. From the
  definition of $h'$, $\chans(h')$ is
  equal to $\chans(h)\cup K$. Hence induction hypothesis (b) yields (b)
  again for $h'$.

  If we project (2)
  onto just object identifiers, we get
  $\objs(\Gamma')\subset(\objs(\Gamma)\setminus\children{h'}(o))\cup\{o\}$.
  From induction hypothesis (a) and the fact that $o\not\in\dom(h)$ we
  get that $o$ is a root in $h'$. Furthermore, all roots of $h$ which
  are not children of $o$ are also roots of $h'$.
  Thus induction
  hypothesis (c) allows us to conclude $\objs(\Gamma')\subset\roots(h')$.
\end{proof}\fi

\begin{lemma}[Rearrangement of typing derivations for expressions]\hfil
  \label{lem:rearrange-expr}
  Suppose we have \judgment \envref\Gamma r > e : T <
  \envref{\Gamma'}{r'} /. Then there exists a typing derivation for
  this judgement in which:
\begin{enumerate}
  \item \textsc{T-Sub} only occurs at the
    very end, just before \textsc{T-Switch} or \textsc{T-SwitchLink}
    as the last rule in the derivation for each of the branches, or 
    just before \textsc{T-Call} as the last rule in the
    derivation for the parameter; 
  \item \textsc{T-SubEnv} only occurs immediately before
    \textsc{T-Sub} in the first three cases and does not occur at all
    in the fourth, \ie \textsc{T-Call}.
\end{enumerate}
\end{lemma}
\ifinlineproofs\begin{proof}
  First note that \textsc{T-Sub} and \textsc{T-SubEnv} commute and
  that any consecutive sequence of occurrences of one of these rules
  can collapse into a single occurrence using transitivity.
  What
  remains to be shown is that these rules can be pushed down in all
  cases but those mentioned in the statement. We enumerate the cases below.
  \begin{itemize}
  \item \textsc{T-Swap}. \textsc{T-Sub} before the premise can be replaced
    with \textsc{T-SubEnv} after the conclusion as $T$ has been
    transferred to the environment. If \textsc{T-SubEnv} was used before the premise, it means the initial derivation looks like:

\smallskip
\AxiomC{$\judgment\envref\Gamma r > e : T < \envref{\Gamma'}{r'}/$}
\rrulename{T-SubEnv}
\UnaryInfC{$\judgment\envref\Gamma r > e : T < \envref{\Gamma''}{r'}/$}
\AxiomC{$\Gamma''(r'.f) = T'$}
\AxiomC{$\ldots$}
\rrulename{T-Swap}
\TrinaryInfC{$\judgment\envref\Gamma r > \swap f e : T' < \envref{\changetype{\Gamma''}{r'.f}{T}}{r'}/$}
\DisplayProof

\smallskip
with $\Gamma' \subt \Gamma''$. Let $T_0 = \Gamma'(r'.f)$; we have $T_0\subt T'$ since $T'=\Gamma''(r'.f)$. Because the type of $r'.f$ moves from the environment to the expression, when we push the subsumption step down, we have to use both \textsc{T-SubEnv} and \textsc{T-Sub}; we can transform the derivation into:

\smallskip
\AxiomC{$\judgment\envref\Gamma r > e : T < \envref{\Gamma'}{r'}/$}
\AxiomC{$\Gamma'(r'.f) = T_0$}
\AxiomC{$\ldots$}
\rrulename{T-Swap}
\TrinaryInfC{$\judgment\envref\Gamma r > \swap f e : T_0 < \envref{\changetype{\Gamma'}{r'.f}{T}}{r'}/$}
\rrulename{T-SubEnv}
\UnaryInfC{$\judgment\envref\Gamma r > \swap f e : T_0 < \envref{\changetype{\Gamma''}{r'.f}{T}}{r'}/$}
\rrulename{T-Sub}
\UnaryInfC{$\judgment\envref\Gamma r > \swap f e : T' < \envref{\changetype{\Gamma''}{r'.f}{T}}{r'}/$}
\DisplayProof

\smallskip
  \item \textsc{T-Call}. If \textsc{T-SubEnv} is used on the premise
    to increase the type of something else than $r'.f$ it can be moved to
    the conclusion. If the type of $r'.f$ is changed, first note that
    the only relevant part is the signature of $m_j$. Suppose the
    subsumption step
    changes it from $\methsign{m_j}{U'_j}{U_j} : S'_j$ to
    $\methsign{m_j}{T'_j}{T_j} : S_j$. For the parameter type we have
    $T'_j\subt U'_j$ so we can
    use \textsc{T-Sub} on the premise to increase the type of $e$ from
    $T'_j$ to $U'_j$ instead. For the session and result types, we
    have two cases:
    \begin{itemize}
    \item if $U_j\subt T_j$ and $S'_j\subt S_j$ it can just be moved
      to a \textsc{T-SubEnv} step on the conclusion.
    \item if $U_j = E$, $T_j = \linkthis$ and $\choice{l}{S'_j}{l\in
        E}\subt S_j$, then the original conclusion of the rule (with
      \textsc{T-SubEnv} on the premise) was:
$$\judgement \envref\Gamma r
      > \methcal{f}{m_j}{e} : \linktype{}f <
      \envref{\changetype{\Gamma'}{r'.f}{S_j}}{r'} /$$
      and the new one with the subsumption step removed is:
$$\judgement \envref\Gamma r
      > \methcal{f}{m_j}{e} : E <
      \envref{\changetype{\Gamma'}{r'.f}{S'_j}}{r'} /.$$
      So in that case the original judgement can be obtained back from
      this new conclusion using \textsc{T-VarS} followed by
      \textsc{T-SubEnv}.
    \end{itemize}
    \item \textsc{T-Seq}. \textsc{T-Sub} on the first premise is
      irrelevant and \textsc{T-SubEnv} on the same premise can be
      removed using Lemma \ref{lem:moreweak}. Subsumption on the
      second premise straightforwardly commutes to the conclusion.
    \item \textsc{T-Switch}. Lemma
      \ref{lem:moreweak} allows us to remove \textsc{T-SubEnv} on the
      first premise. Straightforwardly \textsc{T-Sub} can be removed
      as well as it just makes $E'$ smaller.
    \item \textsc{T-SwitchLink}. \textsc{T-Sub} is irrelevant;
      removing \textsc{T-SubEnv} can only make $E'$ and the initial
      typing environments for the branches smaller and we can use
      Lemma \ref{lem:moreweak}.
    \item \textsc{T-VarF} and \textsc{T-VarS}.
      \textsc{T-Sub} can increase $E$ which
      becomes the indexing set of the variant in the conclusion. By
      definition of subtyping for variants it is possible to increase
      it afterwards using \textsc{T-SubEnv}. \textsc{T-SubEnv}
      straightforwardly commutes.
    \item \textsc{T-Return}. \textsc{T-Sub} straightforwardly
      commutes, as well as the part of \textsc{T-SubEnv} not
      concerning $r'.f$. Subsumption on $\Gamma'(r'.f)$ can be removed
      using Proposition \ref{prop:subfield}.
  \end{itemize}
\end{proof}\fi

\begin{lemma}[Rearrangement of typing derivations for heaps]\hfil
\label{lem:rearrange-heap}
  Suppose $\typedheapchan\Theta\Gamma h$ holds. Let $o$ be an
  arbitrary root of $h$. Then there exists a
  typing derivation for it such that:
  \begin{enumerate}
  \item \textsc{T-Sub} is never used;
  \item \textsc{T-SubEnv} is used at most once, as the last rule leading
    to the right premise of the last occurrence of \textsc{T-Hadd};
  \item every occurrence of \textsc{T-Hide} follows immediately the
    occurrence of \textsc{T-Hadd} concerning the same object
    identifier;
  \item the occurrence of \textsc{T-Hadd} concerning an identifier $o'$
    is always immediately preceded (on the left premise)
    by the occurrences of \textsc{T-Hadd}/\textsc{T-Hide} concerning
    the descendants of $o'$;
  \item the first root added is $o$.
  \end{enumerate}  
\end{lemma}
\ifinlineproofs\begin{proof}
  The first two points are a consequence of Lemma
  \ref{lem:rearrange-expr}: the only expressions which appear in the
  typing derivation are sequences of swaps, not containing any switch
  or method call; furthermore their type is always $\nulltype$, making
  \textsc{T-Sub} at the end irrelevant. What remains to be checked is then
  just that \textsc{T-SubEnv} at the end of the derivation for one
  sequence of swaps can be pushed down to the next occurrence of
  \textsc{T-Hadd} whenever there is one. This is just a matter of
  using Proposition \ref{prop:subfield} in the case of \textsc{T-Hide}
  and Lemma \ref{lem:moreweak} in the case of \textsc{T-Hadd}.

  Note that these points imply in particular that in all applications of
  \textsc{T-Hadd} but the last one, any element in $\dom(\Gamma)$ which
  is not one of the $v_i$ also occurs in $\Gamma'$ with exactly the same type.
  
  For the third point, first notice that the premise of
  \textsc{T-Hide} implies $o$ is a root of $h$ because of Lemma
  \ref{lem:heapproperties}. This implies that the rule immediately above
  \textsc{T-Hide} either is a \textsc{T-Hadd} introducing $o$ or does
  not concern $o$ at all (in particular, $o$ cannot be a $v_i$, otherwise
  it would not be a root in the conclusion). In the second case,
  \textsc{T-Hide} can be pushed upwards.
  
  The fourth and fifth points are a consequence of the remark we made about the
  first two: if $o'$ is not a descendant of $o$ nor vice-versa,
  then the occurrences of
  \textsc{T-Hadd} and \textsc{T-Hide} concerning $o$ and its descendants
  commute with those concerning $o'$ and its descendants as they affect
  completely disjoint parts of the environment. In the case of the last
  occurrence of \textsc{T-Hadd} there may be a subsumption step but it
  is still possible to commute with it by pushing this subsumption step down again.
\end{proof}\fi

\begin{lemma}[Splitting of the heap]\label{lem:heapsplitting}
  Suppose $\typedheapchan\Theta{\Gamma, o : T}h$. Let $\Theta_1 =
  \Theta\setminus\chans(h\downarrow o)$ and let $\Theta_2$ be $\Theta$
  restricted to $\chans(h\downarrow o)$. Then we have:
  $\typedheapchan{\Theta_1}\Gamma{(h\uparrow o)}$ and
  $\typedheapchan{\Theta_2}{o : T}{(h\downarrow o)}$.
\end{lemma}
\ifinlineproofs\begin{proof}
  We know from Lemma \ref{lem:heapproperties} that $o$ is a root in
  $h$. We consider the particular derivation given by Lemma
  \ref{lem:rearrange-heap} where $o$ is the first root added to the
  heap. Now if we look at the conclusion of the last rule
  concerning $o$ (\textsc{T-Hadd} or \textsc{T-Hide} depending whether
  $T$ is a field or session type),
  we know that at this point the heap is $h\downarrow o$,
  and therefore the only object identifier in
  the environment is its only root: $o$.
  Furthermore, this part of the derivation is
  still true if we replace the initial $\Theta$ with $\Theta_2$, with
  the only difference that then the final $\Gamma$ contains no
  channels, and thus is of the form $o : T'$. We also know that the
  type of $o$ is not changed in the rest of the derivation except
  possibly by the subsumption step at the end; therefore
  $T'$ is a subtype of $T$. If they are session types, using
  Proposition \ref{prop:subsess} we
  can change the last occurrence of \textsc{T-Hide} to use $T$ instead
  of $T'$ and get $\typedheapchan{\Theta_2}{o : T}{(h\downarrow o)}$.
  Otherwise, we can add a subsumption step to the derivation for the
  sequence of swaps on the right of \textsc{T-Hadd} to get the same result.

  For the rest of the derivation, we know that $o$ is not used,
  therefore it can be removed from the initial environment without
  affecting the derivation except by the fact that it will not be in
  the final environment either. Furthermore, we know from Lemma
  \ref{lem:heapproperties} that the initial environment minus its only
  object identifier $o$ is included in
  $\Theta\setminus\chans(h\downarrow o) = \Theta_1$. More precisely, the lemma
  gives us inclusion of domains, but because subsumption is not used
  in the first part of the derivation we also know that the types are
  the same. Thus we can replace the first part of the derivation by an
  instance of \textsc{T-Hempty} using $\Theta_1$ and the second part
  is still valid (with all the descendants of $o$ removed from the
  heap), yielding $\typedheapchan{\Theta_1}\Gamma{(h\uparrow
    o)}$ at the bottom.
\end{proof}\fi

\begin{lemma}[Merging of heaps]\label{lem:heapmerging}
  Suppose $\typedheapchan\Theta\Gamma h$ and
  $\typedheapchan{\Theta'}{\Gamma'}{h'}$ with $\dom(h)\cap\dom(h') =
  \emptyset$ and $\dom(\Theta)\cap\dom(\Theta')=\emptyset$. Then we
  have $\typedheapchan{\Theta+\Theta'}{\Gamma+\Gamma'}{h+h'}$.
\end{lemma}
\ifinlineproofs\begin{proof}
  Since $\Theta$ and $\Theta'$ are disjoint, the
  channels in $\Theta'$ cannot appear anywhere in the typing
  derivation for $h$. Thus, it is possible to add $\Theta'$ to every
  typing environment occurring in the derivation for $h$ without
  altering its validity, yielding
  $\typedheapchan{\Theta+\Theta'}{\Gamma+\Theta'}h$. Looking now at
  the derivation for $h'$, since the domains of the heaps are disjoint and
  $\objs(\Gamma)\subset\roots(h)$, none of the identifiers in $\Gamma$
  can appear anywhere in it. Thus we can add $\Gamma$ to every typing
  environment and $h$ to every heap occurring in
  the derivation for $h'$, replacing the \textsc{T-Hempty} at the top
  with the conclusion of the other derivation, which yields the result
  we want.
\end{proof}\fi

\ifexplanatorytext
These two lemmas show, if we apply them repeatedly, that a typing
derivation for a heap can be considered as a set of separate typing
derivations leading to each root of the heap. This will allow us in
particular to
show results for particular cases where a heap has only one root and
generalize them.
\fi


\begin{lemma}\label{lem:heaprenaming}
  Suppose $\typedheapchan\Theta{o : S}h$. Let $\varphi$ be an injective
  function from $\dom(h)$ to $\mathcal{O}$. Then we have
  $\typedheapchan\Theta{\varphi(o) : S}{\varphi(h)}$.
\end{lemma}
\ifinlineproofs\begin{proof}
  Straightforward. Changing the names does not affect the typing derivation
  in any way.
\end{proof}\fi

\begin{lemma}[Opening]\label{lem:opening}
  If $\typedheapchan{\Theta}{\Gamma}{h}$, if $\Gamma(r)$ is a branch
  session type $S$ and if $h(r)$ is an object
  identifier $o$, then
  we know from Lemma \ref{lem:heapproperties} that $h$ contains an
  entry for $o$. Let $C$ be the class of this entry, then
  there exists a field typing $F$ for $C$ such that
  $\typedheapchan\Theta{\changetype{\Gamma}{r}{\objecttype{C}{F}}}{h}$
  and $\typedsess F C S$.
\end{lemma}
\ifinlineproofs\begin{proof}
  We prove this by induction on the depth of $r$. The base case is $r
  = o$. Using Lemmas \ref{lem:heapsplitting} and
  \ref{lem:heapmerging}, we can restrict ourselves to the case where
  $o$ is the only root of $h$. In that case we know that the last rule
  used in the typing derivation for $\typedheapchan\Theta{o : S}h$
  must be \textsc{T-Hide}. The result we want is constituted precisely
  by the premises of that rule.

  For the inductive case, $r$ is of the form $o'.f.\vec f$.
  We consider the case where $o'$ is the only root. The typing
  derivation then ends with \textsc{T-Hadd} and $f$ gets populated in
  the sequence of swaps by some object identifier\footnote{We know
    $o''$ is an object identifier and not a channel
    endpoint because, according to the hypotheses, either it is $o$
    itself or $\vec f$ is nonempty, implying $o''$ has fields.} $o''$.
  Let $r' = o''.\vec f$, and consider what $\Gamma(r')$ can be, knowing
  that in the conclusion $r$ has a branch session type: the only way
  the type can be modified in the sequence of swaps is by subsumption.
  Indeed, \textsc{T-VarS}, the other possibility, introduces a variant
  type. Therefore $\Gamma(r') = S'$ with $S'\subt S$. We can thus
  use the induction hypothesis to replace $\Gamma$ with
  $\changetype\Gamma{r}{\objecttype C F}$ on the left premise, with
  \typedsess F C {S'}. Then just use Proposition \ref{prop:subsess}
  to see that we also have \typedsess F C S and see that the type
  yielded in the conclusion by this new premise is what we want.
\end{proof}\fi

\begin{lemma}[Closing]\label{lem:closing}
  If $\typedheapchan\Theta{\Gamma}{h}$ and $\Gamma(r) = \objecttype{C}{F}$
  and $\typedsess F C S$, then
  $\typedheapchan\Theta{\changetype{\Gamma}{r}{S}}{h}$.
\end{lemma}
\ifinlineproofs\begin{proof}
  Again we prove this by induction on the depth of $r$ and the base
  case is $r = o$. In that case the lemma is nothing more than
  \textsc{T-Hide}. The inductive case is very similar to the above: we
  look at the type of $r'$ (defined as above) in the $\Gamma$ on
  the left premise of
  the last \textsc{T-Hadd}, noticing that the type of $r$ in the
  conclusion can only differ from it by subsumption,
  this time because we know from Lemma
  \ref{lem:heapfield} that \textsc{T-VarF} is never used. Hence the
  original type is $\objecttype C{F'}$ with $F'\subt F$. Proposition
  \ref{prop:subfield} gives us $\typedsess{F'}CS$ and thus we can use
  the induction hypothesis to change the type of $r'$ in this premise,
  which propagates to the type of $r$ in the conclusion.
\end{proof}\fi

\begin{lemma}[modification of the heap]\label{lem:heapchange}
Suppose that we have $\typedheapchan\Theta{\Gamma}{h}$ and
\judgment\envref\Gamma r > v' : T' < \envref{\Gamma'}{r} /,
and that $\Gamma'(r.f) = T$ where $T$ is not a variant. Let $v =
h(r).f$. The modified
heap $\changeval{h}{r.f}{v'}$ can be typed as follows:
\begin{enumerate}
\item if $v$ is an object identifier or a channel endpoint, then:
$$\typedheapchan\Theta{\changetype{\Gamma'}{r.f}{T'}, v : T}
{\changeval{h}{r.f}{v'}}$$
\item if $v$ is not an object or channel and $T$
  is not a link type, then:
$$\typedheapchan\Theta{\changetype{\Gamma'}{r.f}{T'}}
{\changeval{h}{r.f}{v'}}$$
\item if $v = l_0$ and $T = \linktype{}{f'}$,
  then:
  \begin{itemize}
  \item $\Gamma'(r.f') = \choice{l}{S_l}{l\in E}$ for some $E$
    such that $l_0\in E$ and
    some set of branch session types $S_l$. Note that this implies
    $f\neq f'$.
  \item $\typedheapchan\Theta
    {\changetype{\changetype{\Gamma'}{r.f}{T'}}{r.f'}{S_{l_0}}}
    {\changeval{h}{r.f}{v'}}$
  \end{itemize}
\end{enumerate}
\end{lemma}
\ifinlineproofs\begin{proof}
  First of all, note that the hypothesis that $\Gamma'(r.f)$ is
  defined implies that
  $\Gamma'(r)$ is not a variant, hence $T'$ is not $\linkthis$.
  In other words, the judgement cannot be derived from \textsc{T-VarF}.
  Furthermore, the fact that $T$ is not a variant either means that the
  judgement is not derived from an instance of \textsc{T-VarS} referring
  to field $f$. As this rule is the only possibility (beside subsumption) for a
  judgement typing a value to depend on, and modify, the type of a
  field, this implies that $\Gamma(r.f)$ is a subtype of $T$ and also
  that the judgement would still hold with another type for $f$. In
  particular we have
  \judgement\envref{\changetype{\Gamma}{r.f}{\nulltype}}{r} > v' : T' <
  \envref{\changetype{\Gamma'}{r.f}{\nulltype}}{r} /~\textbf{(a)}. We will in the
  following use this judgement (a) rather than the one in the hypothesis.

  We prove the lemma by induction on the depth of $r$, but the inductive case is
  straightforward (just apply the induction hypothesis to the left premise
  of \textsc{T-Hadd}). In the base case, $r$ is an object identifier
  $o$. We use Lemma \ref{lem:heapsplitting} to consider a typing
  derivation for the
  sub-heap $h\downarrow o$. Let $\Gamma_o = o : T_o$ and $\Theta_o$ be the
  environments corresponding to that part of the heap.
  %
  %
  We look at the application of \textsc{T-Hadd} which ends the
  derivation for $\typedheapchan{\Theta_o}{\Gamma_o}{h_o}$. As $T$ is not
  a variant, it is possible to consider that $f$ is the last field to get
  populated in the swap sequence. We thus have something of the form:
  \begin{center}\footnotesize\makebox[0pt][c]{
      \AxiomC{\typedheapchan{\Theta_o}{\Gamma_1}{(h\downarrow o)\setminus o}\hskip -3em\null}
      \AxiomC{\ldots}
      \AxiomC{\ldots}
      \rulename{1}
      \UnaryInfC{\judgement\envref{\Gamma_2}{o} > v : T_v <
        \envref{\Gamma_3}{o} /}
      \rulename{T-Swap}
      \UnaryInfC{\judgement\envref{\Gamma_2}{o} > \swap f v : \nulltype <
        \envref{\changetype{\Gamma_3}{o.f}{T_v}}{o} /}
      \rulename{T-Seq}
      \BinaryInfC{\judgement\envref{\Gamma_1}{o} > \seq{\ldots}{\swap f v} :
        \nulltype < \envref{\changetype{\Gamma_3}{o.f}{T_v}}{o} /}
      \rulename{T-SubEnv}
      \UnaryInfC{\judgement\envref{\Gamma_1}{o} > \seq{\ldots}{\swap f v} :
        \nulltype < \envref{\Gamma_o}{o} /}
      \rrulename{T-Hadd}
      \BinaryInfC{\typedheapchan{\Theta_o}{\Gamma_o}{h\downarrow o}}
      \DisplayProof
    }
  \end{center}
  with $T_v\subt T$ and $\changetype{\Gamma_3}{o.f}{T_v} \subt
  \Gamma_o$. If we change the type of $o.f$, this last relation becomes
  $\Gamma_3\subt\changetype{\Gamma_o}{o.f}{\nulltype}$ \textbf{(b)}.
  
  What we want to do is to replace the judgement on the top right, which
  is an application of some rule (1), by a judgement typing $v'$. For
  this, we need the rest of the environment. We
  consider the judgement for the rest of the heap,
  \typedheapchan{\Theta\setminus\Theta_o}{\Gamma\setminus o}{h\uparrow o}. Since
  the domains are disjoint, we can apply Lemma \ref{lem:heapmerging} to
  this and the leftmost premise of \textsc{T-Hadd}, yielding
  \typedheapchan{\Theta}{\Gamma\setminus o + \Gamma_1}{h\setminus o}. If we
  replace our left premise with this, the initial environment we get on
  the top right is now $\Gamma_4 = \Gamma\setminus o + \Gamma_2$,
  as the additional
  part is unaffected by the sequence of swaps. This environment is almost
  $\changetype\Gamma{o.f}{\nulltype}$, but not quite. We now have three cases
  depending on what rule (1) is, which correspond to the three cases of
  the lemma.
  \begin{enumerate}
  \item If (1) is \textsc{T-Ref} or \textsc{T-Chan}, meaning $v$ is an
    object identifier or a channel endpoint, then $\Gamma_2 = \Gamma_3,
    v : T_v$. Using (b), this yields $\Gamma_2 \subt
    \changetype{\Gamma_o}{o.f}{\nulltype}, v : T_v$.
    Adding $\Gamma\setminus o$ to both sides, we get
    $\Gamma_4\subt\changetype{\Gamma}{o.f}{\nulltype},
    v : T_v$.
    If we replace the initial environment in (a) with
    this one, we get the $v : T_v$ back in the final environment. We
    then use Lemma \ref{lem:moreweak} to replace this initial
    environment with $\Gamma_4$, and \textsc{T-SubEnv} to change $T_v$
    into $T$ in the final one:
    \judgement\envref{\Gamma_4}{o} > v' : T' <
    \envref{\changetype{\Gamma'}{o.f}{\nulltype}, v : T}{o} /. Just see that it
    yields what we want at the bottom of the derivation.
  \item If (1) is \textsc{T-Label} or \textsc{T-Null} or
    \textsc{T-Name}, \ie if $v$ is a literal value of non-link,
    non-linear type, then $\Gamma_2$ is identical to $\Gamma_3$ and we
    have, using (b) and adding $\Gamma\setminus o$ to both sides,
    $\Gamma_4\subt\changetype{\Gamma}{o.f}{\nulltype}$, so we
    can directly (with Lemma \ref{lem:moreweak}) use judgement (a).
  \item If (1) is \textsc{T-VarS}, the last possibility, then $v$ is a
    label $l_0$ and $T_v$ is $\linktype{}f'$ for some $f'$. As it has no strict
    supertype, we have $T = \linktype{}f'$ as well. We also have
    $\Gamma_3(o.f') = \choice{l_0}{S}{}$. From (b) we have that
    $\Gamma_o(o.f') = \Gamma(o.f')$ is a supertype of this variant type, thus also a
    variant. This implies that (a) cannot come from a \textsc{T-VarS}
    concerning $f'$; therefore,
    $\Gamma'(o.f')$ is a supertype of $\Gamma(o.f')$ and hence, by
    transitivity, of $\choice{l_0}{S}{}$, which gives us the first item
    of the conclusion, with $S\subt S_{l_0}$. We now just have to notice
    that $\Gamma_2 = \changetype{\Gamma_3}{o.f'}{S}$ and that (a) is
    independent of the type of $f'$ just like it is of the type of $f$,
    and we can conclude similarly to the two previous cases.
  \end{enumerate}
\end{proof}\fi

\begin{lemma}[Substitution]
\label{lem:substitution}
If $\judgment \envref{\this:\objecttype C F, x : T'}{\this} > e : T <
\envref{\this:\objecttype C {F'}}{\this} /$, and if $\Gamma(r) =
\objecttype C F$, then:
\begin{enumerate}
\item if $T'$ is a base type (\ie neither an object type nor a link)
  and $v$ is a literal value of that type, or if $v$ is an access
  point name declared with type $\access{\Sigma}$ and
  $\sem{\access\Sigma}\subt T'$, we have:
  $$\judgment \envref\Gamma r > e\subs{v}{x} : T <
  \envref{\changetype{\Gamma}{r}{\objecttype C {F'}}}{r} /.$$ 
\item if $T'$ is an object type and $v$ is an object identifier or a
  channel endpoint, we
  have: $$\judgment \envref{\Gamma, v : T'}{r} > e\subs{v}{x} : T <
  \envref{\changetype{\Gamma}{r}{\objecttype C {F'}}}{r} /.$$
\end{enumerate}
\end{lemma}
\ifinlineproofs\begin{proof}
  In order to do an induction, we add the following case where $x$ is
  still present in the final environment : if we have $\judgment
  \envref{\this:\objecttype C F, x : T'}{\this} > e : T <
  \envref{\this:\objecttype C {F'}, x : T''}{\this} /$, and if $T'$ is
  an object type and $v$ is an object identifier, then we have
  \judgment \envref{\Gamma, v : T'}{r} > e\subs{v}{x} : T <
  \envref{\changetype{\Gamma}{r}{\objecttype C {F'}}, v : T''}{r} /.

  We prove this by induction on the derivation of
  \judgment\envref{\this:\objecttype C F, x : T'}{\this} > e : T <
  \envref{\this:\objecttype C {F'}, V}{\this} / (where $V$ is either
  empty or $v : T''$ depending on the case).  For most toplevel rules,
  the result is immediate. The only ones for which it is not are
  \textsc{T-Var} and \textsc{T-LinVar}. For \textsc{T-Var} the result
  is obtained using either \textsc{T-Null} if $T'$ is $\nulltype$ or
  \textsc{T-Label} and \textsc{T-Sub} if it is an enumerated type. In
  the case of an extension adding new base types, we assume there is a similar
  rule to type the corresponding literal values.
  For \textsc{T-LinVar}, if $v$ is an access point name the result is
  obtained using \textsc{T-Name} and \textsc{T-Sub}. Otherwise, $v$ is
  an object identifier and
  the result is obtained using \textsc{T-Ref}, noticing that because
  $\Gamma(r)$ is defined and $v$ is not in $\Gamma$, the path $r$ does
  not start with $v$ and the premise is satisfied.
\end{proof}\fi

\begin{lemma}[Typability of Subterms]
\label{lem:typablesubterms}
  If $\mathcal{D}$ is a derivation of $\judgment \envref\Gamma r >
  \mathcal{E}(e) : T < \envref{\Gamma'}{r'} /$ then there exist
  $\Gamma_1$, $r_1$ and $U$ such that $\mathcal{D}$ has a
  subderivation $\mathcal{D}'$ concluding $\judgment \envref\Gamma r>
  e : U < \envref{\Gamma_1}{r_1} /$ and the position of $\mathcal{D}'$
  in $\mathcal{D}$ corresponds to the position of the hole in
  $\mathcal{E}$.
\end{lemma}
\ifinlineproofs\begin{proof}
  A straightforward induction on the structure of $\mathcal{E}$; the
  expression $e$ is always at the extreme left of the typing
  derivation for $\mathcal{E}(e)$.
\end{proof}\fi

\begin{lemma}[Replacement]
\label{lem:replacement}
If
\begin{enumerate}
\item $\mathcal{D}$ is a derivation of $\judgment \envref\Gamma r >
  \mathcal{E}(e) : T < \envref{\Gamma'}{r'} /$
\item $\mathcal{D}'$ is a subderivation of $\mathcal{D}$ concluding
  $\judgment \envref\Gamma r > e : U < \envref{\Gamma_1}{r_1}/$
\item the position of $\mathcal{D}'$ in $\mathcal{D}$ corresponds to
  the position of the hole in $\mathcal{E}$
\item $\judgment \envref{\Gamma''}{r''} > e' : U < \envref{\Gamma_1}{r_1}/$
\end{enumerate}
then $\judgment \envref{\Gamma''}{r''} > \mathcal{E}(e') : T < \envref{\Gamma'}{r'}/$.
\end{lemma}
\ifinlineproofs\begin{proof}
Replace $\mathcal{D}'$ in $\mathcal{D}$ by the derivation of
$\judgment \envref{\Gamma''}{r''} > e' : U < \envref{\Gamma_1}{r_1}/$.
\end{proof}\fi

\subsection{Type preservation}

We use $\vdash s$ as an abbreviation for $\emptyset\vdash s$;
this represents well-typedness of a closed configuration.
We have the following result:
\begin{theorem}[Subject Reduction]
  \label{thm:subjectreduction}
  If, in a context parameterised by a set of well-typed declarations,
  we have $~\vdash s$ and $s\vreduces s'$, then $~\vdash s'$.
\end{theorem}

This global result is a consequence of a subject reduction theorem for
a single thread, which is similar but not identical to what we stated
as Theorem \ref{thm:subjectreductionSeq} (which will be a particular case).
The
reason it is not identical is that we need to prove that the type of
an expression is preserved not only when this expression reduces on
its own but also when it communicates with another thread. In order to
state precisely this thread-wise type preservation theorem, we introduce a labelled
transition system for threads. Transition labels can be: $\tau$
indicating internal reduction, $c^p\send{v}$ or $c^p\rcv{v}$ indicating that
the non-object value $v$ is sent or received on channel $c^p$,
$c^p\send{h}$ or $c^p\rcv{h}$, where $h$ is a heap with a single root
$o$, indicating that the object $o$ (together with its content) is
sent or received on channel $c^p$, $n[c^p]$ indicating that the
channel endpoint $c^p$ is received from access point $n$, or, finally,
$\methcal{C}{m}{}$ indicating that the thread spawns another one using
method $m$ of class $C$.
\begin{definition}[Labelled transition system]
  We define a labelled transition system for threads by the
  following rules:
\begin{center}\setlength\lineskip\typingRuleSkip
\AxiomC{$\state{h}{r}{e}\reduces\state{h'}{r'}{e'}$}
\rulename{Tr-Red}
\UnaryInfC{$\labtrans{\state{h}{r}{e}}{\tau}{\state{h'}{r'}{e'}}$}
\DisplayProof\hfil
\AxiomC{$h(r).f = c^p$}
\AxiomC{$v\not\in\mathcal{O}$}
\rulename{Tr-Send}
\BinaryInfC{$\labtrans{\state{h}{r}{\mathcal{E}[\methcal{f}{\sendterm}{v}]}}
  {c^p\send{v}}
  {\state{h}{r}{\mathcal{E}[\nullterm]}}$
}
\DisplayProof\hfil
\AxiomC{$h(r).f = c^p$}
\rulename{Tr-SendObj}
\UnaryInfC{
$\labtrans{\state{h}{r}{\mathcal{E}[\methcal{f}{\sendterm}{o}]}}
  {c^p\send{h\downarrow o}}
  {\state{h\uparrow o}{r}{\mathcal{E}[\nullterm]}}$
}
\DisplayProof\hfil
\AxiomC{$h(r).f = c^p$}
\AxiomC{$v\not\in\mathcal{O}$}
\rulename{Tr-Receive}
\BinaryInfC{$\labtrans{\state{h}{r}{\mathcal{E}[\methcal{f}{\rcvterm}{}]}}
  {c^p\rcv{v}}
  {\state{h}{r}{\mathcal{E}[v]}}$
}
\DisplayProof\hfil
\AxiomC{$h(r).f = c^p$}
\AxiomC{$\roots(h') = \{o\}$}
\AxiomC{$\dom(h)\cap\dom(h') = \emptyset$}
\rulename{Tr-RcvObj}
\TrinaryInfC{$\labtrans{\state{h}{r}{\mathcal{E}[\methcal{f}{\rcvterm}{}]}}
  {c^p\rcv{h'}}
  {\state{h+h'}{r}{\mathcal{E}[o]}}$
}
\DisplayProof\hfil
\AxiomC{$h(r).f = n$}
\rulename{Tr-Accept}
\UnaryInfC{$\labtrans{\state{h}{r}{\mathcal{E}[\methcal{f}{\acceptterm}{}]}}
  {n[c^+]}
  {\state{h}{r}{\mathcal{E}[c^+]}}$
}
\DisplayProof\hfil
\AxiomC{$h(r).f = n$}
\rulename{Tr-Request}
\UnaryInfC{$\labtrans{\state{h}{r}{\mathcal{E}[\methcal{f}{\requestterm}{}]}}
  {n[c^-]}
  {\state{h}{r}{\mathcal{E}[c^-]}}$
}
\DisplayProof\hfil
\axiomname{Tr-Spawn}
$\labtrans{\state{h}{r}{\mathcal{E}[\spawn{\methcal{C}{m}{v}}]}}
  {\methcal{C}{m}{}}
  {\state{h}{r}{\mathcal{E}[\nullterm]}}$
\end{center}
Note that \emph{both} $\tau$ and $\methcal{C}{m}{}$ correspond to the
thread being able to reduce on its own. An important feature of this
transition relation is that, for all rules, the right-hand state is
fully determined by the left-hand one and the transition
label. Moreover, the only case where several different transitions are
possible from a given state is when applying the rule
$\mkterm{receive}$, as the right-hand side depends on the value
received.
\end{definition}

\begin{definition}\label{def:thetatrans}
A similar transition relation, with the same set of labels,
is defined on channel environments $\Theta$ as follows:
\begin{center}\setlength\lineskip\typingRuleSkip
%
%
$\labtrans\Theta\tau\Theta$
\hfil
%
%
$\labtrans\Theta{\methcal{C}{m}{}}\Theta$
\hfil
%
%
\AxiomC{$n.\mkterm{protocol} = \access{\Sigma}$}
\AxiomC{$\forall p,c^p\not\in\dom(\Theta)$}
\BinaryInfC{
$\labtrans{\Theta}{n[c^+]}{\Theta, c^+ : \Sigma}$
}
\DisplayProof\hfil
%
%
\AxiomC{$n.\mkterm{protocol} = \access{\Sigma}$}
\AxiomC{$\forall p,c^p\not\in\dom(\Theta)$}
\BinaryInfC{
$\labtrans{\Theta}{n[c^-]}{\Theta, c^- : \overline\Sigma}$
}
\DisplayProof\hfil
%
%
\AxiomC{$\sem{\Sigma'}\subt T$}
\UnaryInfC{
$\labtrans{\Theta, c^p : \send{T}.\Sigma, c'^{p'} : \Sigma'}
{c^p\send{c'^{p'}}}
{\Theta, c^p : \Sigma}$
}
\DisplayProof\hfil
%
%
\AxiomC{$\sem{\Sigma'}\subt T$}
\UnaryInfC{
$\labtrans{\Theta, c^p : \rcv{T}.\Sigma}
{c^p\rcv{c'^{p'}}}
{\Theta, c^p : \Sigma, c'^{p'} : \Sigma'}$
}
\DisplayProof\\
%
%
\AxiomC{\judgment\emptyset > v : T < \emptyset /}
\UnaryInfC{
$\labtrans{\Theta, c^p : \send{T}.\Sigma}
{c^p\send{v}}
{\Theta, c^p : \Sigma}$
}
\DisplayProof\hfil
%
%
\AxiomC{\judgment\emptyset > v : T < \emptyset /}
\UnaryInfC{
$\labtrans{\Theta, c^p : \rcv{T}.\Sigma}
{c^p\rcv{v}}
{\Theta, c^p : \Sigma}$
}
\DisplayProof\hfil
%
%
\AxiomC{$l_0\in E$}
\UnaryInfC{
$\labtrans{\Theta, c^p : \chanchoice{l : \Sigma_l}_{l\in E}}
{c^p\send{l_0}}
{\Theta, c^p : \Sigma_{l_0}}$
}
\DisplayProof\hfil
%
%
\AxiomC{$l_0\in E$}
\UnaryInfC{
$\labtrans{\Theta, c^p : \chanbranch{l : \Sigma_l}_{l\in E}}
{c^p\rcv{l_0}}
{\Theta, c^p : \Sigma_{l_0}}$
}
\DisplayProof\hfil
%
%
\AxiomC{\typedheapchan{\Theta_1}{o : S}{h}}
\AxiomC{$\dom(\Theta_1)\subset\chans(h)$}
\BinaryInfC{
$\labtrans{\Theta_1 + \Theta_2, c^p : \send{S}.\Sigma}
{c^p\send{h}}
{\Theta_2, c^p : \Sigma}$
}
\DisplayProof\hfil
%
%
\AxiomC{\typedheapchan{\Theta'}{o : S}{h}}
\AxiomC{$\dom(\Theta')\subset\chans(h)$}
\AxiomC{$\dom(\Theta)\cap\dom(\Theta') = \emptyset$}
\TrinaryInfC{
$\labtrans{\Theta, c^p : \rcv{S}.\Sigma}
{c^p\rcv{h}}
{\Theta + \Theta', c^p : \Sigma}$
}
\DisplayProof\hfil
\end{center}

Where we use \judgment\emptyset > v : T < \emptyset / as an
abbreviation for \judgment\envref{\mkterm{dummy} : \objecttype C
  {}}{\mkterm{dummy}} > v : T < \envref{\mkterm{dummy} : \objecttype C
  {}}{\mkterm{dummy}} / --- meaning that $v$ is a literal value (or
access point name) of type $T$.
\end{definition}

We can now state our thread-wise type preservation theorem.
\begin{theorem}[Thread-wise progress and type preservation]\label{thm:threadsr}
  Let $\mathcal{D}$ be a set of well-typed declarations, that is, such
  that for every class declaration $D$ in $\mathcal{D}$ we have $\vdash
  D$.  In a context parameterised by $\mathcal{D}$, suppose we have
  $\judgment \Theta;\Gamma > \state{h}{r}{e} : T < \envref{\Gamma'}{r'} /$.

  Then either $e$ is a value or
  there exists a transition label $\lambda$ such that we have
  $\labtrans{\state{h}{r}{e}}{\lambda}{\state{h'}{r''}{e'}}$ for some
  $h'$, $r''$ and $e'$.

  Furthermore, if $\lambda$ is such that
  $\labtrans{\Theta}{\lambda}{\Theta'}$  for some  $\Theta'$, then
  there exists $\Gamma''$ such that
  \judgment\Theta';\Gamma'' > \state{h'}{r''}{e'} : T <
  \envref{\Gamma'}{r'}/ holds.
\end{theorem}
Theorem \ref{thm:subjectreductionSeq} is the particular case where
$\lambda = \tau$.

\begin{corollary}[Theorem \ref{thm:typesafetySeq}]
  If $\mathcal{D}$ contains no name declaration and $\Theta$ is empty,
  then there exists $s'$ such that $\state{h}{r}{e}\reduces s'$.
\end{corollary}
\begin{proof}[(Corollary)]
  In that particular case, \typedheapchan\Theta\Gamma h implies that the heap
  cannot contain any $n$ or $c^p$, hence $\lambda$ can only be $\tau$
  or of the form $\methcal{C}{m}{}$.
\end{proof}

\begin{proof}[(Theorem)]
  We always use typing derivations where subsumption steps only occur
  at the positions described in Lemma \ref{lem:rearrange-expr}.
  Furthermore, it is sufficient to consider only cases where
  subsumption does not occur at the end: indeed, if it does occur,
  then we can add a similar subsumption step to the new judgement.
  The hypothesis in the theorem that $\judgment
  \Theta;\Gamma > \state{h}{r}{e} : T < \envref{\Gamma'}{r'} /$ holds is necessarily a
  result of \textsc{T-State} and therefore is equivalent to the two
  hypotheses $\typedheapchan\Theta\Gamma h$ and \judgment\envref{\Gamma}{r} > e : T <
  \envref{\Gamma'}{r'} /, which we will sometimes refer to directly.

  We prove the theorem by induction on the structure of $e$ with
  respect to contexts, and present the inductive case first:

  If $e$ is of the form $\mathcal{E}[e_1]$ where
  $e_1$ is not a value and $\mathcal{E}$ is not just $[\_]$, then Lemma
  \ref{lem:typablesubterms} tells us that $\judgment\envref{\Gamma}{r} > e_1 : U <
  \envref{\Gamma_1}{r_1}/$ appears in the typing derivation of
  \judgment\envref\Gamma r > e :
  T < \envref{\Gamma'}{r'}/ for some $U$, $r_1$ and $\Gamma_1$. From there we can apply
  \textsc{T-State} and derive
  $\judgment \Theta;\Gamma > \state{h}{r}{e_1} : U <
  \envref{\Gamma_1}{r_1}/$. This allows us 
  to use the induction hypothesis and get $\lambda, e_2, r''$ and $h'$
  such that
  $\labtrans{\state{h}{r}{e_1}}{\lambda}{\state{h'}{r''}{e_2}}$.
  Then we straightforwardly have
  $\labtrans{e}{\lambda}{\mathcal{E}[e_2]}$, either by applying
  \textsc{R-Context} if $\lambda$ is $\tau$ or by replacing the
  context in the transition rule if it is something else.
  Now if $\lambda$ is such that $\labtrans\Theta\lambda{\Theta'}$, then
  the induction hypothesis\footnote{clearly there is no $\lambda$ such that we
  would have $\labtrans{\mathcal{E}[e_1]}{\lambda}{}$ but not
  $\labtrans{e_1}{\lambda}{}$, hence it is legitimate to use the
  induction hypothesis here.} also gives us $\Gamma''$ such that
  \judgment\Theta';\Gamma'' >
  \state{h'}{r''}{e_2} : U < \envref{\Gamma_1}{r_1} / holds.
   From this we get, by reading \textsc{T-State} upwards,
  \typedheapchan{\Theta'}{\Gamma''}{h'} and \judgment\envref{\Gamma''}{r''} > e_2 :
  U < \envref{\Gamma_1}{r_1} /. We use Lemma~\ref{lem:replacement} with the latter
  in order to obtain \judgment \envref{\Gamma''}{r''} > \mathcal{E}[e_2] : T <
  \envref{\Gamma'}{r'} / and conclude with \textsc{T-State}.

  The base cases are if $e$ is of the form $\mathcal{E}[v]$
  with $\mathcal{E}$ elementary (i.e.\ not of the form
  $\mathcal{E'}[\mathcal{E''}]$ with $\mathcal{E''}\neq [\_]$) and if it
  is not of the form $\mathcal{E}[e_1]$ at all. We list them below.
  \begin{itemize}
  \item If $e$ is a value, there is nothing to prove.

  \item $e$ cannot be a variable. Indeed, \typedheapchan\Theta\Gamma h implies
    that $\dom(\Gamma)$ contains only object identifiers and channel
    endpoints. Therefore, 
    \judgment\envref{\Gamma}{r} > e : T < \envref{\Gamma'}{r'}/
    cannot be a conclusion of
    \textsc{T-Var} or \textsc{T-LinVar}, thus $e$ is not a variable.

  \item $e = \seq{v}{e'}$. Then the expression reduces by
    \textsc{R-Seq} and the initial derivation is as follows:
    \begin{center}\footnotesize
      \AxiomC{\typedheapchan\Theta\Gamma h \textbf{(a)}}
      \AxiomC{\dots}
      \rulename{1}
      \UnaryInfC{\judgment \envref{\Gamma}{r} > v : T' < \envref{\Gamma_1}{r} /}
      \AxiomC{\judgment \envref{\Gamma_1}{r} > e' : T < \envref{\Gamma'}{r'} / \textbf{(b)}}
      \rrulename{T-Seq}
      \BinaryInfC{\judgment\envref\Gamma r > \seq v {e'} : T <
        \envref{\Gamma'}{r'} /}
      \rrulename{T-State}
      \BinaryInfC{\judgment\Theta;\Gamma > \state{h}{r}{\seq v {e'}} : T < \envref{\Gamma'}{r'} /}
      \DisplayProof
    \end{center}
    Furthermore, $T'$ is not a link type. Therefore, (1) cannot be
    \textsc{T-VarF} or \textsc{T-VarS}
    and it is either \textsc{T-Ref}, \textsc{T-Chan}, \textsc{T-Name},
    \textsc{T-Label} or
    \textsc{T-Null}, since these are the only rules for typing values. If it is
    \textsc{T-Null}, \textsc{T-Label} or \textsc{T-Name},
    then $\Gamma = \Gamma_1$; if it is \textsc{T-Ref} or
    \textsc{T-Chan}, then $\Gamma \subt \Gamma_1$ and
    we can use Lemma \ref{lem:moreweak} to get
    \judgment\envref{\Gamma}{r} > e' : T < \envref{\Gamma'}{r'} /
    from (b) in both cases. We conclude from this
    using (a) and \textsc{T-State}.

  \item $e = \new C$. Then the expression reduces by \textsc{R-New} and
    the initial reduction is as follows:
    \begin{center}\footnotesize
      \AxiomC{\typedheapchan\Theta\Gamma h \textbf{(a)}}
      \AxiomC{}
      \rrulename{T-New}
      \UnaryInfC{\judgment\envref{\Gamma}{r} > \new C : C.\sessterm <
        \envref{\Gamma} r /}
      \rulename{T-State}
      \BinaryInfC{\judgment \Theta;\Gamma > \state{h}{r}{\new C} : C.\sessterm <
        \envref{\Gamma}{r} /}
      \DisplayProof
    \end{center}
    Let $S = C.\sessterm$. From the hypothesis that $\mathcal{D}$ is well-typed,
    we have $\vdash\class{C}{S}{\vec f}{\vec M}$. This must come from
    \textsc{T-Class}, therefore we have
    $\typedsess{\overrightarrow\nulltype\,\vec f}{C}{S}$ \textbf{(b)}.

    We build the following derivation:
    \begin{center}\footnotesize
      \AxiomC{(a)}
      \AxiomC{\textsc{T-Null, T-Swap, T-Seq}}
      \UnaryInfC{\judgment\envref{\Gamma, o : \objecttype C
          {\overrightarrow\nulltype\,\vec f}}{o} >
        \swap{\vec f}{\overrightarrow\nulltype} : \nulltype <
        \envref{\Gamma, o : \objecttype C
          {\overrightarrow\nulltype\,\vec f}}{o} /}
      \rulename{T-Hadd}
      \BinaryInfC{\typedheapchan\Theta{\Gamma, o : \objecttype C
          {\overrightarrow\nulltype\,\vec
            f}}{\hadd{h}{o=\hentry{C}{\vec f =
              \overrightarrow\nullterm}}}}
      \AxiomC{(b)}
      \rulename{T-Hide}
      \BinaryInfC{\typedheapchan\Theta{\Gamma, o : S}{\hadd{h}{o=\hentry{C}
            {\vec f = \overrightarrow\nullterm}}}}
      \DisplayProof
    \end{center}
    then conclude \judgment\Theta;\Gamma, o : S >
    \state{\hadd{h}{o=\hentry{C}{\vec f =
          \overrightarrow\nullterm}}}{r}{o} : S < \envref{\Gamma}{r} /
    using \textsc{T-Ref} (it is not possible that $r$ starts with $o$
    since $o$ is fresh) and \textsc{T-State}.

  \item $e = \switch{v}{l}{e_l}{l\in E}$. Then we have two cases.
    The slightly more complex one is if the initial derivation is as follows:
    \begin{center}\footnotesize
      \AxiomC{$v$ is a label}
      \UnaryInfC{\judgment\envref{\Gamma}{r} > v : \{v\} <
        \envref{\Gamma}{r} /}
      \AxiomC{$\Gamma(r.f) = S$}
      \BinaryInfC{\judgment\envref{\Gamma}{r} > v : \linktype{}f <
        \envref{\changetype\Gamma{r.f}{\choice{v}{S}{}}}{r} /}
      \AxiomC{$v\in E$ \textbf{(b)}}
      \AxiomC{\judgment\envref\Gamma r > e_v :
        T < \envref{\Gamma'}{r} / \textbf{(c)}}
      \TrinaryInfC{\judgment\envref\Gamma r > \switch{v}{l}{e_l}{l\in E} : T
        < \envref{\Gamma'}{r}/}
      \DisplayProof
    \end{center}(using \textsc{T-Label}, \textsc{T-VarS},
    \textsc{T-SwitchLink} top to bottom).
    As usual we also have \typedheapchan\Theta\Gamma h \textbf{(a)}
    as the other premise of \textsc{T-State} (omitted for lack of
    space). The reason why the initial environment of judgement (c) is
    $\Gamma$ is because it is obtained from the version of
    $\Gamma$ with the type of $r.f$ modified by modifying this type
    again, putting back $S$ instead of the variant.
    (b) implies that the expression reduces by \textsc{R-Switch}. As
    regards type preservation, we can
    conclude $\judgment \Theta;\Gamma > \state{h}{r}{e_v} : T < \envref{\Gamma'}{r}/$
    directly from (a), (c), and \textsc{T-State}.
    The other case is when the \textsc{T-VarS} step is absent and the
    following rule is \textsc{T-Switch} instead of \textsc{T-SwitchLink}; the
    argument is the same.
    
  \item $e = \swap f v$. Then the initial derivation is as follows:
    \begin{center}\footnotesize
      \AxiomC{$\ldots$}
      \UnaryInfC{\judgment\envref\Gamma r > v' : T' <
        \envref{\Gamma_1}{r} / \textbf{(b)}}
      \AxiomC{$\Gamma_1(r.f) = T$ \textbf{(c)}}
      \AxiomC{$T$ is not a variant \textbf{(d)}}
      \rrulename{T-Swap}
      \TrinaryInfC{\judgment\envref\Gamma r > \swap f {v'} : T <
        \envref{\changetype{\Gamma_1}{r.f}{T'}}{r} /}
      \DisplayProof
    \end{center}
    and we also have, as usual, \typedheapchan\Theta\Gamma h
    \textbf{(a)}. The fact that
    $\Gamma_1(r.f)$ is defined implies that $\Gamma(r.f)$ is also
    defined, indeed the effect of typing $v$ can only remove from the
    environment or create a variant type, so it can only decrease the
    set of valid field references. Thus $h(r).f$ is defined as
    well, and the expression reduces by \textsc{R-Swap}.
    Let $v = h(r).f$. From (a), (b), (c) and (d),
    we use Lemma \ref{lem:heapchange} to get $\Gamma''$
    such that $\typedheapchan\Theta{\Gamma''}{\changeval{h}{r.f}{v'}}$.
    We then notice that
    in each of the three cases of the lemma we have
    $\judgment\envref{\Gamma''}{r} > v : T <
    \envref{\changetype{\Gamma_1}{r.f}{T'}}{r} /$:
    \begin{enumerate}
    \item If $v$ is an object identifier or channel endpoint, then $\Gamma'' =
      \changetype{\Gamma_1}{r.f}{T'}, v : T$. We use \textsc{T-Ref}
      or \textsc{T-Chan}.
    \item If $v$ is not an object or channel and $T$ is not a link type, then $v$
      is either $\nullterm$, an access point name or a label. We use
      \textsc{T-Null}, \textsc{T-Name} or \textsc{T-Label}.
    \item If $T = \linktype{}{f'}$, then $\Gamma_1(r.f') =
      \choice{l}{S_l}{l\in E}$ with $v\in E$ and we have
      $\Gamma'' =
      \changetype{\changetype{\Gamma_1}{r.f}{T'}}{r.f'}{S_v}$.
      We use \textsc{T-Label}, \textsc{T-VarS} and \textsc{T-SubEnv}.
    \end{enumerate} 
    Finally we conclude with \textsc{T-State}.

  \item $e = \return{v}{r}$. Then the expression reduces by
    \textsc{R-Return}. The initial derivation is as follows:
    \begin{center}\footnotesize
      \AxiomC{\dots}
      \rulename{1}
      \UnaryInfC{\judgment \envref\Gamma {r.f} > v : T_1 <
        \envref{\Gamma_1}{r.f} /}
      \AxiomC{$\Gamma_1(r.f) = \objecttype C F$}
      \AxiomC{\typedsess F C {S} \textbf{(b)}}
      \rrulename{T-Return}
      \TrinaryInfC{\judgment \envref\Gamma{r.f} >
        \return v {} : T <
        \envref{\changetype{\Gamma_1}{r.f}{S}}{r} /}
      \DisplayProof
    \end{center}
    with also \typedheapchan\Theta\Gamma h \textbf{(a)}.
    We distinguish cases depending on what rule (1) is:
    \begin{itemize}
    \item If (1) is \textsc{T-Null}, \textsc{T-Name} or \textsc{T-Label},
      then $\Gamma = \Gamma_1$, and if it is
      \textsc{T-Ref} or \textsc{T-Chan}, then
      $\Gamma = \Gamma_1, v : T_1$. In both cases we have
      $\Gamma(r.f) = \Gamma_1(r.f)$ and $T_1 = T$. From (a) and (b) we deduce
      \typedheapchan\Theta{\changetype{\Gamma}{r.f}{S}}{h} using the
      closing lemma (Lemma \ref{lem:closing}).
      We then
      use \textsc{T-Null}, \textsc{T-Name}, \textsc{T-Label},
      \textsc{T-Chan} or \textsc{T-Ref},
      as appropriate, to get
      \judgment\envref{\changetype{\Gamma}{r.f}{S}}{r} > v :
      T_1 <
      \envref{\changetype{\Gamma_1}{r.f}{S}}{r} /,
      and we conclude with \textsc{T-State}.
    \item (1) cannot be \textsc{T-VarS} because \textsc{T-Return} forbids
      that $T_1$ be of the form $\linktype{}{f'}$.
    \item If (1) is \textsc{T-VarF}, then
      $T_1 = \linkthis$ and $T = \linktype{}{f}$.
      Furthermore, $v$ is a label, $F = \vfield{v}{F'}{}$ with $F'$
      not a variant, and $\Gamma = \changetype{\Gamma_1}{r.f}{\objecttype{C}{F'}}$.
      (b) then implies that
      $S$ is of the form $\choice{l}{S_l}{l\in E}$ with $v\in E$
      and $\typedsess{F'}{C}{S_v}$. Note that because
      $F'$ is not a variant, $S_v$ must be a branch. Now,
      from that judgement and (a), we use the
      closing lemma to get
      \typedheapchan\Theta{\changetype{\Gamma}{r.f}{S_v}}{h}.
      Let $\Gamma'' = \changetype{\Gamma}{r.f}{S_v}$.
      Since $\Gamma$ only differs from $\Gamma_1$ by the type of
      $r.f$, it is also the case of $\Gamma''$, and as
      $\choice{v}{S_v}{}$ is a subtype of $S$,
      we have $\changetype{\Gamma''}{r.f}{\choice{v}{S_v}{}} \subt
      \changetype{\Gamma_1}{r.f}{S}$.
      From all this, we build the following derivation:
      \begin{center}\footnotesize\null\hskip -2.5em\makebox[0pt][c]{
        \AxiomC{\typedheapchan\Theta{\changetype{\Gamma}{r.f}{S_v}}{h}}
        \AxiomC{$v$ is a label}
        \rulename{T-Label}
        \UnaryInfC{\judgment\envref{\Gamma''}{r} > v :
          \{v\} < \envref{\Gamma''}{r} /}
        \AxiomC{$S_v$ branch}
        \rulename{T-VarS}
        \BinaryInfC{\judgment\envref{\Gamma''}{r} >
          v : T <
          \envref{\changetype{\Gamma''}{r.f}{\choice{v}{S_v}{}}}{r} /}
        \rulename{T-SubEnv}
        \UnaryInfC{\judgment\envref{\Gamma''}{r} >
          v : T <
          \envref{\Gamma'}{r} /}
        \rulename{T-State}
        \BinaryInfC{\judgment\Theta;\Gamma'' >
          \state{h}{r}{v} : T < \envref{\Gamma'}{r} /}
        \DisplayProof
      }
      \end{center}
    \end{itemize}

  \item $e = \spawn{\methcal{C}{m}{v}}$. The initial derivation involves
    \textsc{T-Spawn}, and $v$ is $\nullterm$.
    The premise that the method exists implies that
    the state can reduce by \textsc{R-Spawn}, which corresponds to a
    $\methcal{C}{m}{}$ transition. The new derivation is obtained
    replacing \textsc{T-Spawn} with \textsc{T-Null}.

  \item $e = \methcal{f}{m}{v}$.  The initial derivation is
    as follows, with $m = m_j$ and $j\in I$:
    \begin{center}\footnotesize
        \AxiomC{\ldots}
        \rulename{1}
        \UnaryInfC{\judgment\envref\Gamma r > v : T' <
          \envref{\Gamma_1}{r} /}
        \rulename{T-Sub}
        \UnaryInfC{\judgment\envref\Gamma r > v : T'_j <
          \envref{\Gamma_1}{r} /}
        \AxiomC{$\Gamma_1(r.f) =
          \branch{\methsign{m_i}{T'_i}{T_i}}{S_i}{i\in I}$ \textbf{(b)}}
        \rulename{T-Call}
        \BinaryInfC{\judgment \envref\Gamma r> \methcal{f}{m_j}{v} :
          T <
          \envref{\changetype{\Gamma_1}{r.f}{S_j}}{r} /}
        \DisplayProof
    \end{center}
    and we also have \typedheapchan\Theta\Gamma h \textbf{(a)}. $T$
    is obtained from $T_j$ as specified in \textsc{T-Call}, \ie
    replacing $\linkthis$ with $\linktype{}f$ if necessary. Let $S
    = \branch{\methsign{m_i}{T'_i}{T_i}}{S_i}{i\in I}$.
    First note that $T'_j$ is a part of a method signature and that only a
    restricted set of types is allowed there: it cannot be of the form
    $\linktype{}{f'}$. Furthermore, (1) cannot be \textsc{T-VarF} because
    of (b), thus $T'$ is not $\linkthis$ either.
    Indeed, if $\Gamma_1(r)$ were a variant, $\Gamma_1(r.f)$ would
    not be defined. Therefore (1) is either \textsc{T-Null},
    \textsc{T-Label}, \textsc{T-Chan}, \textsc{T-Name} or
    \textsc{T-Ref} and in all cases we have $\Gamma(r.f) = \Gamma_1(r.f)$.
    As it is a session type, it implies because of (a)
    that $h(r).f$ exists and is either an object
    identifier, an access point name or a channel endpoint. We
    distinguish these three cases:
    \begin{itemize}
    \item $h(r).f$ is an object identifier $o$.
      We use (a) and the
      opening lemma (Lemma \ref{lem:opening}) to get a field typing
      $\objecttype C F$ such that
      $\typedheapchan\Theta{\changetype{\Gamma}{r.f}{\objecttype C {F}}}{h}$ 
      and $\typedsess{F}{C}{S}$. This last judgement implies, by
      definition, that $F$
      is not a variant; that, among others, method $m_j$ appears in
      the declaration of class $C$;  and that, if $e_j$ is its body and $x$
      its parameter, we have
      $\judgment \envref{x:T'_j, \this:\objecttype C F}{\this} > e_j : T_j <
      \envref{\this:\objecttype{C}{F_j}}{\this} /$ and
      $\typedsess{F_j}{C}{S_j}$. The fact that the method is declared
      implies
      $\state{h}{r}{e}\reduces\state{h}{r.f}{\return{e_j\subs{v}{x}}{}}$;
      we now have to type this resulting state. For this, 
      we apply the substitution lemma
      (Lemma \ref{lem:substitution}) to the typing judgement for $e_j$, using
      $\changetype{\Gamma_1}{r.f}{\objecttype C F}$
      as the $\Gamma$ of the lemma and $r.f$ as the $r$ of
      the lemma. The first case of the lemma
      corresponds to (1) being \textsc{T-Null}, \textsc{T-Label} or
      \textsc{T-Name}; the second one corresponds to (1) being
      \textsc{T-Ref} or \textsc{T-Chan}. In both cases, the resulting
      judgement is:
      $$\judgement\envref{\changetype{\Gamma}{r.f}{\objecttype C
          F}}{r.f} > e_j\subs{v}{x} : T_j <
      \envref{\changetype{\Gamma_1}{r.f}{\objecttype{C}{F_j}}}{r.f}
      /$$
      Indeed, the difference between $\Gamma$ and $\Gamma_1$ depends
      on (1) in the same way as the lemma's result.
      From this and \typedsess{F_j}{C}{S_j} we can now apply
      \textsc{T-Return} and get:
      $$\judgement\envref{\changetype{\Gamma}{r.f}{\objecttype C
          F}}{r.f} > \return{e_j\subs{v}{x}}{} : T <
      \envref{\changetype{\Gamma_1}{r.f}{S_j}}{r} /$$
      where $T$ is the same as in the initial derivation. We then
      conclude, using the heap typing that was
      provided by the opening lemma, with \textsc{T-State}.
    \item $h(r).f$ is an access point name $n$. Then $\Gamma(r.f)$
      must come, in the derivation of $\typedheapchan\Theta\Gamma h$, from
      \textsc{T-Name},
      which implies that $n$ is declared, that $m_j$ is either
      $\acceptterm$ or $\requestterm$, and that $T_j \supt \sem\Sigma$
      where $\Sigma$ is either the declared type or its dual depending
      on which one $m_j$ is. All this implies that
      the state does a
      $n[c^p]$ transition where $c$ is fresh and $p$ depends, again,
      on $m_j$, and that $\labtrans{\Theta}{n[c^p]}{\Theta, c^p :
        \Sigma}$. The resulting state is typed using $\Gamma'' =
      \Gamma, c^p : \sem\Sigma$ and \textsc{T-Chan}.
    \item $h(r).f$ is a channel endpoint $c^p$. Then
      \typedheapchan\Theta\Gamma h implies that $c^p\in\dom(\Theta)$
      and $S \supt \sem{\Theta(c^p)}$. Hence $m_j$ is either
      $\sendterm$ or $\rcvterm$. We distinguish the two cases.
      In the first case,  the fact that $S$ contains $\sendterm$
      implies that $\Theta(c^p)$ is either of the form
      $\send{T''_j}.\Sigma$ with $T'_j\subt T''_j$ or $\chanchoice{l :
      \Sigma_l}_{l\in E}$ and then $T'_j = \{v\}$ and $v\in E$.
      If $v$ is not an
      object identifier, then the state does a $c^p\send{v}$
      transition. We can see that in both cases (send and select),
      $\Theta$ is able to follow that transition and evolves in such a
      way that \typedheapchan{\Theta'}{\Gamma'}{h} holds: the session
      type of $c^p$ is advanced and if $v$ was a channel it is removed
      from the environment, which corresponds to the difference
      between $\Gamma$ and $\Gamma'$, thus it suffices to change the
      instance of \textsc{T-Hempty} at the root of the derivation
      leading to (a) to get this new typing. Then the new state is
      typed using \textsc{T-Null} and \textsc{T-State}.
      If $v$ is an object identifier, then (1) is \textsc{T-Ref} and thus
      $v\in\dom(\Gamma)$, which implies (using (a)) that $v$ is a root of
      $h$, so the state does a $c^p\send{h\downarrow v}$ transition.
      We use the splitting lemma (Lemma \ref{lem:heapsplitting}) to
      see that $\Theta$ is able to follow this transition and yields a
      $\Theta'$ such that we have
      \typedheapchan{\Theta'}{\Gamma'}{h\uparrow v}. We can then again
      conclude using \textsc{T-Null} and \textsc{T-State}.

      In the case where $m_j$ is $\rcvterm$, the state can
      straightforwardly do a transition, which will be a receive on
      channel $c^p$,  however the transition label
      is not completely determined by the original state as we do not
      know what will be received. So we have to prove type
      preservation in all cases where the transition label $\lambda$ is such that
      $\labtrans\Theta\lambda{\Theta'}$ for some $\Theta'$. If
      $\lambda$ is of the form $c^p\rcv{v'}$, then this hypothesis
      tells us that $\Theta(c^p)$ is either of the form
      $\rcv{T_0}.\Sigma$, and then $v'$ must be
      a literal value of type $T_0$ or a channel endpoint which gets
      added to the environment with a type smaller that $T_0$, or of the form
      $\chanbranch{l:\Sigma_l}_{l\in E}$, and then $v'\in E$. In the
      first case we must have $T_0\subt T_j$, thus the resulting
      expression, which is $v'$,
      can be typed using the appropriate literal value rule, or \textsc{T-Chan},
      and subsumption. In the second one, $T_j = \linkthis$ so that $T
      = \linktype{}f$; the resulting expression can be typed using
      \textsc{T-Label} and \textsc{T-VarS}. As for the new initial
      environment, it is obtained, as in the case of send, by replacing
      the instance of \textsc{T-Hempty} at the top of the derivation
      for (a) with one using $\Theta'$ instead of $\Theta$, so that $v'$ gets
      added to the initial environment if it is a channel and that
      the session type of $r.f$ is correctly advanced, meaning, in the
      case of a branch, that it is advanced to the particular session
      corresponding to $v'$, the variant type being reconstituted in
      the final environment by \textsc{T-VarS}.
      Finally, if $\lambda$ is of the form $c^p\rcv{h'}$, then we have
      $\Theta' = \Theta + \Theta''$ with
      $\typedheapchan{\Theta''}{o : T_j}{h'}$, where $o$ is the only
      root of $h'$. The merging lemma (Lemma \ref{lem:heapmerging})
      gives us a typing for the new heap and, as in the other cases, advancing the
      session type of $c^p$ yields a session type change
      in $r.f$, corresponding to the difference between $\Gamma$ and
      $\Gamma'$. We conclude using \textsc{T-Ref} and \textsc{T-State}.\qedhere
    \end{itemize}
  \end{itemize}
\end{proof}

\noindent The following two lemmas will allow us to deduce from this theorem the
proof of subject reduction for configurations.

\begin{lemma}
\label{lem:subjectcongruence}
If $\typedconf{\Theta}{s}$ and $s\equiv s'$, then $\typedconf{\Theta}{s'}$.
\end{lemma}
\begin{proof}
By induction on the derivation of $s\equiv s'$.
\end{proof}
\begin{lemma}\label{lem:confreduction}
  If $s\reduces s'$, then either:
  \begin{enumerate}
  \item
      $s \equiv (\nu\vec c)\,(\thread{h}{r}{e}\parcomp s'')$,\\
      $s' \equiv (\nu\vec c)\,(\thread{h'}{r'}{e'}\parcomp s'')$\\
      and $\labtrans{\state{h}{r}{e}}{\tau}{\state{h'}{r'}{e'}}$, or
  \item
      $s \equiv (\nu\vec
      c)\,(\thread{h_1}{r_1}{e_1}\parcomp\thread{h_2}{r_2}{e_2}\parcomp s'')$,\\
      $s' \equiv (\nu\vec c)(\nu
      d)\,(\thread{h_1}{r_1}{e'_1}\parcomp\thread{h_2}{r_2}{e'_2}\parcomp s'')$,
      \\
      $\labtrans{\state{h_1}{r_1}{e_1}}{n[d^+]}{\state{h_1}{r_1}{e'_1}}
      \quad \text{and}\quad
      \labtrans{\state{h_2}{r_2}{e_2}}{n[d^-]}{\state{h_2}{r_2}{e'_2}}$, or
  \item
      $s \equiv (\nu\vec
      c)\,(\thread{h_1}{r_1}{e_1}\parcomp\thread{h_2}{r_2}{e_2}\parcomp s'')$,\\
      $s' \equiv (\nu\vec
      c)\,(\thread{h_1}{r_1}{e'_1}\parcomp\thread{h_2}{r_2}{e'_2}\parcomp
      s'')$, \\
      $\labtrans{\state{h_1}{r_1}{e_1}}{c^p\send{v}}{\state{h_1}{r_1}{e'_1}}
      \quad \text{and} \quad
      \labtrans{\state{h_2}{r_2}{e_2}}{c^p\rcv{v}}{\state{h_2}{r_2}{e'_2}}$, or
  \item
      $s \equiv (\nu\vec
      c)\,(\thread{h_1}{r_1}{e_1}\parcomp\thread{h_2}{r_2}{e_2}\parcomp s'')$,\\
      $s' \equiv (\nu\vec c)\,
      (\thread{h'_1}{r_1}{e'_1}\parcomp\thread{h'_2}{r_2}{e'_2}\parcomp s'')$,
      \\
      $\labtrans{\state{h_1}{r_1}{e_1}}{c^p\send{h'}}{\state{h'_1}{r_1}{e'_1}}
      \quad \text{and}\quad
      \labtrans{\state{h_2}{r_2}{e_2}}{c^p\rcv{\phi(h')}}{\state{h'_2}{r_2}{e'_2}}$\\
      with $h' = h_1\downarrow o$, $h'_1 = h_1\uparrow o$, and $h'_2 =
      h_2 + \phi(h')$, or
  \item
      $s \equiv (\nu\vec c)\,(\thread{h}{r}{e}\parcomp s'')$,\\
      $s' \equiv (\nu \vec c)\,(\thread{h}{r}{e'} \parcomp
      \thread{o = \hentry{C}{\vec f =
          \overrightarrow\nullterm}}{o}{e''\subs{\nullterm}{x}}\parcomp s'')$\\
      and 
      $\labtrans{\state{h}{r}{e}}{\methcal{C}{m}{}}{\state{h}{r}{e'}}$,
      where $C.\fieldsterm = \vec f$, $o$ is fresh
      and $\method{}{m}{x}{e''} \in C$.
  \end{enumerate}
\end{lemma}
\begin{proof}
This is nothing more than a reformulation of the reduction rules in
terms of labelled transitions: the derivation for $s\reduces s'$ can
contain any number of instances of \textsc{R-Par}, \textsc{R-Str} or
\textsc{R-NewChan} but must have one of the other rules at the top. It
is straightforward to see that depending on that top rule we are in
one of the five cases listed: (1) for any of the single-thread rules in Figure
\ref{fig:reduction}, 
(2) for \textsc{R-Init}, (3) for
\textsc{R-ComBase}, (4) for \textsc{R-ComObj}, and (5) for \textsc{R-Spawn}.
\end{proof}

We can now prove Theorem \ref{thm:subjectreduction}.
\begin{proof}[(Theorem \ref{thm:subjectreduction})]
  Because of Lemma \ref{lem:subjectcongruence} we only need to look at
  the different cases described in Lemma \ref{lem:confreduction}.

  In cases (1) and (5), the initial derivation is as follows:
  \begin{center}\footnotesize
    \AxiomC{\judgment \Theta_1;\Gamma > \state{h}{r}{e} : T < \envref{\Gamma'}{r''} /}
    \rulename{T-Thread}
    \UnaryInfC{\typedconf{\Theta_1}{\state{h}{r}{e}}}
    \AxiomC{\typedconf{\Theta_2}{s''}}
    \rulename{T-Par}
    \BinaryInfC{\typedconf{\Theta_1+\Theta_2}{\state{h}{r}{e}\parcomp
        s''}}
    \rulename{T-NewChan}
    \UnaryInfC{$\vdash s$}
    \DisplayProof
  \end{center}
  In case (1), Theorem \ref{thm:subjectreduction} gives us
  \judgment\Theta_1;\Gamma'' > \state{h'}{r'}{e'} : T <
  \envref{\Gamma'}{r''} /; from there
  the final derivation is the same.
  
  In case (5), the theorem gives us the same result, but the final
  derivation is more complicated as there is one more parallel
  component. The $\methcal{C}{m}{}$ transition
  tells us that $e$ must be of the form
  $\mathcal{E}[\spawn{\methcal{C}{m}{v}}]$. From Lemma
  \ref{lem:typablesubterms}, this implies that the subexpression
  $\spawn{\methcal{C}{m}{v}}$ is typable, which must be a consequence of
  \textsc{T-Spawn}, implying that $m$ appears in the initial session
  type $S$ of $C$ with a $\nulltype$ argument type. As, by hypothesis,
  the declaration of class $C$ is well-typed, this implies (from
  \textsc{T-Class}) \judgement\envref{x : \nulltype, \this :
    \objecttype C {\overrightarrow\nulltype\,\vec f}}{\this} > e'' : T <
  \envref{\this : \objecttype C F}{\this} /. We apply the substitution
  lemma (\ref{lem:substitution}) to this judgement to replace $\this$
  with $o$ and $x$ with $\nullterm$, and we build the heap typing
  \typedheapchan\emptyset{o :\objecttype C
    {\overrightarrow\nulltype\,\vec f}}{o = \hentry{C}{\vec f =
      \overrightarrow\nullterm}} from \textsc{T-Hempty} and
  \textsc{T-Hadd}. This gives a typing for the new thread, with an empty
  $\Theta$, using \textsc{T-State} and \textsc{T-Thread} and we can
  conclude with \textsc{T-Par}.
  
  In cases (2), (3), and (4), the initial derivation is:
  \begin{center}\footnotesize
    \AxiomC{\judgment \Theta_1;\Gamma_1 > \state{h_1}{r_1}{e_1} : T_1
      < \envref{\Gamma'_1}{r'_1} /}
    \UnaryInfC{\typedconf{\Theta_1}{\state{h_1}{r_1}{e_1}}}
    \AxiomC{\judgment \Theta_2;\Gamma_2 > \state{h_2}{r_2}{e_2} : T_2
      < \envref{\Gamma'_2}{r'_2} /}
    \UnaryInfC{\typedconf{\Theta_2}{\state{h_2}{r_2}{e_2}}}
    \rulename{T-Par}
    \BinaryInfC{\typedconf{\Theta_1+\Theta_2}{\state{h_1}{r_1}{e_1}\parcomp
      \state{h_2}{r_2}{e_2}}}
    \AxiomC{\typedconf{\Theta}{s''}}
    \rulename{T-Par}
    \BinaryInfC{\typedconf{\Theta_1+\Theta_2+\Theta}{\state{h_1}{r_1}{e_1}\parcomp
      \state{h_2}{r_2}{e_2}\parcomp s''}}
    \rulename{T-NewChan}
    \UnaryInfC{\typedconf~ s}
    \DisplayProof
  \end{center}
  Furthermore, we can deduce from the transition labels that the
  expressions in the two topmost premises are of the form
  $\mathcal{E}_1[\methcal{f_1}{m_1}{v_1}]$ and
    $\mathcal{E}_2[\methcal{f_2}{m_2}{v_2}]$ 
  with $h_1(r_1).f_1$ and $h_2(r_2).f_2$ being, in case (2), $n$,
  and in cases (3) and (4),
  respectively $c^p$ and $c^{\overline p}$. These two topmost premises
  must come from \textsc{T-State}, which implies
  \typedheapchan{\Theta_1}{\Gamma_1}{h_1} and
  \typedheapchan{\Theta_2}{\Gamma_2}{h_2}, from which we deduce, in
  case (2), that $n$ is a
  declared access point name and in cases (3) and (4) that
  $\sem{\Theta_1(c^p)} \subt \Gamma_1(r_1.f_1)$ and
  $\sem{\Theta_2(c^{\overline p})} \subt \Gamma_2(r_2.f_2)$.
  We use Theorem \ref{thm:subjectreduction} on these two topmost
  premises and distinguish cases.

  In case (2), $\Theta_1$ and
  $\Theta_2$ make transitions which introduce two dual types for $d^+$
  and $d^-$, which are fresh so that the disjoint unions are still
  possible, and we just need to add an additional step of
  \textsc{T-NewChan} before the last one.

  In cases (3) and (4), we first remark that because \textsc{T-NewChan}
  in the derivation leads to an empty environment, $c$
  must be one of the channels in $(\nu \vec c)$ and we must have
  $\Theta_1(c^p) = \Sigma$ and $\Theta_2(c^{\overline p}) =
  {\overline\Sigma}$ for some $\Sigma$.
  Then we use Lemma \ref{lem:typablesubterms} to get a typing
  judgement for the method call subexpression on the sending side
  (thread 1). This judgement has $\Gamma_1$ as an initial typing
  environment and comes from \textsc{T-Call}; as we have
  $\Sigma\subt\Gamma_1(r_1.f_1)$, this
  implies
  that $\Sigma$ is either of the form $\send{T}.\Sigma'$ with $v$ (in
  case (3)) or $o$ (in case (4)) of type $T$, or (only in case (3)) of the form
  $\chanchoice{l : \Sigma_l}_{l\in E}$ with $v\in E$. The simplest
  case is (3): then this typing
  information, together with the duality of the two endpoint types,
  shows that $\Theta_2$ follows the transition with the new type of
  $c^{\overline p}$ still dual to the new type of $c^p$. In the case
  where $v$ is a channel endpoint, its typing goes from $\Theta_1$ to
  $\Theta_2$ but stays the same, so that it is unchanged in the sum
  environment yielded by \textsc{T-Par}. Thus we can
  still apply \textsc{T-NewChan}.

  Case (4) is similar but, additionally, a renaming function is
  applied to the transmitted heap. We use Lemma \ref{lem:heaprenaming}
  to see that the type of its only root, which is all we need, stays
  the same, so that again $\Theta_2$ can follow the transition. We
  also have that a whole part of the channel environment can go from
  $\Theta_1$ to $\Theta_2$ but the effect is the same as with just one
  channel: it does not affect the sum environment resulting from
  \textsc{T-Par}. So again we can still apply \textsc{T-NewChan}.
\end{proof}


\subsection{Type safety}
We now have the following safety result, ensuring not only
race-freedom (no two sends or receives in parallel on the same
endpoint of a channel) but also that the communication is successful.
\begin{theorem}[No Communication Errors]

  Let 
\[
s \equiv (\nu \vec c) (s'\parcomp
  \thread{h}{r}{\mathcal{E}[r.f.m(v)]}\parcomp\thread{h'}{r'}{\mathcal{E}[r'.f'.m'(v')]})
\]
  and suppose that $\,\vdash s$ holds.
  If $h(r).f = c^p$ and $h'(r').f' = c^q$ then:
  \begin{enumerate}
  \item $q=\overline p$,
  \item channel $c$ does not occur in $s'$, and
  \item there exists $s''$ such that $s \vreduces s''$.
  \end{enumerate}
\end{theorem}
As the statement is true in particular when $s'$ is empty, it implies
that communication between the two threads is possible.
\begin{proof}
  This is an essentially straightforward consequence of $\,\vdash s$.
  The typing derivation is similar to the one shown for cases 2/3/4 in
  Theorem \ref{thm:subjectreduction} above; the two top premises must be
  consequences of \textsc{T-State} and the heap typing necessary to
  apply this rule implies, respectively,
  $\Gamma_1(r.f) \supt \sem{\Theta_1(c^p)}$ and $\Gamma_2(r'.f')
  \supt \sem{\Theta_2(c^q)}$.
  Because of the disjoint unions in \textsc{T-Par},
  $c^p\in\dom(\Theta_1)$ and $c^q\in\dom(\Theta_2)$ immediately
  imply (1) and (2); (3) is then a consequence of the duality
  constraint imposed by \textsc{T-NewChan}: looking at the translations of
  dual channel types, and because the method call subexpressions must be
  typed by \textsc{T-Call}, if $m$ is
  $\sendterm$ then $m'$ must be $\rcvterm$ and vice-versa.
\end{proof}

This theorem, together with the progress aspect of Theorem
\ref{thm:threadsr}, restricts the set of blocked configurations to the
following: if $\,\vdash s$ and $s\not\reduces$, then all parallel
components in $s$ are either terminated (reduced to values), unmatched
$\acceptterm$s or $\requestterm$s, or method calls on \emph{pairwise
distinct} channels --- this last case corresponding to a deadlock.

\subsection{Conformance}

We now have the technical material necessary to prove Theorem
\ref{thm:conf} (conformance). Note that we do not formally extend this result
to the distributed setting, as stating a similar property in that case
would require more complex definitions describing, among other things,
how call traces are moved around
between threads; however we can see informally that, because objects
keep their content and session type when transmitted, all necessary
information is kept such that we still have a conformance property.

\begin{proof}
We first prove, by strong induction on $n$, a slightly different result,
namely the following: for each $i$ there is $\Gamma_i$
such that $\judgment \Gamma_i > \state{h_i}{r_i}{e_i} : T <
\envref{\Gamma'}{r'}/$ and $\tr_i$ is consistent with
$\Gamma_i$.

We suppose that this property is true for any reduction sequence of
length $n$ or less whose initial state satisfies the hypotheses and
prove that it is true also for length $n+1$. The base case $n=1$ is trivial.

If the $n$th reduction step
$\state{h_n}{r_n}{e_n}\reduces\state{h_{n+1}}{r_{n+1}}{e_{n+1}}$ does not originate
from \textsc{R-Return}, we use the induction hypothesis on the
beginning of the sequence; we refer to the cases in the proof
of Theorem \ref{thm:threadsr} to show that the $\Gamma_{n+1}$ it allows to construct
from $\Gamma_n$ indeed is consistent with $\tr_{n+1}$. Because we are only
interested in $\Gamma_{n+1}$ and not $\Gamma'$, in most cases we can use
Lemmas~\ref{lem:typablesubterms} and~\ref{lem:replacement} to
ignore any context $\mathcal{E}$ and proceed as if the reduction is
exactly an instance of its original rule.

If the rule is \textsc{R-Seq}, \textsc{R-Switch} or \textsc{R-Swap}
then $\tr_{n+1} = \tr_n$.

If the rule is \textsc{R-Seq} or \textsc{R-Switch} then the proof of
Theorem~\ref{thm:threadsr} shows that we can choose
$\Gamma_{n+1}=\Gamma_n$, so there is nothing more to prove.

If the rule is \textsc{R-Swap} then the proof of
Theorem~\ref{thm:threadsr} indicates that $\Gamma_{n+1}$
(called $\Gamma''$ in subject reduction) can
be defined using Lemma~\ref{lem:heapchange} from the $\Gamma'''$ (called
$\Gamma_1$ in subject reduction) obtained after typing $v'$, the value
that gets swapped into the field. First of all note that most objects,
notably all those which
are not $v'$ and not in a field of $r$, have the same type and
position in the heap in $\Gamma_{n+1}$ as they have in $\Gamma$. For
all them the result is straightforward: we only concentrate on those
objects that move or change type. Depending on the nature of $T$ and
$T'$ (object, link, or base type), there may be one or two of them.
Recall that neither type can be $\linkthis$ as else the
expression would not be typable. We distinguish cases separately for
$T$ and $T'$, knowing that any combination is possible (except both
linking to the same field). Cases for $T'$:
\begin{itemize}
\item If $T'$ is an object type (thus $v'$ is an object name $o'$),
  then $\Gamma_{n+1}(r.f) = \Gamma_n(o')$ (the rule used for $v'$ is
  \textsc{T-Ref)}. We also have
  $\tr_{n+1}(h_{n+1}(r.f)) = \tr_{n+1}(o') = \tr_n(o')$, so
  $\tr_{n+1}$ is indeed consistent with $\Gamma_{n+1}$ with respect to
  reference $r.f$.
\item If $T'$ is $\linktype{}{f'}$, the rule used for $v'$ is
  \textsc{T-VarS}, and $\Gamma_{n+1}(r.f') = \choice{v'}{S_{v'}}{}$.
  We have $\Gamma_{n+1}(r.f) = \linktype{}{f'}$ and
  $h_{n+1}(r.f) = v'$, hence the actual session type of $r.f'$ in
  $h_{n+1}$ according to $\Gamma_{n+1}$ is $S_{v'}$. Thus consistency is
  preserved for $r.f'$.
\end{itemize}
Cases for $T$ (corresponding respectively to cases 1 and 3 of
Lemma~\ref{lem:heapchange}):
\begin{itemize}
\item If $T$ is an object type (thus $h_n(r.f)$ is an object name $o$),
  then $\Gamma_{n+1}$ contains a new entry for $o$, with type
  $\Gamma_n(r.f)$. Consistency for this new entry comes from
  consistency for $r.f$ at the previous step.
\item If $T$ is $\linktype{}{f''}$, then $\Gamma_n(r.f') =
  \choice{l}{S_l}{l\in E}$ and $h_n(r.f) = l_0$ is in $E$. Thus the
  actual session type of $r.f''$ in $h_n$ according to $\Gamma_n$ is
  $S_{l_0}$. Lemma \ref{lem:heapchange} also gives us
  $\Gamma_{n+1}(r.f') = S_{l_0}$, hence the actual session type of
  $r.f''$ has not changed, and consistency is preserved.
\end{itemize}

\noindent If the rule is \textsc{R-New} then the proof of
Theorem~\ref{thm:threadsr} shows that a suitable
$\Gamma_{n+1}$ is of the form
$\Gamma_n, o:C.\sessterm$ where $o$ is the fresh
object name introduced by the reduction.
Definition~\ref{def:calltraceextension} states that
$\tr_{n+1}$ extends $\tr_n$ by assigning an empty call trace to $o$;
clearly $\tr_{n+1}$ is consistent with $\Gamma_{n+1}$.

If the rule is \textsc{R-Call} then the proof of
Theorem~\ref{thm:threadsr} shows that a suitable $\Gamma_{n+1}$ is
$\Gamma_n$ with the type of $r.f$ replaced by a type which is not a
session. So there is no consistency requirement in $\Gamma_{n+1}$
for $r.f$, and every other reference is given the same call trace by
$\tr_{n+1}$ as by $\tr_n$. Therefore $\tr_{n+1}$ is consistent with
$\Gamma_{n+1}$. 

Now if the $n$th step originates from \textsc{R-Return}, we reason
slightly differently. We know by hypothesis that $r_1$ is a prefix of
$r_{n+1}$. Furthermore, since the $n$th step is \textsc{R-Return}, 
$r_n$ is of the form $r_{n+1}.f$. Reduction rules can only alter the
current object by removing or adding one single field reference at
once, therefore there must be a previous reduction step in the
sequence, say the $i$th, that last went from $r_{n+1}$ to $r_{n+1}.f$.
That is, we chose $i$ such that $r_{n+1}.f$ is a prefix of all $r_j$
for $j$ between $i+1$ and $n$ and that $r_i = r_{n+1}$. That step
must originate from \textsc{R-Call} as it is the only rule which
adds a field specification to the current object. Thus, it is of the form
$\state{h_i}{r_{n+1}}{\mathcal{E}(\methcal{f}{m}{v'})}\vreduces\state{h_{i+1}}{r_{n+1}.f}{\mathcal{E}(\return{e}{})}$,
where $e$ is the method body of $m$ with the parameter substituted.
Then it is straightforward to see that the whole
reduction sequence from $i+1$ to $n$ consists of reductions of $e$
inside the context $\mathcal{E}(\return{[\_]}{})$.

We first use the induction hypothesis on the first part of the
reduction ($1$ to $i$) so as to get judgments up to
\judgment\Gamma_i >
\state{h_i}{r_{n+1}}{\mathcal{E}(\methcal{f}{m}{v'})} : T <
\envref{\Gamma'}{r'} /.
We then use Lemma~\ref{lem:typablesubterms} to get
\judgment\Gamma_{i} > \state{h_{i}}{r_{n+1}}{\methcal{f}{m}{v'}} : T' <
\envref{\Gamma''}{r''} / and note that this judgment must come from
\textsc{T-Call}, which implies that $r''=r_{n+1}$, that 
$\Gamma_i(r_{n+1}.f)$ is of the form
$\branch{\methsign{m}{\ldots}{T'}}{S, \ldots}{}$ and that
$\Gamma''(r_{n+1}.f) = S$. Furthermore, $T'$ is either a base type if
$S$ is a branch or $\linktype{}{f}$ if it is a variant.
We know that $\tr_i$ is consistent with
$\Gamma_i$, therefore we have
$\labtransstar{\classterm(o).\sessterm}{\tr_i(o)}
{\labtrans{\branch{\methsign{m}{\ldots}{T'}}{S,\ldots}{}}{m}{S}}$.

We now
use the induction hypothesis
again on the reduction sequence from $i$ to $n$ for this particular call
subexpression, recalling that $i$ has been defined such that the
hypothesis on the current object
is indeed satisfied by this sequence.
We can also use Lemma \ref{lem:replacement} at each step in order to lift the
judgements thus obtained to the whole expression. To summarise, this
means that for any $j$ between $i+1$ and $n$ we have:
$e_j=\mathcal{E}(\return{e'_j}{})$ for some $e'_j$, \judgment\Gamma_j
> \state{h_j}{r_j}{\return{e'_j}{}} : T' < \envref{\Gamma''}{r_{n+1}} / and
\judgment\Gamma_j > \state{h_j}{r_j}{e_j} : T < \envref{\Gamma'}{r'}
/, and that $tr_j$ is consistent with $\Gamma_j$.

For the last reduction step, \textsc{R-Return},
the proof of Theorem
\ref{thm:threadsr} tells us that we can choose a
$\Gamma_{n+1}$ which is identical to $\Gamma_n$ except for the type of
$r_{n+1}.f$, and as the call trace for other references is not
modified, consistency is preserved for them. For $r_{n+1}.f$ we have
to look back at the initial subexpression on step $i$.
First note that \textsc{R-Swap} can only act on a field of
the current object, therefore since $r_{n+1}.f$ is a prefix of the current object
during the whole subsequence, its content cannot change and is the
same object $o$ throughout. Similarly, there is no other
\textsc{R-Call} or \textsc{R-Return} acting on that particular object,
hence $\tr_n(o) = \tr_{i+1}(o) = \tr_i(o)m$. We saw above that this
call trace leads the initial session of $o$ to $S$. 
Then the judgement for the final subexpression, at
step $n+1$, is of the form $\judgment\Gamma_{n+1} >
\state{h_{n+1}}{r_{n+1}}{v} : T' < \envref{\Gamma''}{r_{n+1}} /$.
There are two cases, as in the proof of Theorem
\ref{thm:threadsr}.
If $T'$ is a base type then $S$ is a branch and
it is possible to
decide that $\Gamma_{n+1}(r_{n+1}.f)$ is equal to $S$. In that case
the call trace either does not change or has a label appended, but as
$S$ is a branch it can do a transition to itself with any label,
therefore $\tr_{n+1}(o)$ is consistent with $\Gamma_{n+1}(r_{n+1}.f)$
in both cases. If $T'$ is $\linktype{}{f}$, then $v$ is a label, $S$
is a variant $\choice{l}{S_l}{l\in E}$ and $\Gamma_{n+1}(r_{n+1}.f)$
can be chosen equal to $S_v$. We have $\tr_{n+1}(o) = \tr_n(o)v$ and
$\labtrans{S}{v}{S_v}$, so consistency is preserved.

This completes the inductive proof that for every step $i$ in the
reduction sequence there is $\Gamma_i$ such that $\judgment \Gamma_i >
\state{h_i}{r_i}{e_i} : T < \envref{\Gamma'}{r'}/$ and $\tr_i$ is
consistent with $\Gamma_i$. This fact obviously implies that $\tr_i$
is valid for all the objects which have a session type in $\Gamma_i$;
we now argue that it is also the case for the other objects, namely
those which either are not at all in $\Gamma_i$ or do not have a session
type. We know by hypothesis that it is the case for $\tr_1$ and show
by a very simple induction that it cannot change from $i$ to $i+1$.
The $i$th step can only change the call trace for an object $o$ if it
originates from \textsc{R-Call} or \textsc{R-Return} concerning that object.
\textsc{R-Call} can only occur if the reducible part of the expression
is indeed a method call on a field which contains $o$, and that is
only typable if $\Gamma_i$ contains a session type for that field
which is a branch containing the method, and thus allows the
appropriate transition: therefore validity of the call trace for $o$ is
preserved in that case. \textsc{R-Return} on the other
hand can only occur if the reducible part of the expression is a
return and if the current object is (the address of) $o$, and we saw
that in that case the $\Gamma_{i+1}$ constructed in our proof contains
a session type for $o$, so this case is covered by the consistency result.
\end{proof}



\section{Type Checking Algorithm}
\label{sec:algorithm}
\comment{SG: updated 28.3.2011} 
\begin{figure}
\begin{tabbing}
$\mathcal{W}(C) = \alga{C}{C.\sessterm}{C.\fieldsterm}{\emptyset}$\\
~~~~\= if for every $\annotmethod{F}{F'}{T'}{m}{T~x}{e} \in C$\\
\>~~~~\=$F'\not=\langle\_\rangle$ and $\algb{e}{F}{x : T} = \triplet{T'}{F'}{\_}$
\end{tabbing}
\begin{tabbing}
$\alga{C}{S}{F}{\Delta} = \Delta$ if $(F, S)\in \Delta$\\
$\alga{C}{\mu X.S}{F}{\Delta} =
  \alga{C}{S\subs{\mu X.S}{X}}{F}{\Delta\cup\{(F, \mu X.S)\}}$\\[\algskip]
$\alga{C}{\branch{\methsign{m_i}{U_i}{T_i}}{S_i}{1\leqslant
    i\leqslant n}}{F}{\Delta_0} = \Delta_n$\\
~~~~\= where for $i = 1$ to $n$\\
\>~~~~\= let $\triplet{T'_i}{F'_i}{\_} = \algb{e_i}{F}{x_i : U_i}$
where $\method{}{m_i}{x_i}{e_i} \in C$\\
\>\> if $T'_i\subt T_i$ then let $\Delta_i =
\alga{C}{S_i}{F'_i}{\Delta_{i-1}}$\\
\>\> else \= if $T'_i$ is an enumeration $E$ and $T_i = \linkthis$\\
\>\>\> then let $\Delta_i = \alga{C}{S_i}{\choice{l}{F'_i}{l\in
    E}}{\Delta_{i-1}}$\\
\>\> else \> if $T'_i = \linkthis$ and $T_i$ is an enumeration $E$ and
$F'_i = \choice{l}{F_l}{l\in E'}$ and $E'\subseteq E$\\
\>\>\> then let $\Delta_i = \alga{C}{S_i}{\bigvee_{l\in E'}F_l}{\Delta_{i-1}}$\\[\algskip]
$\alga{C}{\choice{l}{S_l}{l \in E}}{\choice{l}{F_l}{l \in E'}}{\Delta_0} = \Delta_n$\\
\> where $E' = \{l_1\ldots l_n\} \subset E$ and for $i = 1$ to $n$,
$\Delta_i = \alga{C}{S_{l_i}}{F_{l_i}}{\Delta_{i-1}}$
\end{tabbing}
%
%
%
%


\textbf{Combining variants}
\begin{tabbing}
  $\iset{T_i\,f_i}{i\in I} \vee \iset{T'_i\,f_i}{i\in I} = \iset{(T_i\vee T'_i)\,f_i}{i\in I}$\\[\algskip]
$\vfield{l}{F_l}{l \in I} \vee \vfield{l}{F'_l}{l\in J} = \vfield{l}{F''_l}{l \in I \cup J}$\\
~~~~\= where $F''_l = F_l\vee F'_l$ if $l\in I\cap J$, $F_l$ if $l\not\in J$, $F'_l$ if $l\not\in I$
\end{tabbing}
\caption{Typechecking: algorithms $\mathcal{W}$ and $\mathcal{A}$.}
\label{fig:fig-alga}
\end{figure}

 
\begin{figure}
\begin{tabbing}
$\algb{\nullterm}{F}{V} = \triplet{\nulltype}{F}{V}$\\[\algskip]
$\algb{n}{F}{V} = \triplet{\sem{\access{n.\mkterm{protocol}}}}{F}{V}$\\[\algskip]
$\algb{x}{F}{y : T} = \triplet{T}{F}{V}$ if $x=y$, where  $V = \emptyset$ if $T$ is linear or $y : T$ otherwise\\[\algskip]
$\algb{\swap{f}{e}}{F}{V} =  \triplet{U}{F''}{V'}$\\
~~~~\= where $\triplet{T}{F'}{V'} = \algb{e}{F}{V}$
and $F'(f)$ is not a variant and\\
\> if $T=\linkthis$ then $F' = \vfield{l}{F_l}{l\in E}$ and $U = (\bigvee_{l\in E}F_l)(f)$ and $F'' =
\changetype{(\bigvee_{l\in E}F_l)}{f}{E}$\\
\> else $U = F'(f)$ and $F'' = \changetype{F'}{f}{T}$\\[\algskip]
$\algb{l}{F}{V} = \triplet{\linkthis}{\langle l : F \rangle}{V}$\\[\algskip]
$\algb{\new{C'}}{F}{V} = \triplet{C'.\sessterm}{F}{V}$\\[\algskip]
$\algb{\methcal{f}{m_j}{e}}{F}{V} =  \triplet{T}{F''}{V'}$\\
\> where  $\triplet{T'}{F'}{V'} = \algb{e}{F}{V}$ and\\
\>~~~~\= if $T'=\linkthis$ then $F' = \vfield{l}{F_l}{l\in E}$ and
$(\bigvee_{l\in
  E}F_l)(f)=\branch{\methsign{m_i}{T'_i}{T_i}}{S_i}{i\in I}$ and $j
\in I$\\
\>\>~~~~\=and $T'_j$ is an enumeration $E'$ and $E\subset E'$ and $F'' =
\changetype{(\bigvee_{l\in E}F_l)}{f}{S'_j}$ and\\
\>\>\>$T = \linktype{}{f}$ if $T_j = \linkthis$,
  $T = T_j$ otherwise\\
\>\> else $F'(f) = \branch{\methsign{m_i}{T'_i}{T_i}}{S_i}{i\in I}$
and $j \in I$ and $T'\subt T_j'$ and\\
\>\>\>$F'' = \changetype{F'}{f}{S'_j}$ and $T = \linktype{}{f}$ if $T_j = \linkthis$,
  $T = T_j$ otherwise\\[\algskip]
$\algb{\selfcal{m}{e}}{F}{V} = \triplet{T'}{F'''}{V'}$\\
\> where  $\triplet{T}{F'}{V'} = \algb{e}{F}{V}$ and $\annotmethod{F''}{F''}{T'}{m}{T''~x}{e}\in C$ and\\
\>~~~~\=if $T=\linkthis$ then $F' = \vfield{l}{F_l}{l\in E}$ and $T''$
is an enumeration $E'$ and\\
\>\>~~~~\= $E\subset E'$ and $\bigvee_{l\in
  E}F_l \subt F''$\\
\>\>else $T\subt T''$ and $F'\subt F''$\\[\algskip]
$\algb{\switch{e}{l}{e_l}{l \in E}}{F}{V} =
\triplet{T}{\bigvee_{l\in E}F''_l}{V''}$ \\
\> where $\triplet{U}{F'}{V'} = \algb{e}{F}{V}$ and\\
\>\> if $U = E'$ then $E'\subseteq E$ and $\forall l \in
E'. \triplet{T}{F''_l}{V''}=\algb{e_l}{F'}{V'}$\\
\>\> else if $U = \linkthis$ then $F' = \vfield{m}{G_m}{m\in E'}$ and
$E'\subset E$ and\\
\>\>~~~~\=$\forall l \in
E'. \triplet{T}{F''_l}{V''}=\algb{e_l}{\bigvee_{m\in E'}G_m}{V'}$\\
\>\> else if $U = \linktype{}{f}$ then\\
\>\>~~~~\=$F'(f) = \vfield{l}{S_l}{l\in E'}$ and $E'\subseteq E$ and
$\forall l \in E'. \triplet{T}{F''_l}{V''} =
\algb{e_l}{\changetype{F'}{f}{S_l}}{V'}$\\[\algskip]
$\algb{\while{e}{e'}}{F}{V} = \triplet{\nulltype}{F''}{V'}$\\
\> where $\triplet{U}{F'}{V'} = \algb{e}{F}{V}$ and \\
\>\> if $U = E'$ then $E'\subseteq \{\true,\false\}$ and
$\algb{e'}{F'}{V'}=\triplet{\nulltype}{F}{V}$ and $F''=F'$\\
\>\> else if $U = \linkthis$ then $F' = \vfield{l}{F_l}{l\in E}$ and
$E\subset \{\true,\false\}$ and\\
\>\>~~~~\=$\algb{e'}{\bigvee_{l\in
    E}F_l}{V'}=\triplet{\nulltype}{F}{V}$ and $F''=\bigvee_{l\in E}F_l$\\
\>\> else if $U = \linktype{}{f}$ then $F'(f) = \langle \true:S_\true, \false:S_\false \rangle$ and\\
\>\>\>$\algb{e'}{\changetype{F'}{f}{S_\true}}{V'}=\triplet{\nulltype}{F}{V}$
and $F''=\changetype{F'}{f}{S_\false}$\\[\algskip]
$\algb{\seq{e}{e'}}{F}{V} = \algb{e'}{F''}{V'}$\\
\> where $\triplet{T}{F'}{V'} = \algb{e}{F}{V}$ and $T \neq
\linktype{}{\_}$ and\\
\>~~~~\= if $T=\linkthis$ then $F' = \vfield{l}{F_l}{l\in E}$ and
$F''=\bigvee_{l\in E}F_l$\\
\>\> else $F''=F'$ \\[\algskip]
%
%
%
$\algb{\spawn{C'.m(e)}}{F}{V} = \triplet{\nulltype}{F'}{V'}$\\
\> where
$(\nulltype,F',V') = \algb{e}{F}{V}$ and $\methsign{m}{\nulltype}{\nulltype}\in C'.\sessterm$
\end{tabbing}

\caption{Typechecking: algorithm $\mathcal{B}$.}
\label{fig:fig-algb}
\end{figure}

 
\begin{figure}
\begin{lstlisting}
class C {
  session { linkthis m(int): < FALSE: SCf, TRUE: SCt > }
  where SCf = ..., SCt = ...
  ...
} 

class D {
  session { linkthis a(int): < FALSE: SDf, TRUE: SDt >,
            linkthis b(int): < FALSE: SDf, TRUE: SDt >,
            { FALSE, TRUE } c(int): SD1,
            { FALSE, TRUE, UNKNOWN } d(int): SD2 }
  where SDf = ..., SDt = ..., SD1 = ..., SD2 = ...

  f;

  a(x) { // Not allowed, because return type is link f
    f <-> new C();
    f.m(x); }

  aa(x) { // Allowed, because body type is linkthis
    f <-> new C();
    switch (f.m(x)) {
      case FALSE: FALSE;
      case TRUE: TRUE; } }

  b(x) { // Allowed, by creating a uniform variant
    even(x); }

  bb(x) { // Allowed, because body type is linkthis
    switch (even(x)) {
      case FALSE: FALSE;
      case TRUE: TRUE; } }

  c(x) { // Allowed, by taking a join of field typings
    f <-> new C();
    switch (f.m(x)) {
      case FALSE: FALSE;
      case TRUE: TRUE; } }

  cc(x) { // Allowed, by taking a join of equal field typings
    switch (even(x)) {
      case FALSE: FALSE;
      case TRUE: TRUE; } }

  d(x) { // Allowed, because of subtyping between enumerations
    even(x); }
}
\end{lstlisting}
\caption{Example for type checking.}
\label{fig:algexample}
\end{figure}


This section introduces a type checking algorithm, sound and complete
with respect to the type system in Section~\ref{sec:distributed}, and
describes a prototype implementation of a programming language based
on the ideas of the paper.

\subsection{The Algorithm}

Figures~\ref{fig:fig-alga} and~\ref{fig:fig-algb} define a type
checking algorithm for the distributed language, including the
sequential extensions from Section~\ref{sec:seq-extensions}. The
algorithm is applied to each component of a distributed system, and in
order to ensure type safety of the complete system there must be some
separate mechanism to check that each access point $n$ is given the same
type everywhere. A program is type checked by calling algorithm
$\mathcal{W}$ on each class definition and checking that no call
generates an error. The definition of algorithm $\mathcal{W}$ follows
the typing rule \mkTrule{Class} in Figure~\ref{fig:self-calls}. It
calls algorithm $\mathcal{A}$ to check the relation $\typedsess F C S$
and algorithm $\mathcal{B}$ to type check the bodies of the methods
that have $\reqterm/\ensterm$ annotations. Algorithm $\mathcal{A}$
also calls algorithm $\mathcal{B}$ to typecheck the bodies of the
methods that appear in the session type.

In both $\mathcal{A}$ and $\mathcal{B}$ there are several ``if'' and
``where'' clauses; they should be interpreted as conditions which, if
not satisfied, cause termination with a typing error.

Because of the coinductive definition of $\typedsess F C S$, algorithm
$\mathcal{A}$ uses a set $\Delta$ of assumed relationships between
field typings $F$ and session types $S$. If there is no error then the
algorithm returns $\Delta$, but at the top level we are only
interested in success or failure, not in the returned value.

Algorithm $\mathcal{B}$ checks the typing judgement for expressions,
defined in Figure~\ref{fig:typingexpr}, specialized to the top-level
form $\judgment \envref{\this:\objecttype{C}{F}, V}{\this} > e:T <
\envref{\this:\objecttype{C}{F'},V'}{\this} /$ as explained in
Section~\ref{sec:typingexpressions}.  The definition of $\mathcal{B}$
follows the typing rules (Figure~\ref{fig:typingexpr}) except for one
point: \textsc{T-VarF} means that the rules are not syntax-directed,
as any expression with type $E$ can also be given type
$\linkthis$. For this reason, clause $l$ of $\mathcal{B}$ produces
type $\linkthis$ and a variant field typing with the single label $l$.
More general variant field typings are produced when typing
$\switchterm$ expressions, as the $\vee$ operator is used to combine
the field typings arising from the branches. This is the typical
situation when typing the body of a method whose return type is
$\linkthis$: the body contains a $\switchterm$ whose branches return
different labels with different associated field typings.

It is possible, however, that giving type $\linkthis$ to $l$ is
incorrect. It might turn out that the expression needs to have an
enumerated type $E$, for example in order to be passed as a method
parameter or returned as a method result of type $E$. An expression
that has been inappropiately typed with $\linkthis$ can, in general, be
associated with any variant field typing, for example if it contains a
$\switchterm$ whose branches yield different field typings. In this
case, the algorithm uses $\vee$ to combine the branches of the variant
field typing into a single field typing; the join is always over all
of the labels in the variant. This happens in several places in
algorithm $\mathcal{B}$, indicated by conditions of the form ``if $T =
\linkthis$'', and in the final ``else'' branch of the third clause of
algorithm $\mathcal{A}$.

The algorithm for checking subtyping is not described here but is
similar to the one defined for channel session types by Gay and Hole
\citeN{GaySJ:substp}.  
%
We write $S \vee S'$ for the least upper bound of
$S$ and $S'$ with respect to subtyping.
It is defined by taking the intersection of
sets of methods and the least upper bound of their
continuations. Details of a similar definition (greatest lower bound
of channel session types) can be found in the work of Mezzina \citeN{MezzinaLG:typs}.

The type checking algorithm is modular in the sense that to check class
$C$ we only need to know the session types of other classes, not their
method definitions.

We have not yet investigated type inference, but there are two ways in
which it might be beneficial. One would be to infer the
$\reqterm/\ensterm$ annotations. The other would be to support some
form of polymorphism over field typings, along the lines that if
method $m$ does not use field $f$ then it should be callable
independently of the type of $f$. This might reduce the need to type
check the definition of $m$ every time it occurs in the session type.

\subsection{Examples of Type Checking}

Figure~\ref{fig:algexample} defines classes \lstinline|C| and
\lstinline|D|. In class \lstinline|C|, only the outer layer of the
session type is of interest; the example uses an object of class
\lstinline|C| but does not need the definition of method
\lstinline|m|. Class \lstinline|D|, as well as the outer layer of the
session type, contains a field \lstinline|f| and one or two candidate
definitions for each of the methods \lstinline|a|, \lstinline|b|,
\lstinline|c| and \lstinline|d|. The definitions of \lstinline|a| and
\lstinline|aa| are alternatives for the method \lstinline|a| specified
in the session type, and so on.

The definition of \lstinline|a| is not typable because the type of the
returned expression is $\linktype{}{f}$. Allowing this would let the
caller of \lstinline|a| have access to field \lstinline|f|. Instead,
the result of \lstinline|f.m(x)| must be analyzed with a
\lstinline|switch|, as in the definition of \lstinline|aa|, which is
typable. The \lstinline|linkthis| type required by the signature of
\lstinline|a| is introduced by the enumeration labels
\lstinline|FALSE| and \lstinline|TRUE| in the branches of the
\lstinline|switch|. A compiler could insert \lstinline|switch|es of
this kind automatically, allowing the definition of \lstinline|a| as
syntactic sugar.

The remaining method definitions are all typable and illustrate
different features of the type system and the algorithm. In the
definition of \lstinline|b|, the method \lstinline|even| is supposed
to be the obvious function for testing parity of an integer, returning
\lstinline|TRUE| or \lstinline|FALSE|. This definition is typable even
though the body of \lstinline|b| does not introduce a
\lstinline|linkthis| type, because algorithm $\mathcal{A}$ constructs
a variant field typing over \lstinline|{TRUE,FALSE}| in which both
options are the same. This is seen in the first \emph{else} clause of
$\mathcal{A}$. The definition of \lstinline|bb| achieves the same
effect by using the labels \lstinline|FALSE| and \lstinline|TRUE| to
introduce the type \lstinline|linkthis|. Each label corresponds to a
partial variant field typing, and checking the \lstinline|switch|
combines them by means of the $\vee$ operator. Because the field
\lstinline|f| is not involved in the method body, the field typing is
the same in both options of the variant.

Method \lstinline|c| has the same definition as \lstinline|a|, but
this time the signature in the session type specifies a simple
enumeration as the return type. This is allowed, by using the $\vee$
operator to construct the join of the field typings, in the second
\emph{else} clause of $\mathcal{A}$. This means that when the
algorithm proceeds to type check method definitions in the session
type \lstinline|SD1|, the type of \lstinline|f| is taken to be the
join of \lstinline|SCf| and \lstinline|SCt|. Whether or not this loss
of information causes a problem will depend on the particular
definitions of those types, which we have not shown. Method
\lstinline|cc| is handled in the same way, but this time there is no
loss of information because the types being joined are identical; this
in turn is because \lstinline|f| is not involved in the method body.

Finally, method \lstinline|d| illustrates straightforward subtyping
between enumerations, defined as set inclusion.

\subsection{Correctness of the Algorithm}

The following sequence of results outlines the proof of soundness and
completeness of the algorithm. The detailed proofs are routine and are
omitted.

\begin{theorem}
  Algorithm $\mathcal{A}$ always terminates, either with an error (and
  then the function $\mathcal{A}$ is undefined) or with a result.
\end{theorem}
\begin{proof}
  Similar to proofs about algorithms for coinductively-defined
  subtyping relations \cite[Chapter 16]{PierceBC:typpl}.
\end{proof}

\begin{lemma}
  If $\judgment \envref{\this:\objecttype{C}{F}, V}{\this} >
  e:\linkthis < \envref{\this:\objecttype{C}{F'},V'}{\this} /$ then
  for some $E$ and $\{F_l\}_{l\in E}$, $F' = \vfield{l}{F_l}{l\in
    E}$ and $\judgment \envref{\this:\objecttype{C}{F}, V}{\this} >
  e:E < \envref{\this:\objecttype{C}{\bigvee_{l\in E}F_l},V'}{\this} /$.
\end{lemma}
\begin{proof}
By induction on the typing derivation.
\end{proof}

\begin{lemma}
  \begin{sloppypar}
    If $\algb{e}{F}{V} = (T,F',V')$ then
    $\judgment \envref{\this:\objecttype{C}{F}, V}{\this} > e:T <
    \envref{\this:\objecttype{C}{F'},V'}{\this} /$.
  \end{sloppypar}
\end{lemma}
\begin{proof}
By induction on the structure of $e$.
\end{proof}

\begin{lemma}
\label{lem:algb}
If $\judgment \envref{\this:\objecttype{C}{F}, V}{\this} > e:T <
\envref{\this:\objecttype{C}{F'},V'}{\this} /$ and $\algb{e}{F}{V}
= (T',F'',V'')$  then $V''\subt V'$ and either
\begin{enumerate}
\item $T'\subt T$ and $F''\subt F'$, or
\item $T = \linkthis$, $T'$ is an enumeration $E$, $F' =
  \vfield{l}{F_l}{l\in E'}$, $E\subset E'$ and $\forall l\in
  E.~F''\subt F_l$, or 
\item $T$ is an enumeration $E$, $T' = \linkthis$, $F'' =
  \vfield{l}{F_l}{l\in E'}$, $E'\subset E$ and $\forall l\in
  E'.~F_l\subt F'$.
\end{enumerate}
\end{lemma}
\begin{proof}
By induction on the typing derivation.
\end{proof}

\begin{theorem}
\label{thm:alga-complete}
If $\typedsess{F}{C}{S}$ then $\alga{C}{S}{F}{\emptyset}$ is defined.
\end{theorem}
\begin{proof} 
  Consider the execution of $\alga{C}{S}{F}{\emptyset}$. It terminates
  and has various calls of the form $\alga{C}{S'}{F'}{\Delta}$,
  including the top-level call. We prove the following statement, by
  induction on the number of recursive calls in the execution of
  $\alga{C}{S'}{F'}{\Delta}$: if $\Delta\subset
  \typedsess{\bullet}{C}{\bullet}$ and $\typedsess{F'}{C}{S'}$ then
  $\alga{C}{S'}{F'}{\Delta}$ is defined and
  $\alga{C}{S'}{F'}{\Delta}\subset \typedsess{\bullet}{C}{\bullet}$.
\end{proof}

\begin{lemma}
\label{lem:alga2}
If $\alga{C}{S}{F}{\Delta}$ is defined then for any $\Delta'$,
$\alga{C}{S}{F}{\Delta\cup\Delta'} = \alga{C}{S}{F}{\Delta}\cup\Delta'$.
\end{lemma}
\begin{proof}
Similar to the proof of Theorem~\ref{thm:alga-complete}, by induction on the
recursive calls within a given top-level call.
\end{proof}

\begin{lemma}
\label{lem:alga3}
Suppose $\alga{C}{\rectype{X}{S_0}}{F_0}{\emptyset}$ is defined and
$(F_0,\rectype{X}{S_0})\not\in\Delta$. Then for all $S$ and $F$, if
$\alga{C}{S}{F}{\Delta\cup\{(F_0,\rectype{X}{S_0})\}}$ is defined then
$\alga{C}{S}{F}{\Delta}$ is defined.
\end{lemma}
\begin{proof}
Similar to the proof of Theorem~\ref{thm:alga-complete}, by induction on the
recursive calls within a given top-level call.
\end{proof}

\begin{lemma}
\label{lem:alga4}
If $\alga{C}{\rectype{X}{S}}{F}{\emptyset}$ is defined then
$\alga{C}{S\subs{\rectype{X}{S}}{X}}{F}{\emptyset}$ is defined.
\end{lemma}
\begin{proof}
By the definition of $\mathcal{A}$,
$\alga{C}{\rectype{X}{S}}{F}{\emptyset} =
\alga{C}{S\subs{\rectype{X}{S}}{X}}{F}{\{(F,\rectype{X}{S})\}}$, which
is therefore defined. By Lemma~\ref{lem:alga3}, $\alga{C}{S\subs{\rectype{X}{S}}{X}}{F}{\emptyset}$ is defined.
\end{proof}

\begin{corollary}
\label{cor:unfold}
If $\alga{C}{S}{F}{\emptyset}$ is defined then
$\alga{C}{\mkterm{unfold}(S)}{F}{\emptyset}$ is defined.
\end{corollary}

\begin{theorem}
\label{thm:alga-sound}
If $\alga{C}{S}{F}{\emptyset}$ is defined then $\typedsess{F}{C}{S}$.
\end{theorem}
\begin{proof}
By Corollary~\ref{cor:unfold} and the fact that 
$\typedsess{F}{C}{S}$ is defined in terms of the unfolded structure of
session types, it is sufficient to consider the case in which $S$ is
guarded. 

Similarly to the proof of Theorem~\ref{thm:alga-complete}, consider the
recursive calls in the execution of
$\alga{C}{S_0}{F_0}{\emptyset}$. We show that the following relation is a
$C$-consistency relation:
\[
\mathcal{R} = \{ (F,S) \mid \text{$\alga{C}{S}{F'}{\Delta}$ is called
  for some $\Delta$ and $F'$ with $F\subt F'$} \}. 
\]
This is easily checked, using the three cases of Lemma~\ref{lem:algb}
to correspond to the three cases in the third clause of the definition
of $\mathcal{A}$. 
\end{proof}

\subsection{Implementation}
\label{sec:implementation}

The ideas introduced in this paper can be used to extend a conventional Java compiler, by including
\lstinline{@session} annotations in classes and in method parameters,
as well as \lstinline{@req} and \lstinline{@ens} annotations for
recursive methods (cf.\ Section~\ref{subsec:self-calls}). The extension only concerns type checking; there is no need to touch the back-end of the
compiler.

To keep in line with the expectations of Java programmers, annotations
follow the first style in Figure~\ref{fig:file},
page~\pageref{fig:file}. Also, the type system is nominal (cf.\
Section~\ref{sec:nominal}); label sets (cf.\ Figure~\ref{fig:syntax})
are explicitly introduced via 
\lstinline{enum} declarations.
The concepts contained in our core language can then be extended towards the whole of Java.
In particular:
\begin{itemize}
\item The \lstinline|while| loop technique described in
  Section~\ref{subsec:while-loops} can be extended to handle
  \lstinline{for} and \lstinline{do-while} loops.
\item The same idea can be used to type the various goto instructions
  present in Java: exceptions, \lstinline{break}, \lstinline{continue}
  and \lstinline{return}, labelled versions included.
\item All control flow instructions (including \lstinline{if-then}, not
  discussed in the paper) can be used with conventional or with
  session-related \lstinline{boolean}/\lstinline{enum} values.
\item Classes not featuring a \lstinline{@session} annotation are
  considered shared rather than linear. Their objects can be treated very
  much like the \lstinline{null} value (cf.\
  Section~\ref{subsec:shared-types}). We do not allow a shared class
  to contain a linear field, even though it is perfectly acceptable
  for a method of a shared class to have a linear parameter.
\item The same technique used for ``top-level'' classes can be used for
  inner, nested, local (defined within methods) and anonymous classes.
\item In order to mention overloaded methods in \lstinline{@session}
  annotations, alias names for these methods can be introduced via extra
  annotations.
\item Static fields are always shared.
\item Class inheritance (cf.\ Section~\ref{sec:nominal}) can be supported.
\end{itemize}

We have used the Polyglot~\cite{NystromN:polecf} system for an initial prototype extension of Java, but a more thorough design and implementation are left for future work.
%


\section{Related Work}
\label{sec:related}

\comment{SG: updated 11.5.11}

There is a large amount of related work, originating from several
different approaches. Our discussion of related work is organised
according to these approaches.

\paragraph{Previous work on session types for object-oriented
  languages.}
Dezani-Ciancaglini, Yoshida et al.\
\cite{CapecchiS:amasmo,Dezani-CiancagliniM:bousto,Dezani-CiancagliniMetal:ost,Dezani-CiancagliniM:disool}
have taken an approach in which a class define sessions \emph{instead
  of} methods. Invoking a session on an object creates a channel which
is used for communication between two blocks of code: the body of the
session, and a \emph{co-body} defined by the invoker of the session. A
session is therefore a generalization of a method, in which there can
be an extended dialogue between caller and callee instead of a single
exchange of parameters and result. The structure of this dialogue is
defined by a session type. This approach proposes a new paradigm for
concurrent object-oriented programming, and as far as we know it has
not yet been implemented. In contrast, our approach maintains the
standard execution model of method calls.

The SJ (Session Java) language, developed by Hu~\cite{HuR:sesbdp}, is
a less radical extension of the object-oriented paradigm. Channels,
described by session types, are essentially the same as those in the
original work based on process calculus. Program code is located in
methods, as usual, and can create channels, communicate on them, and
pass them as messages. SJ has a well-developed implementation and has
been applied to a range of situations. However, SJ has one notable
restriction: a channel cannot be stored in a field of an object. This
means that a channel, once created, must be either completely used, or
else delegated (sent along another channel), within the same
method. It is possible for a channel to be passed as a parameter to
another method, but it is not possible for a session to be split into
methods that can be called separately, each implementing part of the
session type of a channel that is stored in a field.  A distinctive
feature of our work is that we can store a channel in a field of an
object and allow several methods to use it. This is illustrated in
Figures~\ref{fig:remotefileserver1b}~ and~\ref{fig:remote-file}.


Hu et al.~\citeN{HuRetal:tsesj,Ng.ea:safeSJ} have also extended SJ to
support event-driven programming, with a session type discipline to
ensure safe event handling and progress. We have not considered
event-driven programming in our setting.

Campos and Vasconcelos
\cite{CamposJ:linear-and-shared,campos.vasconcelos:channels-as-objects}
developed MOOL, a simple class-based object-oriented language, to
study object usage and access. The novelties are that class usage
types are attached to class definitions, and the communication
mechanism is based on method call instead of being channel-based. The
latter feature is the main difference with respect to our work.

\paragraph{Non-uniform concurrent objects/active objects.}
Another related line of research, started by
Nierstrasz~\citeN{nierstrasz:regular-types}, aimed at describing the
behaviour of non-uniform \emph{active} objects in concurrent systems,
whose behaviour (including the set of available methods) may change
dynamically. He defined subtyping for active objects, but did not
formally define a language semantics or a type system. The topic has
been continued, in the context of process calculi, by
several
authors \cite{CairesL:spabtc,CairesVieira:ct,colaco.ea:safety-analysis,colaco.ea:set-constraint,najm.ea:liveness,najm.ea:infinite-types,puntigam:stateinference,puntigam.peter:deadlock,ravara.vasconcelos:typco}.
%
The work by Caires~\citeN{CairesL:spabtc} is the most relevant work;
it uses an approach based on spatial types to give very fine-grained
control of resources, and Milit\~ao~\citeN{MilitaoF:yak} has
implemented a Java prototype based on this idea.
Damiani et al.~\citeN{DamianiF:typssa} define a concurrent Java-like
language incorporating inheritance and subtyping and equipped with a
type-and-effect system, in which method availability is made dependent
on the state of objects.

The distinctive feature of our approach to non-uniform objects, in
comparison with all of the above work, is that we allow an object's
abstract state to depend on the result of a method call. This gives a
very nice integration with the branching structure of channel session
types, and with subtyping.

Specifically related to the notion of subtyping between session types,
the work of Rossie~\citeN{RossieJ:LOEs} is worth mentioning. He
proposes a type-based approach to ensure that both component objects
and their clients have compatible protocols. The typing discipline
specifies not only how to use the component's methods, but also the
notifications it sends to its clients. Rossie calls this enhanced
specification a Logical Observable Entity (LOE), which is a
finite-state machine equipped with a subtyping notion. An LOE is a
high-level description of an object, specifying which transitions
(method executions) change its state, providing for each state both
the available methods and notifications to be sent to the
clients. LOEs support behavioural subtyping, in its afferent aspects
(how clients may affect the LOE) --- a subtype must allow at least the
traces of its supertype, and in its efferent aspects (how a LOE
processing a method request has effects on clients) --- the subtype
must not send more notifications than the supertype. This behavioural
subtyping notion on finite-state machines, which is in its spirit very
similar to the one of session types --- "more offers, less requests",
is defined as a simulation relation. Rossie shows that this relation
ensures safe substitutability.

\paragraph{Typestate.}
Based on the fact that method availability depends on an object's
internal state (the situation identified by Nierstrasz, as mentioned
above), Strom and Yemini \citeyear{typestates} proposed
\emph{typestate}. The concept consists of
identifying the possible states of an object and defining pre- and
post-conditions that specify in which state an object should be so
that a given method would be available, and in which state the
method execution would leave the object.

\emph{Vault}~\cite{DeLineR:enfhpl,FahndrichM:adofpl} 
follows the typestate approach. It uses linear types to control
aliasing, and uses the \emph{adoption and focus}
mechanism~\cite{FahndrichM:adofpl} to re-introduce aliasing in limited
situations. \emph{Fugue} \cite{DeLineR:fugpcy,FahndrichM:typo} extends
similar ideas to an object-oriented language, and uses explicit pre-
and post-conditions.

Bierhoff and Aldrich~\citeN{BierhoffAldrich:lightweightobject} also work on a typestate
approach in an object-oriented language, defining a sound modular
automated static protocol-checking setting. They define a state and
method refinement relation achieving a behavioural subtyping
relation. The work is extended with access permissions, that combine
typestate with aliasing information about
objects~\cite{BierhoffAldrich:aliasedobjects}, and with concurrency,
via the atomic block synchronization primitive used in transactional
memory systems~\cite{Beckman.etal:atomicblockstypestates}. Like us,
they allow the typestate to depend on the result of a method
call. \emph{Plural} is a prototype language implementation that
embodies this approach, providing automated static analysis in a
concurrent object-oriented language~\cite{BierhoffAldrich:plural}. To
evaluate their approach they annotated and verified several standard
Java APIs~\cite{Bierhoff.etal:praticalapichecking}.

Milit\~ao et al.~\citeN{MilitaoFetal:aliasing-control} develop a new
aliasing control mechanism, finer and more expressive than previous
proposals, based on defining object views according to specific access
constraints. The discipline is implemented in a type system combining
views and a typestate approach, checking user defined aliasing
patterns.

\emph{Sing\#}~\cite{FahndrichM:lansfr} is an extension of C\#
which has been used to implement Singularity, an operating system
based on message-passing. It incorporates session types to specify
protocols for communication channels, and introduces typestate-like
\emph{contracts}.
Bono et al.~\cite{BonoV:typcmp} have formalised a core calculus based
on \emph{Sing\#} and proved type safety. A technical point is that
\emph{Sing\#} uses a single construct $\mathsf{switch~receive}$ to
combine receiving an enumeration value and doing a case-analysis,
whereas our system allows a $\switchterm$ on an enumeration value to
be separated from the method call that produces it.

Aldrich et al.~\cite{AldrichJ:tysop} have proposed
\emph{typestate-oriented programming}. The aim is to integrate
typestate into language design from the beginning, instead of adding
typestate constraints to an existing language. Their prototype
language is called Plaid. Instead of class definitions, a program
consists of state definitions; each state has methods which cause
transitions to other states when they are called. Like classes, states
are organised into an inheritance hierarchy. The specifications of
state transitions caused by methods are similar to the pre- and
post-conditions of Plural. Aliasing is managed by a system of access
permissions \cite{BierhoffAldrich:aliasedobjects}. More recent
work~\cite{GarciaR:foutop,WolffR:grat} combines gradual typing and
typestate, to integrate static and dynamic typestate checking.

Session types and typestate are related approaches, but there are
stylistic and technical differences. With respect to the former,
session types are like labelled transition systems or finite-state
automata, capturing the behaviour of an object. When developing an
application, one may start from session types and then implement the
classes. Typestates take each transition of a session type and attach
it to a method as pre- and post-conditions. Because typestate systems
allow pre- and post-conditions to be specified arbitrarily, the
possible sequences of method calls are less explicit. With respect to
technical differences, the main ones are: (a) session types unify
types and typestates in a single class type as a global behavioural
specification; (b) our subtyping relation is structural, while the
typestates refinement relation is nominal; (c) \emph{Plural} uses a
software transactional model as concurrency control mechanism (thus,
shared memory), which is lighter and easier than locks, but one has to
mark atomic blocks in the code, whereas our communication-centric
model (using channels) is simpler and allows us to use the same type
abstraction (session types) instead of a new programming construct;
moreover, channel-based communication also allows us to specify the
client-server communication protocol as the channel session type, and
to implement it modularly, in several methods which may even be in
different classes;
(d) typestate approaches allow flexible aliasing control, whereas our
approach uses only linear objects (to add better alias/access control
is simple and an orthogonal issue).


\paragraph{Affine types.}  Tov and Pucella \cite{TovJ:praat}
have developed Alms, a language in the style of OCaml with an affine
type system as a generalisation of linear typing. Alms is a
general-purpose programming language, in which the affine type system
provides an infrastructure suitable for defining a variety of
type-based resource control patterns including alias control, session
types and typestate. It has been implemented, and type safety has been
proved for a formal calculus. Representing a particular approach to
typestate, such as our specifications of allowed sequences of method
calls, would require an encoding; in contrast, our language aims to
provide a convenient high-level programming style.

\paragraph{Static verification of protocols.}
\emph{Cyclone}~\cite{GrossmanD:regbmm} and
\emph{CQual}~\cite{foster.etal:flow-sensitive-type-qualifiers} are
systems based on the C programming language that allow protocols to be
statically enforced by a compiler. \emph{Cyclone} adds many benefits
to C, but its support for protocols is limited to enforcing locking of
resources. Between acquiring and releasing a lock, there are no
restrictions on how a thread may use a resource. In contrast, our
system uses types both to enforce locking of objects (via linearity)
and to enforce the correct sequence of method calls.
\emph{CQual} expects users to annotate programs with type qualifiers;
its type system, simpler and less expressive than the above, provides
for type inference.

\paragraph{Unique ownership of objects.}
In order to demonstrate the key idea of modularizing session
implementations by integrating session-typed channels and non-uniform
objects, we have taken the simplest possible approach to ownership
control: strict linearity of non-uniform objects. This idea goes back
at least to the work of Baker~\citeN{BakerHG:usevlo} and has been
applied many times.  However, linearity causes problems of its own:
linear objects cannot be stored in shared data structures, and this
tends to restrict expressivity. There is a large literature on less
extreme techniques for static control of aliasing: Hogg's
\emph{Islands}~\cite{HoggJ:islapo}, Almeida's \emph{balloon
  types}~\cite{AlmeidaPS:baltcs},
Clarke \ea's \emph{ownership types}~\cite{ClarkeDG:owntfa}, F\"{a}hndrich and DeLine's
\emph{adoption and focus}~\cite{FahndrichM:adofpl}, \"{O}stlund \ea's
$\mathsf{Joe}_3$~\cite{Ostlund.etal:ownership-uniqueness-immutability}
among others. In future work we intend to use an off-the-shelf
technique for more sophisticated alias analysis. The property we
need is that when changing the type of an object (by calling a method
on it or by performing a $\switchterm$ or a $\whileterm$ on an
enumeration constant returned from a method call) there must be a
unique reference to it.

\paragraph{Resource usage analysis.}
Igarashi and Kobayashi~\citeN{IgarashiA:resua} define a general
resource usage analysis problem for an extended $\lambda$-calculus,
including a type inference system, that statically checks the order of
resource usage.
Although quite expressive, their system only analyzes the sequence of
method \emph{calls} and does not consider branching on method
\emph{results} as we do.

\paragraph{Analysis of concurrent systems using pi-calculus.}
Some work on static analysis of concurrent systems expressed in
pi-calculus is also relevant, in the sense that it addresses the
question (among others) of whether attempted uses of a resource are
consistent with its state.  Igarashi and Kobayashi have developed a
generic framework~\cite{igarashi.kobayashi:generic-type} including a
verification tool \cite{KobayashiN:typifa} in which to define type
systems for analyzing various behavioural properties including
sequences of resource uses \cite{KobayashiN:resuap}. In some of this
work, types are themselves abstract processes, and therefore in some
situations resemble our session types.  Chaki at
al.~\citeN{chakietal:types-as-models} use CCS to describe properties
of pi-calculus programs, and verify the validity of temporal formulae
via a combination of type-checking and model-checking techniques,
thereby going beyond static analysis.

All of this pi-calculus-based work follows the approach of modelling
systems in a relatively low-level language which is then analyzed. In
contrast, we work directly with the high-level abstractions of session
types and objects.


\section{Conclusion}
\label{sec:conclusion}

\comment{SG: updated 11.5.2011}

We have extended existing work on session types for object-oriented
languages by allowing the implementation of a session to be divided
between several methods which can be called
independently. This supports a modular approach which is absent from
previous work. Technically, it is achieved by integrating session
types for communication channels and a static type system for
non-uniform objects. A session-typed channel is one kind of
non-uniform object, but objects whose fields are non-uniform are also,
in general, non-uniform. Typing guarantees that the sequence of
messages on every channel, and the sequence of method calls on every
non-uniform object, satisfy specifications expressed as session types.

We have formalized the syntax, operational semantics and static type
system of a core distributed class-based object-oriented language
incorporating these ideas. Soundness of the type system is expressed
by type preservation, conformance and correct communication
theorems. The type system includes a form of typestate and uses
simple linear type theory to guarantee unique ownership of non-uniform
objects. It allows the typestate of an object after a
method call to depend on the result of the call, if this is of an
enumerated type, and in this situation, the necessary case-analysis of
the method result does not need to be done immediately after the call. 

We have illustrated our ideas with an example based on a remote file
server, and described a prototype implementation. By incorporating
further standard ideas from the related literature, it should be
straightforward to extend the implementation to a larger and more
practical language.


In the future we intend to work on the following topics.
(1) More
flexible control of aliasing.  The mechanism for controlling aliasing
should be orthogonal to the theory of how operations affect
uniquely-referenced objects. We intend to adapt existing work to relax
our strictly linear control and obtain a more flexible language.
(2) In Section~\ref{sec:nominal} we outlined an adaptation of our
structural type system to a nominal type system as found in
languages such as Java. We would also like to account for Java's
distinction and relationship between classes and interfaces.
%
(3)
Specifications involving several objects.
Multi-party session types~\cite{BonelliE:mulstd,HondaK:mulast} and
conversation types~\cite{CairesVieira:ct} specify
protocols with more than two participants. It would be
interesting to adapt those theories into type systems for more
complex patterns of object usage.



\paragraph{Acknowledgements}

We thank Jonathan Aldrich and Lu\'{\i}s Caires for helpful
discussions. 
Gay was partially supported by the UK EPSRC (EP/E065708/1
``Engineering Foundations of Web Services'', EP/F037368/1 ``Behavioural Types for Object-Oriented Languages'', EP/K034413/1 ``From Data Types to Session Types: A Basis for Concurrency and Distribution'' and EP/L00058X/1 ``Exploiting Parallelism through Type Transformations for Hybrid Manycore Systems''). He
thanks the University of Glasgow for the sabbatical leave during which
part of this research was done.
Gay and Ravara were partially supported by the Security and Quantum
Information Group at Instituto de Telecomunica\c{c}\~{o}es, Portugal.
%
%
Vasconcelos was partially supported by the Large-Scale Informatics
Systems Laboratory, Portugal.
Ravara was partially supported the Portuguese
Fun\-da\-\c{c}\~{a}o para a Ci\^{e}ncia e a Tecnologia 
FCT (SFRH/BSAB/757/2007), and  by the UK EPSRC (EP/F037368/1).
Gesbert was supported by the UK EPSRC (EP/E065708/1) and by the French ANR (project ANR-08-EMER-004 ``CODEX''). All of the authors have received support from COST Action IC1201 ``Behavioural Types for Reliable Large-Scale Software Systems''.
%

\bibliographystyle{simonplain}
\bibliography{main}

\appendix

\section*{Appendix: Proofs of lemmas from Section \ref{sec:techlemmas}}\label{app:proofs}

{
\renewcommand\thetheorem{7.\arabic{theorem}}
\renewcommand\label[1]{}

\inlineproofstrue
\explanatorytextfalse

\begin{lemma}\label{lem:heapproperties}
  Suppose $\typedheapchan\Theta\Gamma h$. Then \emph{(a)} $h$ is
  complete,
  \emph{(b)} $\chans(\Gamma)\subset\dom(\Theta)\setminus \chans(h)$ and
  \emph{(c)} $\objs(\Gamma)\subset\roots(h)$.
\end{lemma}
\ifinlineproofs\begin{proof}
  By induction on the derivation of $\typedheapchan\Theta\Gamma h$.
  The only axiom is \textsc{T-Hempty} for which the properties are true.
  Then \textsc{T-Hide} does not change either $h$ or
  $\dom(\Gamma)$ so it preserves all three properties. The other case is
  \textsc{T-Hadd}. Let $h'$ be the heap in the conclusion. Then
  $\children{h'}(o)$ is the set of $v_i$ which are object
  identifiers. Let $K$ be the set of $v_i$ which are channel endpoints.
  The typing derivation for the sequence
  of swaps in the right premise must include an occurrence of
  \textsc{T-Ref} for each object identifier, and of \textsc{T-Chan}
  for each channel endpoint, 
  each followed by \textsc{T-Swap} and a number of occurrences of
  \textsc{T-Seq}. Looking at these rules, we can see that this
  implies:\begin{enumerate}
  \item $\children{h'}(o)\cup K\subset\dom(\Gamma)$ and
  \item $\dom(\Gamma')\subset(\dom(\Gamma)\setminus
  (\children{h'}(o)\cup K))\cup\{o\}$. (Note that $o$ cannot be one of the
  $v_i$ because it is the current object in the judgement: the premise
  of \textsc{T-Ref} forbids it.)
  \end{enumerate}

  \noindent From (1) and induction hypothesis (c) we get
  $\children{h'}(o)\subset\roots(h)$. We have
  $\roots(h)\subset\dom(h)\subset\dom(h')$ and $o$ is the only new
  object in $h'$, so $h'$ is complete.

  If we project (2) onto just channel endpoints, we get
  $\chans(\Gamma')\subset \chans(\Gamma)\setminus K$. From the
  definition of $h'$, $\chans(h')$ is
  equal to $\chans(h)\cup K$. Hence induction hypothesis (b) yields (b)
  again for $h'$.

  If we project (2)
  onto just object identifiers, we get
  $\objs(\Gamma')\subset(\objs(\Gamma)\setminus\children{h'}(o))\cup\{o\}$.
  From induction hypothesis (a) and the fact that $o\not\in\dom(h)$ we
  get that $o$ is a root in $h'$. Furthermore, all roots of $h$ which
  are not children of $o$ are also roots of $h'$.
  Thus induction
  hypothesis (c) allows us to conclude $\objs(\Gamma')\subset\roots(h')$.
\end{proof}\fi

\begin{lemma}[Rearrangement of typing derivations for expressions]\hfil
  \label{lem:rearrange-expr}
  Suppose we have \judgment \envref\Gamma r > e : T <
  \envref{\Gamma'}{r'} /. Then there exists a typing derivation for
  this judgement in which:
\begin{enumerate}
  \item \textsc{T-Sub} only occurs at the
    very end, just before \textsc{T-Switch} or \textsc{T-SwitchLink}
    as the last rule in the derivation for each of the branches, or 
    just before \textsc{T-Call} as the last rule in the
    derivation for the parameter; 
  \item \textsc{T-SubEnv} only occurs immediately before
    \textsc{T-Sub} in the first three cases and does not occur at all
    in the fourth, \ie \textsc{T-Call}.
\end{enumerate}
\end{lemma}
\ifinlineproofs\begin{proof}
  First note that \textsc{T-Sub} and \textsc{T-SubEnv} commute and
  that any consecutive sequence of occurrences of one of these rules
  can collapse into a single occurrence using transitivity.
  What
  remains to be shown is that these rules can be pushed down in all
  cases but those mentioned in the statement. We enumerate the cases below.
  \begin{itemize}
  \item \textsc{T-Swap}. \textsc{T-Sub} before the premise can be replaced
    with \textsc{T-SubEnv} after the conclusion as $T$ has been
    transferred to the environment. If \textsc{T-SubEnv} was used before the premise, it means the initial derivation looks like:

\smallskip
\AxiomC{$\judgment\envref\Gamma r > e : T < \envref{\Gamma'}{r'}/$}
\rrulename{T-SubEnv}
\UnaryInfC{$\judgment\envref\Gamma r > e : T < \envref{\Gamma''}{r'}/$}
\AxiomC{$\Gamma''(r'.f) = T'$}
\AxiomC{$\ldots$}
\rrulename{T-Swap}
\TrinaryInfC{$\judgment\envref\Gamma r > \swap f e : T' < \envref{\changetype{\Gamma''}{r'.f}{T}}{r'}/$}
\DisplayProof

\smallskip
with $\Gamma' \subt \Gamma''$. Let $T_0 = \Gamma'(r'.f)$; we have $T_0\subt T'$ since $T'=\Gamma''(r'.f)$. Because the type of $r'.f$ moves from the environment to the expression, when we push the subsumption step down, we have to use both \textsc{T-SubEnv} and \textsc{T-Sub}; we can transform the derivation into:

\smallskip
\AxiomC{$\judgment\envref\Gamma r > e : T < \envref{\Gamma'}{r'}/$}
\AxiomC{$\Gamma'(r'.f) = T_0$}
\AxiomC{$\ldots$}
\rrulename{T-Swap}
\TrinaryInfC{$\judgment\envref\Gamma r > \swap f e : T_0 < \envref{\changetype{\Gamma'}{r'.f}{T}}{r'}/$}
\rrulename{T-SubEnv}
\UnaryInfC{$\judgment\envref\Gamma r > \swap f e : T_0 < \envref{\changetype{\Gamma''}{r'.f}{T}}{r'}/$}
\rrulename{T-Sub}
\UnaryInfC{$\judgment\envref\Gamma r > \swap f e : T' < \envref{\changetype{\Gamma''}{r'.f}{T}}{r'}/$}
\DisplayProof

\smallskip
  \item \textsc{T-Call}. If \textsc{T-SubEnv} is used on the premise
    to increase the type of something else than $r'.f$ it can be moved to
    the conclusion. If the type of $r'.f$ is changed, first note that
    the only relevant part is the signature of $m_j$. Suppose the
    subsumption step
    changes it from $\methsign{m_j}{U'_j}{U_j} : S'_j$ to
    $\methsign{m_j}{T'_j}{T_j} : S_j$. For the parameter type we have
    $T'_j\subt U'_j$ so we can
    use \textsc{T-Sub} on the premise to increase the type of $e$ from
    $T'_j$ to $U'_j$ instead. For the session and result types, we
    have two cases:
    \begin{itemize}
    \item if $U_j\subt T_j$ and $S'_j\subt S_j$ it can just be moved
      to a \textsc{T-SubEnv} step on the conclusion.
    \item if $U_j = E$, $T_j = \linkthis$ and $\choice{l}{S'_j}{l\in
        E}\subt S_j$, then the original conclusion of the rule (with
      \textsc{T-SubEnv} on the premise) was:
$$\judgement \envref\Gamma r
      > \methcal{f}{m_j}{e} : \linktype{}f <
      \envref{\changetype{\Gamma'}{r'.f}{S_j}}{r'} /$$
      and the new one with the subsumption step removed is:
$$\judgement \envref\Gamma r
      > \methcal{f}{m_j}{e} : E <
      \envref{\changetype{\Gamma'}{r'.f}{S'_j}}{r'} /.$$
      So in that case the original judgement can be obtained back from
      this new conclusion using \textsc{T-VarS} followed by
      \textsc{T-SubEnv}.
    \end{itemize}
    \item \textsc{T-Seq}. \textsc{T-Sub} on the first premise is
      irrelevant and \textsc{T-SubEnv} on the same premise can be
      removed using Lemma \ref{lem:moreweak}. Subsumption on the
      second premise straightforwardly commutes to the conclusion.
    \item \textsc{T-Switch}. Lemma
      \ref{lem:moreweak} allows us to remove \textsc{T-SubEnv} on the
      first premise. Straightforwardly \textsc{T-Sub} can be removed
      as well as it just makes $E'$ smaller.
    \item \textsc{T-SwitchLink}. \textsc{T-Sub} is irrelevant;
      removing \textsc{T-SubEnv} can only make $E'$ and the initial
      typing environments for the branches smaller and we can use
      Lemma \ref{lem:moreweak}.
    \item \textsc{T-VarF} and \textsc{T-VarS}.
      \textsc{T-Sub} can increase $E$ which
      becomes the indexing set of the variant in the conclusion. By
      definition of subtyping for variants it is possible to increase
      it afterwards using \textsc{T-SubEnv}. \textsc{T-SubEnv}
      straightforwardly commutes.
    \item \textsc{T-Return}. \textsc{T-Sub} straightforwardly
      commutes, as well as the part of \textsc{T-SubEnv} not
      concerning $r'.f$. Subsumption on $\Gamma'(r'.f)$ can be removed
      using Proposition \ref{prop:subfield}.
  \end{itemize}
\end{proof}\fi

\begin{lemma}[Rearrangement of typing derivations for heaps]\hfil
\label{lem:rearrange-heap}
  Suppose $\typedheapchan\Theta\Gamma h$ holds. Let $o$ be an
  arbitrary root of $h$. Then there exists a
  typing derivation for it such that:
  \begin{enumerate}
  \item \textsc{T-Sub} is never used;
  \item \textsc{T-SubEnv} is used at most once, as the last rule leading
    to the right premise of the last occurrence of \textsc{T-Hadd};
  \item every occurrence of \textsc{T-Hide} follows immediately the
    occurrence of \textsc{T-Hadd} concerning the same object
    identifier;
  \item the occurrence of \textsc{T-Hadd} concerning an identifier $o'$
    is always immediately preceded (on the left premise)
    by the occurrences of \textsc{T-Hadd}/\textsc{T-Hide} concerning
    the descendants of $o'$;
  \item the first root added is $o$.
  \end{enumerate}  
\end{lemma}
\ifinlineproofs\begin{proof}
  The first two points are a consequence of Lemma
  \ref{lem:rearrange-expr}: the only expressions which appear in the
  typing derivation are sequences of swaps, not containing any switch
  or method call; furthermore their type is always $\nulltype$, making
  \textsc{T-Sub} at the end irrelevant. What remains to be checked is then
  just that \textsc{T-SubEnv} at the end of the derivation for one
  sequence of swaps can be pushed down to the next occurrence of
  \textsc{T-Hadd} whenever there is one. This is just a matter of
  using Proposition \ref{prop:subfield} in the case of \textsc{T-Hide}
  and Lemma \ref{lem:moreweak} in the case of \textsc{T-Hadd}.

  Note that these points imply in particular that in all applications of
  \textsc{T-Hadd} but the last one, any element in $\dom(\Gamma)$ which
  is not one of the $v_i$ also occurs in $\Gamma'$ with exactly the same type.
  
  For the third point, first notice that the premise of
  \textsc{T-Hide} implies $o$ is a root of $h$ because of Lemma
  \ref{lem:heapproperties}. This implies that the rule immediately above
  \textsc{T-Hide} either is a \textsc{T-Hadd} introducing $o$ or does
  not concern $o$ at all (in particular, $o$ cannot be a $v_i$, otherwise
  it would not be a root in the conclusion). In the second case,
  \textsc{T-Hide} can be pushed upwards.
  
  The fourth and fifth points are a consequence of the remark we made about the
  first two: if $o'$ is not a descendant of $o$ nor vice-versa,
  then the occurrences of
  \textsc{T-Hadd} and \textsc{T-Hide} concerning $o$ and its descendants
  commute with those concerning $o'$ and its descendants as they affect
  completely disjoint parts of the environment. In the case of the last
  occurrence of \textsc{T-Hadd} there may be a subsumption step but it
  is still possible to commute with it by pushing this subsumption step down again.
\end{proof}\fi

\begin{lemma}[Splitting of the heap]\label{lem:heapsplitting}
  Suppose $\typedheapchan\Theta{\Gamma, o : T}h$. Let $\Theta_1 =
  \Theta\setminus\chans(h\downarrow o)$ and let $\Theta_2$ be $\Theta$
  restricted to $\chans(h\downarrow o)$. Then we have:
  $\typedheapchan{\Theta_1}\Gamma{(h\uparrow o)}$ and
  $\typedheapchan{\Theta_2}{o : T}{(h\downarrow o)}$.
\end{lemma}
\ifinlineproofs\begin{proof}
  We know from Lemma \ref{lem:heapproperties} that $o$ is a root in
  $h$. We consider the particular derivation given by Lemma
  \ref{lem:rearrange-heap} where $o$ is the first root added to the
  heap. Now if we look at the conclusion of the last rule
  concerning $o$ (\textsc{T-Hadd} or \textsc{T-Hide} depending whether
  $T$ is a field or session type),
  we know that at this point the heap is $h\downarrow o$,
  and therefore the only object identifier in
  the environment is its only root: $o$.
  Furthermore, this part of the derivation is
  still true if we replace the initial $\Theta$ with $\Theta_2$, with
  the only difference that then the final $\Gamma$ contains no
  channels, and thus is of the form $o : T'$. We also know that the
  type of $o$ is not changed in the rest of the derivation except
  possibly by the subsumption step at the end; therefore
  $T'$ is a subtype of $T$. If they are session types, using
  Proposition \ref{prop:subsess} we
  can change the last occurrence of \textsc{T-Hide} to use $T$ instead
  of $T'$ and get $\typedheapchan{\Theta_2}{o : T}{(h\downarrow o)}$.
  Otherwise, we can add a subsumption step to the derivation for the
  sequence of swaps on the right of \textsc{T-Hadd} to get the same result.

  For the rest of the derivation, we know that $o$ is not used,
  therefore it can be removed from the initial environment without
  affecting the derivation except by the fact that it will not be in
  the final environment either. Furthermore, we know from Lemma
  \ref{lem:heapproperties} that the initial environment minus its only
  object identifier $o$ is included in
  $\Theta\setminus\chans(h\downarrow o) = \Theta_1$. More precisely, the lemma
  gives us inclusion of domains, but because subsumption is not used
  in the first part of the derivation we also know that the types are
  the same. Thus we can replace the first part of the derivation by an
  instance of \textsc{T-Hempty} using $\Theta_1$ and the second part
  is still valid (with all the descendants of $o$ removed from the
  heap), yielding $\typedheapchan{\Theta_1}\Gamma{(h\uparrow
    o)}$ at the bottom.
\end{proof}\fi

\begin{lemma}[Merging of heaps]\label{lem:heapmerging}
  Suppose $\typedheapchan\Theta\Gamma h$ and
  $\typedheapchan{\Theta'}{\Gamma'}{h'}$ with $\dom(h)\cap\dom(h') =
  \emptyset$ and $\dom(\Theta)\cap\dom(\Theta')=\emptyset$. Then we
  have $\typedheapchan{\Theta+\Theta'}{\Gamma+\Gamma'}{h+h'}$.
\end{lemma}
\ifinlineproofs\begin{proof}
  Since $\Theta$ and $\Theta'$ are disjoint, the
  channels in $\Theta'$ cannot appear anywhere in the typing
  derivation for $h$. Thus, it is possible to add $\Theta'$ to every
  typing environment occurring in the derivation for $h$ without
  altering its validity, yielding
  $\typedheapchan{\Theta+\Theta'}{\Gamma+\Theta'}h$. Looking now at
  the derivation for $h'$, since the domains of the heaps are disjoint and
  $\objs(\Gamma)\subset\roots(h)$, none of the identifiers in $\Gamma$
  can appear anywhere in it. Thus we can add $\Gamma$ to every typing
  environment and $h$ to every heap occurring in
  the derivation for $h'$, replacing the \textsc{T-Hempty} at the top
  with the conclusion of the other derivation, which yields the result
  we want.
\end{proof}\fi

\ifexplanatorytext
These two lemmas show, if we apply them repeatedly, that a typing
derivation for a heap can be considered as a set of separate typing
derivations leading to each root of the heap. This will allow us in
particular to
show results for particular cases where a heap has only one root and
generalize them.
\fi


\begin{lemma}\label{lem:heaprenaming}
  Suppose $\typedheapchan\Theta{o : S}h$. Let $\varphi$ be an injective
  function from $\dom(h)$ to $\mathcal{O}$. Then we have
  $\typedheapchan\Theta{\varphi(o) : S}{\varphi(h)}$.
\end{lemma}
\ifinlineproofs\begin{proof}
  Straightforward. Changing the names does not affect the typing derivation
  in any way.
\end{proof}\fi

\begin{lemma}[Opening]\label{lem:opening}
  If $\typedheapchan{\Theta}{\Gamma}{h}$, if $\Gamma(r)$ is a branch
  session type $S$ and if $h(r)$ is an object
  identifier $o$, then
  we know from Lemma \ref{lem:heapproperties} that $h$ contains an
  entry for $o$. Let $C$ be the class of this entry, then
  there exists a field typing $F$ for $C$ such that
  $\typedheapchan\Theta{\changetype{\Gamma}{r}{\objecttype{C}{F}}}{h}$
  and $\typedsess F C S$.
\end{lemma}
\ifinlineproofs\begin{proof}
  We prove this by induction on the depth of $r$. The base case is $r
  = o$. Using Lemmas \ref{lem:heapsplitting} and
  \ref{lem:heapmerging}, we can restrict ourselves to the case where
  $o$ is the only root of $h$. In that case we know that the last rule
  used in the typing derivation for $\typedheapchan\Theta{o : S}h$
  must be \textsc{T-Hide}. The result we want is constituted precisely
  by the premises of that rule.

  For the inductive case, $r$ is of the form $o'.f.\vec f$.
  We consider the case where $o'$ is the only root. The typing
  derivation then ends with \textsc{T-Hadd} and $f$ gets populated in
  the sequence of swaps by some object identifier\footnote{We know
    $o''$ is an object identifier and not a channel
    endpoint because, according to the hypotheses, either it is $o$
    itself or $\vec f$ is nonempty, implying $o''$ has fields.} $o''$.
  Let $r' = o''.\vec f$, and consider what $\Gamma(r')$ can be, knowing
  that in the conclusion $r$ has a branch session type: the only way
  the type can be modified in the sequence of swaps is by subsumption.
  Indeed, \textsc{T-VarS}, the other possibility, introduces a variant
  type. Therefore $\Gamma(r') = S'$ with $S'\subt S$. We can thus
  use the induction hypothesis to replace $\Gamma$ with
  $\changetype\Gamma{r}{\objecttype C F}$ on the left premise, with
  \typedsess F C {S'}. Then just use Proposition \ref{prop:subsess}
  to see that we also have \typedsess F C S and see that the type
  yielded in the conclusion by this new premise is what we want.
\end{proof}\fi

\begin{lemma}[Closing]\label{lem:closing}
  If $\typedheapchan\Theta{\Gamma}{h}$ and $\Gamma(r) = \objecttype{C}{F}$
  and $\typedsess F C S$, then
  $\typedheapchan\Theta{\changetype{\Gamma}{r}{S}}{h}$.
\end{lemma}
\ifinlineproofs\begin{proof}
  Again we prove this by induction on the depth of $r$ and the base
  case is $r = o$. In that case the lemma is nothing more than
  \textsc{T-Hide}. The inductive case is very similar to the above: we
  look at the type of $r'$ (defined as above) in the $\Gamma$ on
  the left premise of
  the last \textsc{T-Hadd}, noticing that the type of $r$ in the
  conclusion can only differ from it by subsumption,
  this time because we know from Lemma
  \ref{lem:heapfield} that \textsc{T-VarF} is never used. Hence the
  original type is $\objecttype C{F'}$ with $F'\subt F$. Proposition
  \ref{prop:subfield} gives us $\typedsess{F'}CS$ and thus we can use
  the induction hypothesis to change the type of $r'$ in this premise,
  which propagates to the type of $r$ in the conclusion.
\end{proof}\fi

\begin{lemma}[modification of the heap]\label{lem:heapchange}
Suppose that we have $\typedheapchan\Theta{\Gamma}{h}$ and
\judgment\envref\Gamma r > v' : T' < \envref{\Gamma'}{r} /,
and that $\Gamma'(r.f) = T$ where $T$ is not a variant. Let $v =
h(r).f$. The modified
heap $\changeval{h}{r.f}{v'}$ can be typed as follows:
\begin{enumerate}
\item if $v$ is an object identifier or a channel endpoint, then:
$$\typedheapchan\Theta{\changetype{\Gamma'}{r.f}{T'}, v : T}
{\changeval{h}{r.f}{v'}}$$
\item if $v$ is not an object or channel and $T$
  is not a link type, then:
$$\typedheapchan\Theta{\changetype{\Gamma'}{r.f}{T'}}
{\changeval{h}{r.f}{v'}}$$
\item if $v = l_0$ and $T = \linktype{}{f'}$,
  then:
  \begin{itemize}
  \item $\Gamma'(r.f') = \choice{l}{S_l}{l\in E}$ for some $E$
    such that $l_0\in E$ and
    some set of branch session types $S_l$. Note that this implies
    $f\neq f'$.
  \item $\typedheapchan\Theta
    {\changetype{\changetype{\Gamma'}{r.f}{T'}}{r.f'}{S_{l_0}}}
    {\changeval{h}{r.f}{v'}}$
  \end{itemize}
\end{enumerate}
\end{lemma}
\ifinlineproofs\begin{proof}
  First of all, note that the hypothesis that $\Gamma'(r.f)$ is
  defined implies that
  $\Gamma'(r)$ is not a variant, hence $T'$ is not $\linkthis$.
  In other words, the judgement cannot be derived from \textsc{T-VarF}.
  Furthermore, the fact that $T$ is not a variant either means that the
  judgement is not derived from an instance of \textsc{T-VarS} referring
  to field $f$. As this rule is the only possibility (beside subsumption) for a
  judgement typing a value to depend on, and modify, the type of a
  field, this implies that $\Gamma(r.f)$ is a subtype of $T$ and also
  that the judgement would still hold with another type for $f$. In
  particular we have
  \judgement\envref{\changetype{\Gamma}{r.f}{\nulltype}}{r} > v' : T' <
  \envref{\changetype{\Gamma'}{r.f}{\nulltype}}{r} /~\textbf{(a)}. We will in the
  following use this judgement (a) rather than the one in the hypothesis.

  We prove the lemma by induction on the depth of $r$, but the inductive case is
  straightforward (just apply the induction hypothesis to the left premise
  of \textsc{T-Hadd}). In the base case, $r$ is an object identifier
  $o$. We use Lemma \ref{lem:heapsplitting} to consider a typing
  derivation for the
  sub-heap $h\downarrow o$. Let $\Gamma_o = o : T_o$ and $\Theta_o$ be the
  environments corresponding to that part of the heap.
  %
  %
  We look at the application of \textsc{T-Hadd} which ends the
  derivation for $\typedheapchan{\Theta_o}{\Gamma_o}{h_o}$. As $T$ is not
  a variant, it is possible to consider that $f$ is the last field to get
  populated in the swap sequence. We thus have something of the form:
  \begin{center}\footnotesize\makebox[0pt][c]{
      \AxiomC{\typedheapchan{\Theta_o}{\Gamma_1}{(h\downarrow o)\setminus o}\hskip -3em\null}
      \AxiomC{\ldots}
      \AxiomC{\ldots}
      \rulename{1}
      \UnaryInfC{\judgement\envref{\Gamma_2}{o} > v : T_v <
        \envref{\Gamma_3}{o} /}
      \rulename{T-Swap}
      \UnaryInfC{\judgement\envref{\Gamma_2}{o} > \swap f v : \nulltype <
        \envref{\changetype{\Gamma_3}{o.f}{T_v}}{o} /}
      \rulename{T-Seq}
      \BinaryInfC{\judgement\envref{\Gamma_1}{o} > \seq{\ldots}{\swap f v} :
        \nulltype < \envref{\changetype{\Gamma_3}{o.f}{T_v}}{o} /}
      \rulename{T-SubEnv}
      \UnaryInfC{\judgement\envref{\Gamma_1}{o} > \seq{\ldots}{\swap f v} :
        \nulltype < \envref{\Gamma_o}{o} /}
      \rrulename{T-Hadd}
      \BinaryInfC{\typedheapchan{\Theta_o}{\Gamma_o}{h\downarrow o}}
      \DisplayProof
    }
  \end{center}
  with $T_v\subt T$ and $\changetype{\Gamma_3}{o.f}{T_v} \subt
  \Gamma_o$. If we change the type of $o.f$, this last relation becomes
  $\Gamma_3\subt\changetype{\Gamma_o}{o.f}{\nulltype}$ \textbf{(b)}.
  
  What we want to do is to replace the judgement on the top right, which
  is an application of some rule (1), by a judgement typing $v'$. For
  this, we need the rest of the environment. We
  consider the judgement for the rest of the heap,
  \typedheapchan{\Theta\setminus\Theta_o}{\Gamma\setminus o}{h\uparrow o}. Since
  the domains are disjoint, we can apply Lemma \ref{lem:heapmerging} to
  this and the leftmost premise of \textsc{T-Hadd}, yielding
  \typedheapchan{\Theta}{\Gamma\setminus o + \Gamma_1}{h\setminus o}. If we
  replace our left premise with this, the initial environment we get on
  the top right is now $\Gamma_4 = \Gamma\setminus o + \Gamma_2$,
  as the additional
  part is unaffected by the sequence of swaps. This environment is almost
  $\changetype\Gamma{o.f}{\nulltype}$, but not quite. We now have three cases
  depending on what rule (1) is, which correspond to the three cases of
  the lemma.
  \begin{enumerate}
  \item If (1) is \textsc{T-Ref} or \textsc{T-Chan}, meaning $v$ is an
    object identifier or a channel endpoint, then $\Gamma_2 = \Gamma_3,
    v : T_v$. Using (b), this yields $\Gamma_2 \subt
    \changetype{\Gamma_o}{o.f}{\nulltype}, v : T_v$.
    Adding $\Gamma\setminus o$ to both sides, we get
    $\Gamma_4\subt\changetype{\Gamma}{o.f}{\nulltype},
    v : T_v$.
    If we replace the initial environment in (a) with
    this one, we get the $v : T_v$ back in the final environment. We
    then use Lemma \ref{lem:moreweak} to replace this initial
    environment with $\Gamma_4$, and \textsc{T-SubEnv} to change $T_v$
    into $T$ in the final one:
    \judgement\envref{\Gamma_4}{o} > v' : T' <
    \envref{\changetype{\Gamma'}{o.f}{\nulltype}, v : T}{o} /. Just see that it
    yields what we want at the bottom of the derivation.
  \item If (1) is \textsc{T-Label} or \textsc{T-Null} or
    \textsc{T-Name}, \ie if $v$ is a literal value of non-link,
    non-linear type, then $\Gamma_2$ is identical to $\Gamma_3$ and we
    have, using (b) and adding $\Gamma\setminus o$ to both sides,
    $\Gamma_4\subt\changetype{\Gamma}{o.f}{\nulltype}$, so we
    can directly (with Lemma \ref{lem:moreweak}) use judgement (a).
  \item If (1) is \textsc{T-VarS}, the last possibility, then $v$ is a
    label $l_0$ and $T_v$ is $\linktype{}f'$ for some $f'$. As it has no strict
    supertype, we have $T = \linktype{}f'$ as well. We also have
    $\Gamma_3(o.f') = \choice{l_0}{S}{}$. From (b) we have that
    $\Gamma_o(o.f') = \Gamma(o.f')$ is a supertype of this variant type, thus also a
    variant. This implies that (a) cannot come from a \textsc{T-VarS}
    concerning $f'$; therefore,
    $\Gamma'(o.f')$ is a supertype of $\Gamma(o.f')$ and hence, by
    transitivity, of $\choice{l_0}{S}{}$, which gives us the first item
    of the conclusion, with $S\subt S_{l_0}$. We now just have to notice
    that $\Gamma_2 = \changetype{\Gamma_3}{o.f'}{S}$ and that (a) is
    independent of the type of $f'$ just like it is of the type of $f$,
    and we can conclude similarly to the two previous cases.
  \end{enumerate}
\end{proof}\fi

\begin{lemma}[Substitution]
\label{lem:substitution}
If $\judgment \envref{\this:\objecttype C F, x : T'}{\this} > e : T <
\envref{\this:\objecttype C {F'}}{\this} /$, and if $\Gamma(r) =
\objecttype C F$, then:
\begin{enumerate}
\item if $T'$ is a base type (\ie neither an object type nor a link)
  and $v$ is a literal value of that type, or if $v$ is an access
  point name declared with type $\access{\Sigma}$ and
  $\sem{\access\Sigma}\subt T'$, we have:
  $$\judgment \envref\Gamma r > e\subs{v}{x} : T <
  \envref{\changetype{\Gamma}{r}{\objecttype C {F'}}}{r} /.$$ 
\item if $T'$ is an object type and $v$ is an object identifier or a
  channel endpoint, we
  have: $$\judgment \envref{\Gamma, v : T'}{r} > e\subs{v}{x} : T <
  \envref{\changetype{\Gamma}{r}{\objecttype C {F'}}}{r} /.$$
\end{enumerate}
\end{lemma}
\ifinlineproofs\begin{proof}
  In order to do an induction, we add the following case where $x$ is
  still present in the final environment : if we have $\judgment
  \envref{\this:\objecttype C F, x : T'}{\this} > e : T <
  \envref{\this:\objecttype C {F'}, x : T''}{\this} /$, and if $T'$ is
  an object type and $v$ is an object identifier, then we have
  \judgment \envref{\Gamma, v : T'}{r} > e\subs{v}{x} : T <
  \envref{\changetype{\Gamma}{r}{\objecttype C {F'}}, v : T''}{r} /.

  We prove this by induction on the derivation of
  \judgment\envref{\this:\objecttype C F, x : T'}{\this} > e : T <
  \envref{\this:\objecttype C {F'}, V}{\this} / (where $V$ is either
  empty or $v : T''$ depending on the case).  For most toplevel rules,
  the result is immediate. The only ones for which it is not are
  \textsc{T-Var} and \textsc{T-LinVar}. For \textsc{T-Var} the result
  is obtained using either \textsc{T-Null} if $T'$ is $\nulltype$ or
  \textsc{T-Label} and \textsc{T-Sub} if it is an enumerated type. In
  the case of an extension adding new base types, we assume there is a similar
  rule to type the corresponding literal values.
  For \textsc{T-LinVar}, if $v$ is an access point name the result is
  obtained using \textsc{T-Name} and \textsc{T-Sub}. Otherwise, $v$ is
  an object identifier and
  the result is obtained using \textsc{T-Ref}, noticing that because
  $\Gamma(r)$ is defined and $v$ is not in $\Gamma$, the path $r$ does
  not start with $v$ and the premise is satisfied.
\end{proof}\fi

\begin{lemma}[Typability of Subterms]
\label{lem:typablesubterms}
  If $\mathcal{D}$ is a derivation of $\judgment \envref\Gamma r >
  \mathcal{E}(e) : T < \envref{\Gamma'}{r'} /$ then there exist
  $\Gamma_1$, $r_1$ and $U$ such that $\mathcal{D}$ has a
  subderivation $\mathcal{D}'$ concluding $\judgment \envref\Gamma r>
  e : U < \envref{\Gamma_1}{r_1} /$ and the position of $\mathcal{D}'$
  in $\mathcal{D}$ corresponds to the position of the hole in
  $\mathcal{E}$.
\end{lemma}
\ifinlineproofs\begin{proof}
  A straightforward induction on the structure of $\mathcal{E}$; the
  expression $e$ is always at the extreme left of the typing
  derivation for $\mathcal{E}(e)$.
\end{proof}\fi

\begin{lemma}[Replacement]
\label{lem:replacement}
If
\begin{enumerate}
\item $\mathcal{D}$ is a derivation of $\judgment \envref\Gamma r >
  \mathcal{E}(e) : T < \envref{\Gamma'}{r'} /$
\item $\mathcal{D}'$ is a subderivation of $\mathcal{D}$ concluding
  $\judgment \envref\Gamma r > e : U < \envref{\Gamma_1}{r_1}/$
\item the position of $\mathcal{D}'$ in $\mathcal{D}$ corresponds to
  the position of the hole in $\mathcal{E}$
\item $\judgment \envref{\Gamma''}{r''} > e' : U < \envref{\Gamma_1}{r_1}/$
\end{enumerate}
then $\judgment \envref{\Gamma''}{r''} > \mathcal{E}(e') : T < \envref{\Gamma'}{r'}/$.
\end{lemma}
\ifinlineproofs\begin{proof}
Replace $\mathcal{D}'$ in $\mathcal{D}$ by the derivation of
$\judgment \envref{\Gamma''}{r''} > e' : U < \envref{\Gamma_1}{r_1}/$.
\end{proof}\fi

} 

\end{document}